\documentclass[11pt,a4paper]{article}

\usepackage{amssymb,latexsym,mathtools,amsfonts,amsthm}
\usepackage{graphicx}
\usepackage{amsbsy}
\usepackage{hyperref}
\usepackage{fullpage}
\usepackage[inline]{enumitem}
\usepackage{graphicx}
\usepackage{color}
\newcommand{\Ai}{\mathrm{Ai}}

\newcommand{\gfrac}[2]{\genfrac{}{}{0pt}{}{#1}{#2}}
\newcommand{\MeijerG}[5]{G^{#1}_{#2} \left( \left. \gfrac{#3}{#4} \right\rvert {#5} \right)}
\newcommand{\FoxH}[5]{H^{#1}_{#2} \left[ {#5} \left\rvert \gfrac{#3}{#4} \right. \right]}
\newcommand{\compC}{\mathbb{C}}
\newcommand{\realR}{\mathbb{R}}
\newcommand{\intZ}{\mathbb{Z}}
\newcommand{\halfH}[1][\theta]{\mathbb{H}_{#1}}
\newcommand{\bigO}{\mathcal{O}}
\newcommand{\g}{\widetilde{g}}

\newcommand{\G}{\widetilde{G}}
\newcommand{\Y}{\widetilde{Y}}
\newcommand{\T}{\widetilde{T}}

\newcommand{\transform}{\mathcal{T}}

\newcommand{\J}{\widetilde{J}_c}
\newcommand{\V}{\widetilde{V}}
\newcommand{\Pb}{\widetilde P^{(b)}}
\newcommand{\Vb}{\widetilde V^{(b)}}

\renewcommand{\c}{\widetilde{c}}

\newcommand{\U}{\widetilde{U}}

\renewcommand{\H}{\widetilde{H}}
\newcommand{\f}{\widetilde{f}}

\renewcommand{\P}{\widetilde{P}}

\newcommand{\C}{\widetilde{C}}
\newcommand{\R}{\widetilde{R}}
\newcommand{\Q}{\widetilde{Q}}
\newcommand{\Pinfty}{\widetilde P^{(\infty)}}
\newcommand{\s}{\widetilde{s}}
\newcommand{\n}{\widetilde{n}}
\newcommand{\m}{\widetilde{m}}
\renewcommand{\b}{\widetilde{b}}
\renewcommand{\a}{\widetilde{a}}

\newcommand{\tP}{\mathcal{P}}
\newcommand{\tR}{\mathcal{R}}

\renewcommand{\S}{\widetilde{S}}

\newcommand{\tRtilde}{\widetilde{\mathcal{R}}}
\newcommand{\SigmaR}{\pmb{\mathit{\Sigma}}}
\newcommand{\SigmaRtilde}{\widetilde{\pmb{\mathit{\Sigma}}}}
\newcommand{\Deltatilde}{\widetilde{\Delta}}

\newcommand{\natN}{\mathbb{N}}

\newcommand{\Z}{\widetilde{Z}}

\newcommand{\Dressing}{\mathcal{D}}
\newcommand{\Dressinginv}{\mathcal{E}}
\newcommand{\Dressingtilde}{\widetilde{\mathcal{D}}}
\newcommand{\Dressinginvtilde}{\widetilde{\mathcal{E}}}
\newcommand{\transJ}{\mathcal{J}}
\newcommand{\Inversion}{\mathcal{I}}
\newcommand{\Cauchy}{\mathcal{C}}
\newcommand{\Cauchytilde}{\widetilde{\mathcal{C}}}
\newcommand{\translate}{T}
\newcommand{\rb}{r^{(b)}}

\DeclareMathOperator{\supp}{supp}

\DeclareMathOperator{\dressing}{dressing}
\DeclareMathOperator{\id}{id}

\DeclareMathOperator{\pre}{pre}
\DeclareMathOperator{\model}{model}

\newtheorem{theorem}{Theorem}[section]
\newtheorem{lemma}[theorem]{Lemma}

\newtheorem{prop}[theorem]{Proposition}

\newtheorem{RHP}[theorem]{RH problem}

\theoremstyle{remark}
\newtheorem{rmk}[theorem]{Remark}

\theoremstyle{definition}
\newtheorem{defn}[theorem]{Definition}

\numberwithin{equation}{section}
\begin{document}

\title{Hard edge universality of Muttalib-Borodin ensembles with real parameter $\theta$}

\author{Dong Wang\footnotemark}
\maketitle
\renewcommand{\thefootnote}{\fnsymbol{footnote}}
\footnotetext{School of Mathematical Sciences, University of Chinese Academy of Sciences, Beijing 100049, P.~R.~China. E-mail: wangdong\symbol{'100}wangd-math.xyz}

\begin{abstract}
  We analyse the hard edge limit of the Muttalib-Borodin ensembles with general potential, and show that the limiting correlation kernel found in the ensemble with linear potential is universal. We also prove the Plancherel-Rotach type asymptotics of the biorthogonal polynomials associated to the Muttalib-Borodin ensembles around zero, where the limits are given by Wright's generalized Bessel functions. To accomplish these results, we implement the Deift-Zhou steepest-descent method on the vector Riemann-Hilbert problems for the biorthogonal polynomials, and develop a new method to construct the hard edge local parametrix at zero. The results in this paper are valid for all real parameter $\theta > 0$ in the Muttalib-Borodin ensembles, and this paper generalizes \cite{Wang-Zhang21} that considers only the integer $\theta$ case.
\end{abstract}

\section{Introduction}

\subsection{The model and the goal}

In this paper we are concerned with the particle system that has $n$ particles $x_1< \dotsb < x_n$ distributed on $[0,\infty)$, with the probability density function
\begin{equation} \label{eq:bioLag}
  \frac{1}{\mathcal{Z}_n}\prod_{i < j} (x_i - x_j) (x_i^\theta - x_j^\theta)\prod_{j=1}^n x^{\alpha}_j e^{-nV(x_j)},
\end{equation}
where $\theta > 0$ and $\alpha > -1$ are fixed parameters, $V$ is a potential function over the positive real axis and $\mathcal{Z}_n$ is the normalization constant. For the well-definedness of the particle system, we assume
\begin{equation}\label{eq:loggrowth}
  \lim_{x\to +\infty}\frac{V(x)}{\log x}= + \infty.
\end{equation}
This particle system was first introduced by Muttalib as a toy model in the studies of quasi-$1$ dimensional disordered conductors \cite{Muttalib95}, and was analysed by Borodin when the weight function is the classical Laguerre weight $V(x) = x$ \cite{Borodin99}. Hence, it is named the Muttalib-Borodin ensemble.

The Muttalib-Borodin ensemble has applications in physics \cite{Beenakker97}, \cite{Lueck-Sommers-Zirnbauer06}, \cite{Katori-Takahashi12}, and has relations to the product of random matrices \cite{Akemann-Ipsen15}, \cite{Akemann-Ipsen-Kieburg13}, \cite{Kuijlaars-Zhang14}. In mathematical literature, it has been studied in various aspects, for example, in \cite{Adler-van_Moerbeke-Wang11}, \cite{Betea-Occelli20a}, \cite{Betea-Occelli20}, \cite{Cheliotis14}, \cite{Forrester23}, \cite{Forrester-Liu14}, \cite{Forrester-Wang15}, \cite{Gautie-Le_Doussal-Majumdar-Schehr19}, \cite{Grela-Majumdar-Schehr21}, \cite{Kuijlaars-Stivigny14} and \cite{Zhang15}. In our paper, we view the Muttalib-Borodin ensemble as an archetype of biorthogonal ensembles, and prove its local universality at the hard edge $0$ via the Plancherel-Rotach type asymptotics of the biorthogonal polynomials associated to it.

In the special case that $\theta=1$, the Muttalib-Borodin ensemble becomes the classical Laguerre-type unitary invariant ensemble \cite{Anderson-Guionnet-Zeitouni10}, \cite{Forrester10}, \cite{Vanlessen07}, which features the hard-edge universality at $0$ as follows: For a large class of potential function $V$, the limiting distribution of the left-most particles converge, upon a $V$-dependent scaling factor, to a limiting distribution that is independent of $V$, and the limiting distribution is a determinantal point process defined by the Bessel kernel. This universality stems from the Plancherel-Rotach type limit around $x = 0$ of the orthogonal polynomials associated to the weight function $x^{\alpha} e^{-nV(x)}$. For any parameter $\theta > 0$, the Muttalib-Borodin ensemble is a determinantal point process \cite{Borodin99}, and for general $\theta$, the correlation kernel is expressed by the biorthogonal polynomials with respect to $x^{\alpha} e^{-nV(x)}$, which generalizes the orthogonal polynomials for $\theta = 1$. Hence, the method for $\theta = 1$ may also work for the general $\theta > 0$ case, and we expect the local limiting distribution of the left-most particles to be derived by the Plancherel-Rotach type limit around $x = 0$ of the biorthogonal polynomials. This approach has witnessed partial success:
\begin{itemize}
\item
  If $V(x) = x$, then the biorthogonal polynomials have explicit formulas, from which their Plancherel-Rotach type limits can be computed. Borodin \cite{Borodin99} solved this case, and see also \cite{Forrester-Wang15} and \cite{Zhang15}.
\item
  If $V(x)$ is in a general class of functions that satisfies the one-cut regular condition (see Section \ref{subsec:regularity} below), and $\theta$ is either an integer \cite{Wang-Zhang21} or the reciprocal of an integer \cite{Kuijlaars-Molag19,Molag20}, then the limiting distribution of the left-most particles around $0$ is proved the be the same as the $V(x) = x$ case. Hence, the universality is verified for such special $\theta$.
\end{itemize}

In this paper, we find the Plancherel-Rotach asymptotics of Muttalib-Borodin biorthogonal polynomials for any $V$ that satisfies the one-cut regular condition, and any $\theta > 0$, and prove that the universality holds for the left-most particles in the Muttalib-Borodin ensemble.

\subsection{Basic definitions of biorthogonal polynomials and determinantal point processes}

The Muttalib-Borodin ensemble is an example of biorthogonal ensembles \cite{Borodin99}, \cite{Desrosiers-Forrester08}, which are a subclass of determinantal point processes \cite{Soshnikov00}, \cite{Tracy-Widom98}. This means that there exits a correlation kernel $K_n(x,y)$ such that the density function \eqref{eq:bioLag} can be rewritten in the following determinantal form:
\begin{equation} \label{eq:kernel_intro}
  \frac{1}{n!}\det\left(K_n(x_i,x_j)\right)_{i,j=1}^n,
\end{equation}
and $K_n(x, y)$ encodes all the information of this ensemble. To find the limiting distribution of the left-most particles around $0$, we only need to compute the limit of $K_n(x, y)$ as $n \to \infty$ and $x, y \to 0$ at an appropriate speed. (See \cite{Borodin11} for general properties of determinantal point processes.)

The biorthogonal polynomials are two sequences of monic polynomials $\{p_j(x)=p^{(V)}_{n, j}(x) \}_{j=0}^{\infty}$ and $\{q_k(x)=q^{(V)}_{n, k}(x)\}_{k=0}^\infty$ that satisfy the biorthogonal conditions
\begin{equation}\label{eq:pqbioOP}
  \int^{\infty}_0 p_{j}(x) q_{k}(x^\theta) x^{\alpha} e^{-nV(x)} dx = \kappa_j\delta_{j, k},
\end{equation}
where $\kappa_j=\kappa_{n,j}^{(V)}>0$. The functions $p_n(z)$ and $q_n(z^{\theta})$ can be interpreted as the averages over the Muttalib-Borodin ensemble (\cite[Equation (1.14)]{Wang-Zhang21}). Although we do not need the probability interpretation of $p_n(z)$ and $q_n(z^{\theta})$, we note that it directly implies the existence and uniqueness of $\{p_j(x)=p^{(V)}_{n, j}(x) \}_{j=0}^{\infty}$ and $\{q_k(x)=q^{(V)}_{n, k}(x)\}_{k=0}^\infty$.

With the biorthogonal polynomials given in \eqref{eq:pqbioOP}, it is known that the correlation kernel in \eqref{eq:kernel_intro} can be written as \cite[Equation (2.11)]{Borodin99}
\begin{equation} \label{eq:sum_form_K}
  K_n(x, y) = x^{\alpha} e^{-nV(x)}%\frac{p_{j}(x)q_{j}(y^{\theta})}{\kappa_j},
  \sum^{n - 1}_{j = 0} k^{(V)}_{n, j}(x, y),
 \quad \text{where} \quad k^{(V)}_{n, j}(x, y) = \frac{p_j(x) q_j(y^{\theta})}{\kappa_j} = \frac{p^{(V)}_{n, j}(x) q^{(V)}_{n, j}(y^{\theta})}{\kappa^{(V)}_{n, j}} .
\end{equation}
Here $p_j = p^{(V)}_{n, j}$, $q_j = q^{(V)}_{n, j}$ and $\kappa_j = \kappa^{(V)}_{n, j}$ depend on $n$ and the potential function $V$. 

\subsection{The one-cut regular condition} \label{subsec:regularity}

Throughout this paper, we assume that $V(x)$ is real analytic on $[0, \infty)$. Since $V$ satisfies \eqref{eq:loggrowth}, the limiting empirical measure of the particles in \eqref{eq:bioLag} exists as $n\to \infty$, and it is the unique probability measure over $[0,+\infty)$ that minimizes the energy functional
\begin{equation} \label{eq:equilibriump}
  I^{(V)}(\nu):=\frac{1}{2}\iint \log \frac{1}{\lvert x-y \rvert} d\nu(x) d \nu(y)+\frac{1}{2}\iint \log \frac{1}{\lvert x^\theta-y^\theta \rvert} d\nu(x) d \nu(y) +\int V(x) d\nu(x);
\end{equation}
see \cite[Theorem 2.1 and Corollary 2.2]{Eichelsbacher-Sommerauer-Stolz11} and \cite[Theorem 1.2 and Corollary 1.4]{Butez17}.
Moreover, the equilibrium measure $\mu=\mu^{(V)}$ is characterized by the following Euler-Lagrange conditions:
\begin{align}
  \int \log \lvert x-y \rvert d\mu(y)+ \int \log \lvert x^\theta-y^\theta \rvert d\mu(y)-V(x)&=\ell, \quad x \in \supp(\mu), \label{eq:EL1}
  \\
  \int \log \lvert x-y \rvert d\mu(y)+ \int \log \lvert x^\theta-y^\theta \rvert d\mu(y)-V(x)&\leq \ell, \quad x \in [0,+\infty), \label{eq:EL2}
\end{align}
where $\ell$ is some real constant.

Following \cite{Kuijlaars-Molag19,Molag20}, we require the potential $V$ to be \emph{one-cut regular}, in the sense that
\begin{itemize}
\item
  the equilibrium measure $\mu$ is supported on one interval $[0, b]$ with a continuous density function $\psi=\psi^{(V)}$ for some $b=b^{(V)}>0$, that is,
  \begin{equation}\label{eq:dmu}
    d\mu(x) = \psi (x) dx,\qquad x \in (0, b);
  \end{equation}
\item
  $\psi(x) > 0$ on $(0, b)$, and there exist two positive numbers $d_1=d^{(V)}_1$ and $d_1=d^{(V)}_1$ such that
  \begin{equation}\label{eq:psiasy}
    \psi(x) = \left\{
      \begin{array}{ll}
        d_1 x^{-\frac{1}{\theta + 1}}(1 + o(1)), & \hbox{$x \to 0_+$,} \\
        d_2(b - x)^{\frac12}(1 + o(1)), & \hbox{$x\to  b_-$;}
      \end{array}
    \right.
  \end{equation}
\item
  the inequality \eqref{eq:EL2} holds strictly for $x\in (b,+\infty)$.
\end{itemize}
An explicit expression of $\psi$ is given in \cite{Claeys-Romano14}; see also \cite[Section 2.1]{Wang-Zhang21}. Given a potential $V$, it is not easy to check if it is one-cut regular. A concrete sufficient condition that implies the one-cut regularity is, by \cite{Kuijlaars-Molag19} and \cite{Molag20}\footnote{\cite[Proposition 3.6]{Kuijlaars-Molag19} claims the result for $\theta = 1/2$, and \cite{Molag20} confirms that the argument in \cite{Kuijlaars-Molag19} works for all rational $\theta > 0$. We note that the argument in \cite{Kuijlaars-Molag19} works for all real $\theta > 0$.},
\begin{equation} \label{eq:one-cut_regular}
  V''(x)x+V'(x)>0, \qquad x > 0.
\end{equation}

\subsection{Statement of main results} \label{subsec:statement_results}

To state the main theorems in our paper, we denote, with $d_1$ is defined in \eqref{eq:psiasy},
\begin{align} \label{eq:defn_rho}
  m_{\theta} = {}& \min \left( 1 + \frac{1}{\theta}, 2 \right), & c = {}& c^{(V)} = b \theta(1+\theta)^{-1-\theta^{-1}}, & \rho = {}& \rho^{(V)} = \left. \theta^{-1} d_1 \pi \middle/ \sin\left(\frac{\pi}{1 + \theta}\right), \right.
\end{align}

\begin{theorem} \label{thm:pqkappa}
  Suppose $V$ is real analytic on $[0,\infty)$ and satisfies the one-cut regular condition. As $n \to \infty$, 
 \begin{align}
    p_n \left( \frac{z}{(\rho n)^{1 + \theta^{-1}}} \right) = {}& (-1)^n C_n \left( J_{\frac{\alpha + 1}{\theta}, \frac{1}{\theta}}(\theta z) + \bigO(n^{\frac{1 - m_{\theta}}{1 + m_{\theta}}}) \right), \label{eq:result_p_n} \\
    q_n \left( \frac{z^{\theta}}{(\rho n)^{\theta + 1}} \right) = {}& (-1)^n \C_n \left( J_{\alpha + 1, \theta}((\theta z)^{\theta}) + \bigO(n^{\frac{1 - m_{\theta}}{1 + m_{\theta}}}) \right), \label{eq:result_q_n}
  \end{align}
  where
  \begin{align}
    C_n = {}& \sqrt{2\pi} c^{\frac{2(\alpha + 1) - \theta}{2(\theta + 1)}} (\rho n)^{\frac{\alpha + 1}{\theta} - \frac{1}{2}} e^{n\Re g(0)}, & \C_n = {}& \sqrt{2\pi} c^{\frac{\alpha + 1/2}{1 + \theta^{-1}}} (\theta \rho n)^{\alpha + 1/2} e^{n\Re \g(0)}.
  \end{align}
  and $J_{a_1, a_2}(x)$ is Wright's generalized Bessel function (also called Bessel-Maitland function) \footnote{Here we use the notational convention in \cite{Borodin99}, which is as remarked there, different from the original notational convention in \cite{Wright35}. The relation between these two notational conventions is $J_{a, b}(z) = \phi(a, b; -z)$. We remark that classical literature like \cite{Erdelyi-Magnus-Oberhettigner-Tricomi81a} follows \cite{Wright35}, while literature in random matrix theory like \cite{Forrester-Wang15} follows \cite{Borodin99}.}
  \begin{equation}
    J_{a_1, a_2}(x) = \sum^{\infty}_{j = 0} \frac{(-x)^j}{j! \Gamma(a_1 + ja_2)}.
  \end{equation}
\end{theorem}

\begin{lemma} \label{lem:kappa}
  With $V$ under the same condition as in Theorem \ref{thm:pqkappa}, we have
  \begin{equation} \label{eq:limit_kappa_n}
    \kappa_n = 2\pi \theta^{-1/2} c^{\alpha + 1} e^{n\ell}(1 + \bigO(n^{\frac{-m_{\theta}}{m_{\theta} + 1}}).
  \end{equation}
\end{lemma}

From the Plancherel-Rotach type asymptotics of $p_n$ and $q_n$ in Theorem \ref{thm:pqkappa}, the limit of $\kappa_n$ in \eqref{eq:limit_kappa_n}, and the summation formula \eqref{eq:sum_form_K}, we derive the following limit formula of the correlation kernel for the Muttalib-Borodin ensemble:
\begin{theorem} \label{thm:kernel}
  For real analytic $V$ that satisfies the conditions \eqref{eq:loggrowth} and \eqref{eq:one-cut_regular}, we have, for $\theta\in \mathbb{N}$,
  \begin{equation} \label{eq:thm:kernel}
    \lim_{n \to \infty} \theta^{-1} (\rho n)^{-(1 + \frac{1}{\theta})} K_n \left( \frac{x}{\theta (\rho n)^{1 + 1/\theta}}, \frac{y}{\theta (\rho n)^{1 + 1/\theta}} \right) = \theta \int^1_0 J_{\frac{\alpha + 1}{\theta}, \frac{1}{\theta}}(xu) J_{\alpha + 1, \theta}((yu)^{\theta}) u^{\alpha} du.
  \end{equation}
  uniformly for $x, y$ in compact subsets of $(0,\infty)$, where $\rho$ is given in \eqref{eq:defn_rho}.
\end{theorem}
The limit kernel on the right-hand side of \eqref{eq:thm:kernel} is the same as \cite[Equation (3.6)]{Borodin99}. It also has a double contour integral expression, see \cite[Corollary 5.2]{Forrester-Wang15} and \cite[Theorem 1.1]{Zhang15}.

 In this paper we only consider the local statistics of the Muttalib-Borodin ensemble at the hard edge $0$. We expect that $p_n(x)$ and $q_n(x^{\theta})$ have a sinusoidal limiting behaviour on $(0, b)$ and have the limit behaviour expressed by the Airy function at $b$. From these limiting results, we expect the local statistics to be given by the Airy kernel at the soft edge $b$ and by the Sine kernel in the bulk $(0, b)$. We do not prove these claims, but refer the interested reader to \cite{Claeys-Wang11} and \cite{Wang23} that prove limiting results for biorthogonal polynomials for different biorthogonal ensembles, and \cite{Claeys-Wang22} that carries out the proof of the Airy and Sine universality results for a the biorthogonal ensemble considered in \cite{Claeys-Wang11}.

\subsection{Related results}

\paragraph{Specializations of the limiting kernel $K^{(\alpha, \theta)}$}

When $\theta = 1$, the limiting kernel $K^{(\alpha, 1)}$ becomes the classical Bessel kernel \cite{Forrester93}, \cite{Tracy-Widom94b}, \cite{Vanlessen07}.

When $\theta = p/q$ is a rational number, then the limit functions $J_{\frac{\alpha + 1}{\theta}, \frac{1}{\theta}}$ and $J_{\alpha + 1, \theta}$ are Meijer G functions \cite[Equation (13)]{Gorenflo-Luchko-Mainardi00}, \cite[Equation (A.4)]{Zhang17}
\begin{align}
  J_{\frac{\alpha + 1}{\theta}, \frac{1}{\theta}}(z) = {}& (2\pi)^{\frac{q - p}{2}} p^{\frac{1}{2}} q^{-\frac{(\alpha + 1)q}{p} + \frac{1}{2}} \MeijerG{p, 0}{0, p + q}{-}{0, \frac{1}{p}, \dotsc, \frac{p - 1}{p},  \frac{1}{q} - \frac{\alpha + 1}{p}, \dotsc, 1 - \frac{\alpha + 1}{p}}{\frac{z^p}{p^p q^q}}, \\
  J_{\alpha + 1, \theta}(z) = {}& (2\pi)^{\frac{p - q}{2}} p^{-\alpha - \frac{1}{2}} q^{\frac{1}{2}} \MeijerG{q, 0}{0, p + q}{-}{0, \frac{1}{q}, \dotsc, \frac{q - 1}{q}, \frac{1}{p} - \frac{\alpha + 1}{p}, \dotsc, 1 - \frac{\alpha + 1}{p}}{\frac{z^q}{p^p q^q}}.
\end{align}
When $\theta$ or $1/\theta$ is a positive integer, $K^{(\alpha, \theta)}$ reduces to the Meijer G-kernel appearing in various random matrix models, like the products of random matrices and Cauchy two-matrix models (cf. \cite{Akemann-Strahov16}, \cite{Bertola-Bothner14}, \cite{Bertola-Gekhtman-Szmigielski14}, \cite{Kuijlaars-Zhang14}, \cite{Silva-Zhang20}) as observed in \cite{Kuijlaars-Stivigny14}.

\paragraph{Various results related to Muttalib-Borodin ensemble}

The corresponding gap probabilities for the limiting distribution defined by $K^{(\alpha, \theta)}$ are investigated in \cite{Charlier-Claeys20}, \cite{Charlier-Lenells-Mauersberger19}, \cite{Claeys-Girotti-Stivigny19}, \cite{Zhang17} and the local limits of $K_n$ away from the origin for the classical weights are given in \cite{Forrester-Wang15} and \cite{Zhang15}. The limiting mean distribution of the particles in Muttalib-Borodin ensemble is formulated as the minimizer of a (vector) equilibrium problem in \cite{Claeys-Romano14,Kuijlaars16}, and the large deviation results can be found in \cite{Bloom-Levenberg-Totik-Wielonsky17}, \cite{Butez17}, \cite{Credner-Eichelsbacher15}, \cite{Eichelsbacher-Sommerauer-Stolz11}; see also \cite{Charlier21}, \cite{Forrester23}, \cite{Lambert18}, \cite{Alam-Muttalib-Wang-Yadav20} for other investigations and extensions of the Muttalib-Borodin ensemble.

\subsection{Novelty of the paper}

Currently, some other particle systems whose joint probability density functions are in the form
\begin{equation} \label{eq:general_biorthogonal}
  \prod_{1 \leq i < j \leq n} (x_i - x_j) (f(x_i) - f(x_j)) \prod_{1 \leq i \leq n} w(x_i)
\end{equation}
have been solved, in the sense that the Plancherel-Rotach type asymptotics of the biorthogonal polynomials and the limiting local correlation kernel are computed \cite{Claeys-Wang11}, \cite{Wang23}, \cite{Claeys-Wang22}. These results do not overshadow the study of Muttalib-Borodin ensemble, because they only yield local universal limits that also occur in models only involving orthogonal polynomials, like the Airy kernel at the soft edge, Bessel kernel at the hard edge, and the sine kernel in the bulk. The hard edge  limiting correlation kernel $K^{(\alpha, \theta)}$ of the Muttalib-Borodin ensemble, to the contrary, has not been observed in models involving orthogonal polynomials. In this sense, the Muttalib-Borodin ensemble is an archetype of biorthogonal ensembles that shows limiting features not seen in orthogonal polynomial ensembles (i.e.~ensembles defined by \eqref{eq:general_biorthogonal} with $f(x) = x$).

To tackle the new limiting correlation kernel, new techniques need to be developed. It has been known that the $2 \times 2$ matrix-valued Riemann-Hilbert (RH) problem is a powerful method to study the asymptotics of orthogonal polynomials, and also the limiting correlation kernels associated to them \cite{Deift99}. Larger size matrix-valued Riemann-Hilbert problems have been introduced to solve problems involving variations of orthogonal polynomials, like multiple orthogonal polynomials (see \cite{Kuijlaars10} and references therein). For the Muttalib-Borodin ensemble, in the special case $\theta = 1/m$, $m \in \intZ_+$, the $(m + 1) \times (m + 1)$ matrix-valued Riemann-Hilbert problem was successfully applied in \cite{Kuijlaars-Molag19}, and it was suggested a similar approach may be applied for rational $\theta = p/q$, and the size of the matrix-valued Riemann-Hilbert problem, together with the technical difficulty, would increase as $p, q$ become larger. It seems to be hopeless to apply this method for irrational $\theta$.

Claeys and Romano \cite{Claeys-Romano14} found a pair of $1 \times 2$ vector-valued Riemann-Hilbert problem for the biorthogonal polynomials associated to the Muttalib-Borodin ensemble. Based on them, in \cite{Wang-Zhang21}, a new approach was introduced to find the Plancherel-Rotach asymptotics of the biorthogonal polynomials. The advantage of this approach is that the framework does not depend on the value of $\theta$, and it in principle works for all real positive $\theta$. Practically, all the arguments in \cite{Wang-Zhang21}, except for the construction of the local parametrices at $0$, are applicable for all real positive $\theta$. In \cite{Wang-Zhang21} the local parametrices at $0$ were constructed for $\theta = m$, $m \in \intZ_+$, and they were in the form of $(m + 1) \times (m + 1)$ matrices. Thus, \cite{Wang-Zhang21} localizes the Riemann-Hilbert analysis of the biorthogonal polynomials to a local parametrix construction problem, and it calls for new ideas to construct the local parametrices at $0$ for general $\theta$.

Although the matrix-valued local parametrix construction seems not to be generalizable to irrational $\theta$, if the matrix parametrices are regarded as operators, then they are finite rank operators, and it is natural to seek infinite rank operators to be the local parametrices at $0$ for general $\theta$, which degenerates to finite rank when $\theta$ is rational. \emph{Regarding the local parametrix as an operator} is a new idea in Riemann-Hilbert problems, to the author's limited knowledge. Implementing this new idea to obtain a nontrivial universality result is the main novelty of this paper. We expect this approach to find applications in other biorthogonal ensembles and hope that it helps in other Riemann-Hilbert problems.

\subsection{Organization of this paper}

In Section \ref{sec:preliminary} we set up notation and review results in \cite{Wang-Zhang21}. In Section \ref{sec:V_alpha_R} we introduce the operators to be used in the construction of the local parametrices at $0$. In Sections \ref{sec:asy_RH_Y} and \ref{sec:asy_RH_Y_tilde} we perform the transforms of the Riemann-Hilbert problems. In Section \ref{sec:proof_main} we prove the results in Section \ref{subsec:statement_results}.

\subsection*{Acknowledgements}

The author thanks Lun Zhang for the discussions at the early stage of this project. The author is partially supported by the National Natural Science Foundation of China under grant numbers 12271502 and 11871425, and the University of Chinese Academy of Sciences start-up grant 118900M043.

\section{Preliminary definitions} \label{sec:preliminary}

Throughout this paper, $\natN = \{ 0, 1, 2, \dotsc \}$. The following notations are frequently used. We denote the upper/lower half complex plane $\compC_{\pm} := \{z\in\compC \mid \pm \Im z>0\}$, open disk $D(z_0,\delta) := \{z\in \mathbb{C} : \lvert z - z_0 \rvert < \delta \}$ and punctured open disk $D^*(z_0,\delta) := D(z_0, \delta) \setminus \{ z_0 \}$. By an abuse of notation, we sometimes understand $D(0, \delta)$ as the subset $\{ 0 \} \cup \{ r e^{it} : 0 < r < \delta \text{ and } -\theta < t < \theta \}$ of $\halfH$ defined below, and $D^*(0, \delta)$ in the same way. Like \cite[Equation (1.12)]{Wang-Zhang21}, we denote
\begin{equation}\label{def:Htheta}
  \mathbb{H}_\theta := \left\{ z\in \mathbb{C} \mid z = 0 \text{ or } -\frac{\pi}{\theta} < \arg z <  \frac{\pi}{\theta} \right\},
\end{equation}
and when $\theta < 1$ we understand $\halfH$ as a sheet of Riemann surface without overlapping, rather than a subset of $\compC$, and understand $\halfH \cap \compC$ as $\left\{ z\in \mathbb{C} \mid z = 0 \text{ or } -\pi < \arg z <  \pi \right\}$. We denote the function space
\begin{equation} \label{eq:func_space_H}
  H(R) = \{ \text{analytic functions on $D(0, R)$} \}.
\end{equation}

As in \cite[Equation (2.1)]{Wang-Zhang21}, we define the analytic function
\begin{equation}\label{def:Jcs}
  J_c(s)=c(s + 1) \left( \frac{s + 1}{s} \right)^{1/\theta}, \qquad s\in\mathbb{C} \setminus [-1,0].
\end{equation}
where $c$ is defined in \eqref{eq:defn_rho}, the branch of the $1/\theta$ power function is chosen such that $J_c(s) \sim c s$ as $s \to \infty$. It has two critical points $-1$ and $s_b=1/\theta$, which are mapped respectively to $0$ and $b$, respectively.

Note that $J_{c}(s)$ is real for $s\in (-\infty,-1]\cup [s_b, +\infty)$, and there is a unique curve $\gamma_1 \subseteq \compC_+$ that connects $-1$ and $s_b$, such that $J_c(\gamma_1) = (0, b)$ and the mapping is bijective (see \cite[Proof of Lemma 4.1]{Claeys-Romano14}). Then $\gamma_2 := \overline{\gamma_1} \subseteq \compC_-$, a curve going from $-1$ to $s_b$, are mapped to the interval $[0,b]$. Let $\gamma := (-\gamma_1) \cup \gamma_2$ and denote by $D$ the region bounded by $\gamma$. We have from \cite[Lemma 4.1]{Claeys-Romano14} that $J_c$ maps $\mathbb{C} \setminus \overline{D}$ bijectively to $\mathbb{C} \setminus [0,b]$, and maps $D\setminus [-1,0]$ bijectively to $\mathbb{H}_\theta \setminus [0,b]$, where $\halfH$ is defined in \eqref{def:Htheta}. Let $I_1$ and $I_2$ be the inverses of two branches of the mapping $J_c$ satisfying
\begin{align}
  I_1(J_c(s)) = {}& s, & s \in {}& \mathbb{C}\setminus \overline{D}, \label{eq:inverse1} \\
  I_2(J_c(s)) = {}& s, & s \in {}& D \setminus [-1,0]. \label{eq:inverse2}
\end{align}

% where the branch cut of $\log z$ is taken along the negative real axis.
\subsection{Special functions used in construction of local parametrices}

For $a \in \realR$ and $z \in \compC \setminus (-\infty, 0)$, let
\begin{align}
  I^{(1)}_{\theta, a}(z) = {}& \frac{1}{2\pi i} \int_L \frac{\Gamma \left( \frac{1}{2} - a - \frac{\theta}{\theta + 1} v \right)}{\Gamma \left( 1 - a + \frac{1}{\theta + 1} v \right)} \left[ \left( \frac{\theta}{\theta + 1} \right)^{\frac{\theta}{\theta + 1}} \left( \frac{1}{\theta + 1} \right)^{\frac{1}{\theta + 1}} z \right]^v dv, \label{eq:defn_I_theta_a} \\
  I^{(2)}_{\theta, a}(z) = {}& \frac{1}{2\pi i} \int_L \Gamma \left( \frac{1}{2} - a - \frac{\theta}{\theta + 1} v \right) \Gamma \left( a - \frac{1}{\theta + 1} v \right) \left[ \left( \frac{\theta}{\theta + 1} \right)^{\frac{\theta}{\theta + 1}} \left( \frac{1}{\theta + 1} \right)^{\frac{1}{\theta + 1}} z \right]^v dv, \label{eq:defn_J_theta_a} \\
  I^{(3)}_{\theta, a}(z) = {}& \frac{1}{2\pi i} \int_L \frac{\Gamma \left( a - \frac{1}{\theta + 1} v \right)}{\Gamma \left( \frac{1}{2} + a + \frac{\theta}{\theta + 1} v \right)} \left[ \left( \frac{\theta}{\theta + 1} \right)^{\frac{\theta}{\theta + 1}} \left( \frac{1}{\theta + 1} \right)^{\frac{1}{\theta + 1}} z \right]^v dv = I^{(1)}_{\frac{1}{\theta}, \frac{1}{2} - a}(z), \label{eq:defn_K_theta_a}
\end{align}
where the contour $L$ is of Mellin-Barnes type, beginning and ending at $+\infty$, and encircling all poles exactly once in the negative direction. They are special cases of the Fox H-functions (see \cite[Chapter 1, Section 1.5, (15)]{Srivastava-Manocha84}), such that (with $u = \left( \frac{\theta}{\theta + 1} \right)^{\frac{\theta}{\theta + 1}} \left( \frac{1}{\theta + 1} \right)^{\frac{1}{\theta + 1}} z$)
\begin{align}
  I^{(1)}_{\theta, a}(z) = {}& \FoxH{1, 0}{0, 2}{-}{(\frac{1}{2} - a, \frac{\theta}{\theta + 1}), (a, \frac{1}{\theta + 1})}{u}, \\
  I^{(2)}_{\theta, a}(z) = {}& \FoxH{2, 0}{0, 2}{-}{(\frac{1}{2} - a, \frac{\theta}{\theta + 1}), (a, \frac{1}{\theta + 1})}{u}, \\
  I^{(3)}_{\theta, a}(z) = {}& \FoxH{1, 0}{0, 2}{-}{(a, \frac{1}{\theta + 1}), (\frac{1}{2} - a, \frac{\theta}{\theta + 1})}{u}.
\end{align}
$I^{(1)}_{\theta, a}$ and $I^{(3)}_{\theta, a}$ can also be expressed by Wright's generalized Bessel function
\begin{align}
  I^{(1)}_{\theta, a}(z) = {}& (1 + \theta^{-1}) u^{(1 + \theta^{-1})(\frac{1}{2} - a)} J_{\frac{\frac{1}{2} - a}{\theta} + 1 - a, \frac{1}{\theta}}(u^{1 + \theta^{-1}}), \\
  I^{(3)}_{\theta, a}(z) = {}& (1 + \theta) u^{(\theta + 1)a} J_{\theta a + \frac{1}{2} + a, \theta}(u^{\theta + 1}).
\end{align}

As $z \to \infty$, we have the following limit results. Let $\epsilon > 0$ be a fixed small constant, then
\begin{align}
  \frac{\theta^{a}}{\sqrt{2\pi(\theta + 1)}} I^{(2)}_{\theta, a}(z) = {}& e^{-z}(1 + \bigO(z^{-1})), \quad \arg z \in (-\pi + \epsilon, \pi - \epsilon), \label{eq:asy_J_theta_a_infty} \\
  \frac{\sqrt{2\pi} \theta^{a}}{\sqrt{\theta + 1}} I^{(1)}_{\theta, a}(z) = {}&
                                                                          \begin{cases}
                                                                            e^{\frac{1}{2} - a \pi i} e^{-z e^{\frac{\pi i}{\theta + 1}}} (1 + \bigO(z^{-1}), & \arg z \in (\epsilon, \frac{\pi}{1 + \theta^{-1}}), \\
                                                                            e^{-\frac{1}{2} - a \pi i} e^{-z e^{-\frac{\pi i}{\theta + 1}}} (1 + \bigO(z^{-1}), & \arg z \in (-\frac{\pi}{1 + \theta^{-1}}, -\epsilon).
                                                                          \end{cases} \label{eq:asy_I_theta_a_infty} \\
  \frac{\sqrt{2\pi} \theta^{a}}{\sqrt{\theta + 1}} I^{(3)}_{\theta, a}(z) = {}&
                                                                          \begin{cases}
                                                                            e^{a \pi i} e^{-z e^{\frac{\pi i}{1 + \theta^{-1}}}} (1 + \bigO(z^{-1}), & \arg z \in (\epsilon, \frac{\pi}{\theta + 1}), \\
                                                                            e^{-a \pi i} e^{-z e^{-\frac{\pi i}{1 + \theta^{-1}}}} (1 + \bigO(z^{-1}), & \arg z \in (-\frac{\pi}{\theta + 1}, -\epsilon).
                                                                          \end{cases} \label{eq:asy_K_theta_a_infty}
\end{align}
Here the asymptotics of $I^{(1)}_{\theta, a}(z)$ and $I^{(3)}_{\theta, a}(z)$ are from \cite[Theorem 2]{Wright35}, the asymptotics of $I^{(2)}_{\theta, a}(z)$ with $\arg z \in (-\pi/2 + \epsilon, \pi/2 - \epsilon)$ are from \cite[Equation (7.7)]{Braaksma64}, the asymptotics of $I^{(2)}_{\theta, a}(z)$ with $\arg z \in (-\pi/2 - \epsilon, -\pi/2 + \epsilon]$ are from \cite[Theorem 8]{Braaksma64}, the asymptotics of $I^{(2)}_{\theta, a}(z)$ with $\arg z \in [-\pi, -\pi/2 - \epsilon]$ are from \cite[Theorem 5]{Braaksma64}, and the remaining part is from $I^{(2)}_{\theta, a}(\bar{z}) = \overline{I^{(2)}_{\theta, a}(z)}$.
% \marginpar{See Figures \ref{fig:I_theta_a_p1} and \ref{fig:I_theta_a_p2} for $I^{(1)}_{\theta, a}(z)$; see Figures \ref{fig:J_theta_a_p1_small} \ref{fig:J_theta_a_p2_small}, \ref{fig:J_theta_a_thm_5} and \ref{fig:J_theta_a_p2_small} for $I^{(2)}_{\theta, a}(z)$.}

By writing $I^{(1)}_{\theta, a}(z), I^{(2)}_{\theta, a}(z), I^{(3)}_{\theta, a}(z)$ into power function series (see \cite[Theorem 1]{Braaksma64}), we obtain the estimate as $z \to 0$
\begin{align} \label{eq:asy_IJK_at_0}
  I^{(1)}_{\theta, a}(z) = {}& \bigO(z^{(1 + \theta^{-1})(\frac{1}{2} - a)}), & I^{(3)}_{\theta, a}(z) = {}& \bigO(z^{(1 + \theta)a}), & I^{(2)}_{\theta, a}(z) = {}&
                                                                                                                                                \begin{cases}
                                                                                                                                                  \bigO(z^{(1 + \theta^{-1})(\frac{1}{2} - a)}), & a > \frac{1}{2(\theta + 1)}, \\
                                                                                                                                                  \bigO(z^{1/2} \log z), & a = \frac{1}{2(\theta + 1)}, \\
                                                                                                                                                  \bigO(z^{(1 + \theta)a}), & a < \frac{1}{2(\theta + 1)}, 
                                                                                                                                                \end{cases}
\end{align}
and also the relation
\begin{align}
  e^{-(a - 1/2)\pi i} I^{(2)}_{\theta, a}(z e^{\frac{\pi i}{\theta + 1}}) + e^{(a - 1/2)\pi i} I^{(2)}_{\theta, a}(z e^{-\frac{\pi i}{\theta + 1}}) = {}& 2\pi I^{(1)}_{\theta, a}(z), \label{eq:G_split} \\
  e^{a\pi i} I^{(2)}_{\theta, a}(z e^{\frac{\pi i}{1 + \theta^{-1}}}) + e^{-a\pi i} I^{(2)}_{\theta, a}(z e^{-\frac{\pi i}{1 + \theta^{-1}}}) = {}& 2\pi I^{(3)}_{\theta, a}(z). \label{eq:H_split}
\end{align}

\subsection{Properties of the $g$-functions recalled from \cite{Wang-Zhang21}} \label{sec:equmeasure}

From the equilibrium measure $d\mu(x) = \psi(x)dx$, we define the $g$-functions (cf.~\cite[Equations (4.4)--(4.5)]{Claeys-Romano14}) as
\begin{align}
  g(z) := {}& \int^b_0 \log(z - y) \psi(y) dy, & z \in {}& \compC \setminus (-\infty, b], \label{def:g} \\
 \g(z) := {}& \int^b_0 \log(z^{\theta} - y^{\theta}) \psi(y) dy, & z \in {}& \mathbb{H}_{\theta} \setminus [0, b], \label{def:tildeg} \\
  \intertext{and also define \cite[Equation (3.5)]{Wang-Zhang21}}
  \phi(z) := {}& g(z) + \g(z) - V(z) - \ell, & z\in {}& (\compC \cap \halfH) \setminus (-\infty, b]. \label{def:phi}
\end{align}

\begin{prop} \cite[Proposition 2.1]{Wang-Zhang21} \label{prop:propg}
  The $g$-functions defined in \eqref{def:g} and \eqref{def:tildeg} have the following properties.
  \begin{enumerate}[label=\emph{(\alph*)}, ref=(\alph*)]
  \item \label{enu:prop:propg_a}
    $g(z)$ and $\g(z)$ are analytic in $\compC \setminus (-\infty, b]$ and $\mathbb{H}_{\theta} \setminus [0, b]$, respectively. Furthermore, as $z\to \infty$, we have
    \begin{align}
      g(z) = {}& \log z+ \bigO(z^{-1}), & z \in {}& \compC, \\
      \g(z)= {}& \theta \log z+ \bigO(z^{-\theta}), & z\in {}& \mathbb{H}_{\theta}. \label{eq:gtilde_asy}
    \end{align}
  \item
    The $g$-functions satisfy the following boundary conditions:
    \begin{align}
      \g(e^{-\frac{\pi i}{\theta}}x) & =\g(e^{\frac{\pi i}{\theta}}x)-2 \pi i, & x > {}& 0, \label{eq:tildgpm} \\
      g_+(x) & =g_-(x)+2 \pi i, & x < {}& 0. \label{eq:gpm}
    \end{align}
  \item
    For $x\in(0,b)$, we have
    \begin{equation}\label{eq:psiandg}
      \psi(x)=-\frac{1}{2\pi i}(g_+'(x)-g_-'(x))=-\frac{1}{2\pi i}(\g_+'(x)-\g_-'(x)).
    \end{equation}
  \item \label{enu:prop:propg_d}
    With the constant $\ell$ depending on $V$ (see \cite[Equations (1.6) and (1.7)]{Wang-Zhang21}, we have
    \begin{align}
      g_{\pm}(x)+\g_{\mp}(x)-V(x)-\ell = {}& 0, & x \in {}& (0,b], \label{eq:gequal}      \\
      \phi(x) = g(x)+\g(x)-V(x)-\ell < {}& 0, & x > {}& b. \label{eq:gequal2}
    \end{align}
  \item
    If $\theta>1$, we have, as $z \to 0$,
    \begin{align}
      g(z) - g_+(0) = {}&
                          \begin{cases}
                            g^{\pre}(z) + \bigO(z), & \arg z \in (0, \pi), \\
                            g^{\pre}(z) - 2\pi i + \bigO(z), & \arg z \in (-\pi, 0),
                          \end{cases} \label{eq:g_error} \\
      \g(z) - \g_+(0) = {}&
                            \begin{cases}
                              \g^{\pre}(z) + \bigO(z^{\frac{2\theta}{1+\theta}}), & \arg z \in (0, \pi/\theta), \\
                              \g^{\pre}(z) - 2\pi i + \bigO(z^{\frac{2\theta}{1+\theta}}), & \arg z \in (-\pi/\theta, 0),
                            \end{cases} \label{eq:g_tilde_error}
    \end{align}
    where $g_+(0) = \lim_{z\to 0,\ 0 < \arg z < \pi} g(z)$, $\g_+(0) = \lim_{z\to 0,\ 0 < \arg z < \pi/\theta} \g(z)$, and, with $d_1$ given in \eqref{eq:psiasy},
    \begin{align}
      g^{\pre}(z) = {}&
      \begin{cases}
        \frac{(1+\theta)d_1\pi}{\theta} \frac{e^{\frac{2+\theta}{1+\theta}\pi i}}{\sin (\frac{\pi}{1+\theta})}z^{\frac{\theta}{1+\theta}}, & \arg z \in (0,\pi),
        \\
        \frac{(1+\theta)d_1\pi}{\theta} \frac{e^{\frac{\theta}{1+\theta}\pi i}}{\sin (\frac{\pi}{1+\theta})}z^{\frac{\theta}{1+\theta}}, & \arg z \in (-\pi,0),
      \end{cases} \label{eq:asy_formula_for_g} \\
      \g^{\pre}(z) = {}&
      \begin{cases}
        \frac{(1+\theta)d_1\pi}{\theta} \frac{e^{\frac{1+2\theta}{1+\theta}\pi i}}{\sin (\frac{\pi}{1+\theta} )}z^{\frac{\theta}{1+\theta}}, & \arg z \in (0,\frac{\pi}{\theta}), \\
        \frac{(1+\theta)d_1\pi}{\theta}  \frac{e^{\frac{3+2\theta}{1+\theta}\pi i}}{\sin (\frac{\pi}{1+\theta})}z^{\frac{\theta}{1+\theta}}, & \arg z \in (-\frac{\pi}{\theta},0).
      \end{cases} \label{eq:asy_formula_for_g_tilde}
    \end{align}
  \end{enumerate}
\end{prop}
From the properties of $g$ and $\g$, we have 
\begin{equation} \label{eq:ineq:psi}
  \Re \phi_+(x) = \Re \phi_-(x) = 0, \quad -\Im \phi'_+(x) = \Im \phi'_-(x) = 2\pi \psi(x) > 0, \quad \text{for all $x \in (0, b)$}.
\end{equation}

Let $\rb$ be a small positive constant, and define \cite[Equation (3.35)]{Wang-Zhang21}
\begin{equation}\label{def:fb}
  f_b(z):=\left(-\frac34 \phi(z)\right)^{\frac23}
\end{equation}
for $z \in D(b, \rb)$. Here $f_b(z)$ is a conformal mapping satisfying $f_b(b)=0$ and $f_b'(b)>0$. 

\subsection{RH characterization of the biorthogonal polynomials and auxiliary functions}\label{sec:RHP}

The proofs of our results rely on $1 \times 2$ vector-valued RH problems characterizing the biorthogonal polynomials $p_j$  and $q_k$ in \eqref{eq:pqbioOP}, as shown in \cite[Section 3]{Claeys-Romano14} and \cite[Section 2.2]{Wang-Zhang21}. Let
\begin{equation} \label{def:Y}
  Y (z) = (p_n(z), Cp_n(z)),
\end{equation}
where for any polynomial $p$,
\begin{equation} \label{eq:defn_of_Cp_n}
  Cp(z) = \frac{1}{2 \pi i}\int^{\infty}_0 \frac{p(x)}{x^\theta - z^\theta}w(x) dx=\frac{1}{2 \pi i}\int^{\infty}_0 \frac{p(x)}{x^\theta - z^\theta}x^{\alpha}e^{-nV(x)} dx
\end{equation}
is the modified Cauchy transform of the polynomial $p$. We will consider $Cp_n(z)$ as a function defined in $\mathbb{H}_{\theta} \setminus [0, +\infty)$. By \cite[Theorem 1.4]{Claeys-Romano14}, $Y$ is the unique solution of the following RH problem (see also the proof of Proposition \ref{prop_uniqueness}).

\begin{RHP} \cite[RH problem 2.2]{Wang-Zhang21} \label{RHP:original_p}
\hfill
  \begin{enumerate}[label=\emph{(\arabic*)}, ref=(\arabic*)]
  \item \label{enu:RHP:original_p:1}
    $Y = (Y_1, Y_2)$ is analytic in $(\compC, \mathbb{H}_{\theta} \setminus [0, +\infty))$.
  \item \label{enu:RHP:original_p:2}
    $Y$ has continuous boundary values $Y_{\pm}$ when approaching $(0, +\infty)$ from above $(+)$ and below $(-)$, and satisfies
    \begin{equation} \label{eq:jump_Y_realR}
      Y_+(x) = Y_-(x)
      \begin{pmatrix}
        1 & \frac{1}{\theta x^{\theta-1}} x^{\alpha} e^{-nV(x)} \\
        0 & 1
      \end{pmatrix},
      \qquad x>0.
    \end{equation}
  \item \label{enu:RHP:original_p:3}
    As $z \to \infty$ in $\compC$, $Y_1$ behaves as $Y_1(z) = z^n + \bigO(z^{n - 1})$.
  \item \label{enu:RHP:original_p:4}
    As $z \to \infty$ in $\mathbb{H}_{\theta}$, $Y_2$ behaves as $Y_2(z) = \bigO(z^{-(n + 1)\theta})$.
  \item \label{enu:RHP:original_p:5}
    As $z \to 0$ in $\compC$, we have
    \begin{equation}
    Y_1(z)=\bigO(1).
    \end{equation}
  \item \label{enu:RHP:original_p:6}
    As $z \to 0$ in $\mathbb{H}_{\theta}$,
    we have
    \begin{equation} \label{eq:Y_2_limit_at_0}
      Y_2(z) =
      \begin{cases}
        \bigO(1), & \text{$\alpha+1-\theta > 0$,} \\
         \bigO(\log z), & \text{$\alpha+1-\theta = 0$,} \\
       \bigO(z^{\alpha+1-\theta}), & \text{$\alpha+1-\theta < 0$.}
      \end{cases}
    \end{equation}
  \item \label{enu:RHP:original_p:7}
    $Y_2(z)$ has continuous boundary values $Y_2(e^{\frac{\pi i}{\theta} }x)$ and $Y_2(e^{-\frac{\pi i}{\theta}}x)$ for $x > 0$, and it satisfies the boundary condition $Y_2(e^{\frac{\pi i}{\theta} }x) = Y_2(e^{-\frac{\pi i}{\theta}}x)$.
  \end{enumerate}
\end{RHP}

The polynomials $q_j$ are also characterized by a similar RH problem. By setting
\begin{equation}\label{eq:defn_of_Cq_n}
\widetilde C q_n(z):=\frac{1}{2\pi i}\int_0^{+\infty} \frac{q_n(x^{\theta})}{x-z}w(x)dx
=\frac{1}{2\pi i}\int_0^{+\infty} \frac{q_n(x^{\theta})}{x-z}x^{\alpha}e^{-nV(x)}dx,\quad z\in \compC \setminus [0,+\infty),
\end{equation}
we have that
\begin{equation} \label{def:Ytilde}
\Y(z) = (q_n(z^{\theta}), \widetilde C q_n(z))
\end{equation}
is the unique solution of the following RH problem; see \cite[Theorem 1.4]{Claeys-Romano14}.

\begin{RHP} \cite[RH problem 2.3]{Wang-Zhang21} \label{RHP:original_q} \hfill
  \begin{enumerate}[label=\emph{(\arabic*)}, ref=(\arabic*)]
  \item
    $\Y = (\Y_1, \Y_2)$ is analytic in $(\mathbb{H}_{\theta}, \compC \setminus [0, +\infty))$.
  \item $\Y$ has continuous boundary values $\Y_{\pm}$ when approaching $(0, +\infty)$ from above $(+)$ and below $(-)$, and satisfies
    \begin{equation}
      \Y_+(x) =\Y_-(x)
      \begin{pmatrix}
        1 & x^{\alpha} e^{-nV(x)} \\
        0 & 1
      \end{pmatrix},
      \qquad x>0.
    \end{equation}
  \item
    As $z \to \infty$ in $\mathbb{H}_{\theta}$, $\Y_1$ behaves as $\Y_1(z) = z^{n\theta} + \bigO(z^{(n - 1)\theta})$.
  \item \label{enu:RHP:original_q:4}
    As $z \to \infty$ in $\compC$, $\Y_2$ behaves as $\Y_2(z) = \bigO(z^{-(n + 1)})$.
  \item
    As $z \to 0$ in $\mathbb{H}_{\theta}$, we have
    \begin{equation}
    \Y_1(z)=\bigO(1).
    \end{equation}
   \item
    As $z \to 0$ in $\compC$, we have
    \begin{equation}
    \Y_2(z) = \left\{
               \begin{array}{ll}
                 \bigO(1), & \hbox{ $\alpha > 0$,} \\
                 \bigO(\log z), & \hbox{ $\alpha = 0$,} \\
                 \bigO(z^{\alpha}), & \hbox{ $\alpha < 0$.}
               \end{array}
             \right.
        \end{equation}

  \item
    $\Y_1(z)$ has continuous boundary values $\Y_1(e^{\frac{\pi i}{\theta} }x)$ and $\Y_1(e^{-\frac{\pi i}{\theta}}x)$ for $x > 0$, and it satisfies the boundary condition $\Y_1(e^{\frac{\pi i}{\theta} }x) = \Y_1(e^{-\frac{\pi i}{\theta}}x)$.
  \end{enumerate}
\end{RHP}

We remark that with the constant $\kappa_n$ defined in \eqref{eq:pqbioOP}, Item \ref{enu:RHP:original_p:4} of RH problem \ref{RHP:original_p} and Item \ref{enu:RHP:original_q:4} of RH problem \ref{RHP:original_q} can be refined to
\begin{align} \label{eq:Y_2_asy_infty}
  Y_2(z) = {}& \frac{\kappa_n z^{-(n + 1)\theta}}{2\pi i} (1 + \bigO(z^{-\theta})), &
  \Y_2(z) = {}& \frac{\kappa_n z^{-(n + 1)}}{2\pi i} (1 + \bigO(z^{-1})), && \text{as $z \to \infty$ in $\halfH$ or $\compC$}.
\end{align}

The transforms of RH problems \ref{RHP:original_p} and \ref{RHP:original_q} require the global parametrices and Airy parametrix. The former is constructed in \cite{Wang-Zhang21} and the latter is commonly used in Riemann-Hilbert literature. We present the results here for the reader's convenience.

\paragraph{Global parametrices}

With $I_1, I_2$ defined in \eqref{eq:inverse1} and \eqref{eq:inverse2}, we define \cite[Equations (3.19), (3.20), (4.15) and (4.16)]{Wang-Zhang21}
\begin{align}
  P_1^{(\infty)}(z) = {}& \tP(I_1(z)), & \widetilde P_2^{(\infty)}(z)&=\widetilde \tP(I_1(z)), & z \in {}& \mathbb{C}\setminus [0,b], \label{eq:P1} \\
  P_2^{(\infty)}(z) = {}& \tP(I_2(z)), & \widetilde P_1^{(\infty)}(z)&=\widetilde \tP(I_2(z)), & z \in {}& \mathbb{H}_\theta \setminus [0,b], \label{eq:P2}
\end{align}
where $\tP$ and $\widetilde \tP$ a given by \cite[Equations (3.18) and (4.14)]{Wang-Zhang21}
\begin{align} \label{eq:Fs}
  \tP(s)= {}&
              \begin{cases}
                \frac{s}{\sqrt{(s + 1)(s - s_b)}}  \left( \frac{s+1}{s} \right)^{\frac{\theta-\alpha-1}{\theta}}, & s \in \compC \setminus \overline{D}, \\
                \frac{c^{\alpha+1-\theta}s(s+1)^{\alpha+1-\theta}}{\theta\sqrt{(s+1)(s-s_b)}}, & s\in D,
              \end{cases}
  &
    \widetilde \tP(s)= {}&
                           \begin{cases}
                             \frac{c^{\alpha}\sqrt{s_b}i}{\sqrt{(s + 1)(s - s_b)}} \left( \frac{s+1}{s} \right)^{\frac{\alpha}{\theta}}, & \text{$s \in \compC \setminus \overline{D}$,} \\
                             \frac{\sqrt{s_b}i }{\sqrt{(s + 1)(s - s_b)}} (s+1)^{-\alpha}, & \text{$s\in D$,}
                           \end{cases}
\end{align}
where $s_b=1/\theta$, the branch cuts of $\sqrt{(s+1)(s - s_b)}$, $\left(\frac{s+1}{s}\right)^{\frac{\theta-\alpha-1}{\theta}}$ and $(s+1)^{\alpha+1-\theta}$ for $\tP$ are taken along $\gamma_1$, $[-1,0]$ and $(-\infty,-1]$, respectively, and the branch cuts of $\sqrt{(s+1)(s - s_b)}$, $\left(\frac{s+1}{s}\right)^{\frac{\alpha}{\theta}}$ and $(s+1)^{-\alpha}$ for $\widetilde{\tP}$ are taken along $\gamma_2$, $[-1,0]$ and $(-\infty,-1]$, respectively. 

We let
\begin{align}
  P^{(\infty)} = {}&(P^{(\infty)}_1, P^{(\infty)}_2), & \P^{(\infty)} = (\P^{(\infty)}_1, \P^{(\infty)}_2).
\end{align}
Then we have for $x \in (0,b)$ \cite[Equations (3.14) and (4.10)]{Wang-Zhang21},
\begin{align} \label{eq:P^infty_jump}
  P^{(\infty)}_+(x) = {}& P^{(\infty)}_-(x)
  \begin{pmatrix}
    0 & \frac{x^{\alpha+1-\theta}}{\theta} \\
    -\frac{\theta}{x^{\alpha+1-\theta}} & 0
  \end{pmatrix},
  & \Pinfty_+(x) = {}& \Pinfty_-(x)
                       \begin{pmatrix}
                         0 & x^{\alpha}
                         \\
                         -x^{-\alpha} & 0
                       \end{pmatrix},
\end{align}
and by \cite[Proposition 3.6 and 4.6]{Wang-Zhang21}, we have as $z \to b$,
\begin{align} \label{eq:Pinfty_asy_b}
  P^{(\infty)}_i(z)) = {}& \bigO((z - b)^{-1/4}), & \Pinfty_i(z)) = {}& \bigO((z - b)^{-1/4}), & i = {}& 1, 2,
\end{align}
and as $z\to 0$, 
\begin{align}
  P^{(\infty)}_1(z) = {}& P^{(\infty), \pre}_1(z) (1 + \bigO(z^{\frac{\theta}{1+\theta}})), & P^{(\infty)}_2(z) = P^{(\infty), \pre}_2(z)  (1 + \bigO(z^{\frac{\theta}{1+\theta}})), \label{eq:P_and_P^pre} \\
  \Pinfty_1(z) = {}& \P^{(\infty), \pre}_1(z) (1 + \bigO(z^{\frac{\theta}{1+\theta}})), & \Pinfty_2(z) = \P^{(\infty), \pre}_2(z)  (1 + \bigO(z^{\frac{\theta}{1+\theta}})), \label{eq:P_and_P^pre_tilde}
\end{align}
where
\begin{align}
  P^{(\infty), \pre}_1(z) = {}&
                          \begin{cases}
                            \sqrt{\frac{\theta}{1+\theta}}c^{\frac{2(\alpha+1)-\theta}{2(1+\theta)}}e^{\frac{2(\alpha+1)-\theta}{2(1+\theta)}\pi i}z^{\frac{\theta-2(\alpha+1)}{2(1+\theta)}}, & \arg z \in (0, \pi),
                            \\
                            \sqrt{\frac{\theta}{1+\theta}}c^{\frac{2(\alpha+1)-\theta}{2(1+\theta)}}e^{\frac{\theta-2(\alpha+1)}{2(1+\theta)}\pi i}z^{\frac{\theta-2(\alpha+1)}{2(1+\theta)}}, & \arg z \in (-\pi, 0),
                          \end{cases} \label{eq:asy_P^infty_1} \\
  P^{(\infty), \pre}_2(z) = {}&
                          \begin{cases}
                            \frac{c^{\frac{2(\alpha+1)-\theta}{2(1+\theta)}}}{\sqrt{\theta(1+\theta)}}e^{\frac{\theta-2(\alpha+1)}{2(1+\theta)}\pi i} z^{\frac{(\alpha+\frac12 -\theta)\theta}{1+\theta}}, & \arg z \in (0, \frac{\pi}{\theta}), \\
                            \frac{c^{\frac{2(\alpha+1)-\theta}{2(1+\theta)}}}{\sqrt{\theta(1+\theta)}}e^{\frac{2\alpha-3\theta}{2(1+\theta)}\pi i} z^{\frac{(\alpha+\frac12 -\theta)\theta}{1+\theta}}, & \arg z \in (-\frac{\pi}{\theta}, 0),
                          \end{cases} \label{eq:asy_P^infty_2} \\
  \P^{(\infty), \pre}_1(z) = {}& \left\{
                             \begin{array}{ll}
                               \frac{c^{\frac{(\alpha+1/2)\theta}{1+\theta}}}{\sqrt{1+\theta}}e^{\frac{\alpha+1/2}{1+\theta}\pi i}z^{-\frac{(\alpha+1/2)\theta}{1+\theta}}, & \arg z \in (0, \frac{\pi}{\theta}),
                               \\
                               \frac{c^{\frac{(\alpha+1/2)\theta}{1+\theta}}}{\sqrt{1+\theta}}e^{-\frac{\alpha+1/2}{1+\theta}\pi i}z^{-\frac{(\alpha+1/2)\theta}{1+\theta}}, & \arg z \in (-\frac{\pi}{\theta}, 0),
                             \end{array}
                             \right. \label{eq:asy0tildeP1infty} \\
  \P^{(\infty), \pre}_2(z) = {}&
                             \begin{cases}
                               \frac{c^{\frac{(\alpha+1/2)\theta}{1+\theta}}}{\sqrt{1+\theta}}e^{-\frac{\alpha+1/2}{1+\theta}\pi i} z^{\frac{2\alpha -\theta}{2(1+\theta)}}, & \arg z \in (0, \pi), \\
                               -\frac{c^{\frac{(\alpha+1/2)\theta}{1+\theta}}}{\sqrt{1+\theta}}e^{\frac{\alpha+1/2}{1+\theta}\pi i} z^{\frac{2\alpha -\theta}{2(1+\theta)}}, & \arg z \in (-\pi, 0),
                             \end{cases}
                             \label{eq:asy0tildeP1infty2}
\end{align}

\paragraph{Airy parametrix}

The Airy parametrix $\Psi^{(\Ai)}$ is the following a $2\times 2$ matrix-valued function
\begin{align} \label{eq:defn_Psi_Airy}
  \Psi^{(\Ai)}(\zeta) = {}&
                            \begin{cases}
                              \begin{pmatrix}
                                y_0(\zeta) &  -y_2(\zeta) \\
                                y_0'(\zeta) & -y_2'(\zeta)
                              \end{pmatrix}, & \text{$\arg \zeta \in (0,\frac{2\pi}{3})$,} \\
                              \begin{pmatrix}
                                -y_1(\zeta) &  -y_2(\zeta) \\
                                -y_1'(\zeta) & -y_2'(\zeta)
                              \end{pmatrix}, & \text{$\arg \zeta \in (\frac{2\pi}{3},\pi)$,} \\
                              \begin{pmatrix}
                                -y_2(\zeta) &  y_1(\zeta) \\
                                -y_2'(\zeta) & y_1'(\zeta)
                              \end{pmatrix}, & \text{$\arg \zeta \in (-\pi,-\frac{2\pi}{3})$,} \\
                              \begin{pmatrix}
                                y_0(\zeta) &  y_1(\zeta) \\
                                y_0'(\zeta) & y_1'(\zeta)
                              \end{pmatrix}, & \text{$\arg \zeta \in  (-\frac{2\pi}{3},0)$,}
                            \end{cases}
  &&
     \begin{aligned}
       & \text{with} \\
         y_0(\zeta) = {}& \sqrt{2\pi}e^{-\frac{\pi i}{4}} \Ai(\zeta), \\
       y_1(\zeta) = {}& \sqrt{2\pi}e^{-\frac{\pi i}{4}} \omega\Ai(\omega \zeta), \\
       y_2(\zeta) = {}& \sqrt{2\pi}e^{-\frac{\pi i}{4}} \omega^2\Ai(\omega^2 \zeta),                    
     \end{aligned}
\end{align}
where $\Ai$ is the usual Airy function (cf. \cite[Chapter 9]{Boisvert-Clark-Lozier-Olver10}) and $\omega=e^{2\pi i/3}$. See, for example, \cite[Appendix B]{Wang-Zhang21} for properties of $\Psi^{(\Ai)}$.

\section{Function spaces $V_{\alpha}(R)$ and $\V_{\alpha}(R)$} \label{sec:V_alpha_R}

\subsection{Definition and series representation of $V_{\alpha}(R)$ and $\V_{\alpha}(R)$}

Let $R$ be a positive constant or $+\infty$, and $\gamma$ be a small positive constant. We define the function spaces $V_{\alpha}(R)$ and $\V_{\alpha}(R)$ as follows.
\begin{defn} \label{defn:V_alpha}
  $V_{\alpha}(R)$ consists of functions $f(z)$ on $z \in D^*(0, R) \setminus \{ \arg z = \pm \frac{\theta^{-1} \pi + \gamma}{1 + \theta^{-1}} \}$, such that
  \begin{enumerate}
  \item
    $f(z)$ is analytic in the sector $\arg (-z) \in (\frac{-\pi + \gamma}{1 + \theta^{-1}}, \frac{\pi - \gamma}{1 + \theta^{-1}})$ and the sector $\arg z \in (\frac{-\theta^{-1} \pi - \gamma}{1 + \theta^{-1}}, \frac{\theta^{-1} \pi + \gamma}{1 + \theta^{-1}})$ separately, and $f(z)$ is continuous up to the boundary on the two rays $\{ \arg z = \pm \frac{\theta^{-1} \pi + \gamma}{1 + \theta^{-1}} \}$.
  \item
    Let the two rays be oriented from $0$ to $\infty$. The boundary values of $f$ on the sides of the two rays satisfy
    \begin{equation} \label{eq:jump_V_alpha(R)}
      f _+(z) - f_-(z) =
      \begin{cases}
        -e^{-\frac{2\alpha + 3}{\theta + 1} \pi i} f(z e^{-\frac{2\pi}{\theta + 1} i}), & \arg z = \frac{\theta^{-1} \pi + \gamma}{1 + \theta^{-1}}, \\
        e^{\frac{2\alpha + 3}{\theta + 1} \pi i} f(z e^{\frac{2\pi}{\theta + 1} i}), & \arg z = \frac{-\theta^{-1} \pi - \gamma}{1 + \theta^{-1}}.
      \end{cases}
    \end{equation}
  \item
    As $z \to 0$, $f(z)$ has the limit behaviour depending on $\alpha$ and $\theta$ as follows:
    \begin{enumerate}[label=\emph{(\alph*)}, ref=(\arabic*)]
    \item If $\alpha > \theta - 1$, then
      \begin{equation} \label{eq:V_alpha(R)_at_0:1}
        f(z) =
        \begin{cases}
          \bigO (z^{-1/2 + (\alpha + 1)/\theta}), & \arg (-z) \in (\frac{-\pi + \gamma}{1 + \theta^{-1}}, \frac{\pi - \gamma}{1 + \theta^{-1}}), \\
          \bigO (z^{\theta - \alpha - 1/2}), & \arg z \in (\frac{-\theta^{-1} \pi - \gamma}{1 + \theta^{-1}}, \frac{\theta^{-1} \pi + \gamma}{1 + \theta^{-1}}).
        \end{cases}
      \end{equation}
    \item 
      If $\alpha = \theta - 1$, then
      \begin{equation} \label{eq:V_alpha(R)_at_0:2}
        f(z) =
        \begin{cases}
          \bigO(z^{1/2}), & \arg (-z) \in (\frac{-\pi + \gamma}{1 + \theta^{-1}}, \frac{\pi - \gamma}{1 + \theta^{-1}}), \\
          \bigO (z^{1/2} \log z), & \arg z \in (\frac{-\theta^{-1} \pi - \gamma}{1 + \theta^{-1}}, \frac{\theta^{-1} \pi + \gamma}{1 + \theta^{-1}}).
        \end{cases}
      \end{equation}
    \item
      If $-1 < \alpha < \theta - 1$, then in both sectors
      \begin{equation} \label{eq:V_alpha(R)_at_0:3}
        f(z) = \bigO (z^{-1/2 + (\alpha + 1)/\theta}).
      \end{equation}
    \end{enumerate}
  \end{enumerate}
\end{defn}

\begin{defn} \label{defn:V_alpha_tilde}
  $\V_{\alpha}(R)$ consists of functions $\f(z)$ on $z \in D^*(0, R) \setminus \{ \arg z = \pm \frac{\theta^{-1} \pi - \gamma}{1 + \theta^{-1}} \}$, such that
  \begin{enumerate}
  \item
    $\f(z)$ is analytic in the sector $\arg (-z) \in (\frac{-\pi + \gamma}{1 + \theta^{-1}}, \frac{\pi + \gamma}{1 + \theta^{-1}})$ and the sector $\arg z \in (\frac{-\theta^{-1} \pi + \gamma}{1 + \theta^{-1}}, \frac{\theta^{-1} \pi - \gamma}{1 + \theta^{-1}})$ separately, and $\f(z)$ is continuous up to the boundary on the two rays $\{ \arg z = \pm \frac{\theta^{-1} \pi - \gamma}{1 + \theta^{-1}} \}$.
  \item
    Let the two rays be oriented from $0$ to $\infty$. The boundary values of $\f$ on the sides of the two rays satisfy
    \begin{equation} \label{eq:jump_V_alpha(R)_tilde}
      \f _+(z) - \f_-(z) =
      \begin{cases}
        -e^{\frac{2\alpha + 1}{\theta + 1} \pi i} \f(z e^{-\frac{2\pi}{\theta + 1} i}), & \arg z = \frac{\theta^{-1} \pi - \gamma}{1 + \theta^{-1}}, \\
        e^{-\frac{2\alpha + 1}{\theta + 1} \pi i} \f(z e^{\frac{2\pi}{\theta + 1} i}), & \arg z = \frac{-\theta^{-1} \pi + \gamma}{1 + \theta^{-1}}.
      \end{cases}
    \end{equation}
  \item
    As $z \to 0$, $\f(z)$ has the limit behaviour depending on $\alpha$ and $\theta$ as follows:
    \begin{enumerate}[label=\emph{(\alph*)}, ref=(\arabic*)]
    \item
      If $\alpha > 0$, then
      \begin{equation} \label{eq:V_alpha(R)_at_0:1_tilde}
        \f(z) =
        \begin{cases}
          \bigO (z^{1/2 - \alpha/\theta}), & \arg (-z) \in (\frac{-\pi - \gamma}{1 + \theta^{-1}}, \frac{\pi + \gamma}{1 + \theta^{-1}}), \\
          \bigO (z^{\alpha + 1/2}), & \arg z \in (\frac{-\theta^{-1} \pi + \gamma}{1 + \theta^{-1}}, \frac{\theta^{-1} \pi - \gamma}{1 + \theta^{-1}}).
        \end{cases}
      \end{equation}
    \item 
      If $\alpha = 0$, then
      \begin{equation} \label{eq:V_alpha(R)_at_0:2_tilde}
        \f(z) =
        \begin{cases}
          \bigO(z^{1/2} \log z), & \arg (-z) \in (\frac{-\pi - \gamma}{1 + \theta^{-1}}, \frac{\pi + \gamma}{1 + \theta^{-1}}), \\
          \bigO (z^{1/2}), & \arg z \in (\frac{-\theta^{-1} \pi + \gamma}{1 + \theta^{-1}}, \frac{\theta^{-1} \pi - \gamma}{1 + \theta^{-1}}).
        \end{cases}
      \end{equation}
    \item
      If $-1 < \alpha < 0$, then in both sectors
      \begin{equation} \label{eq:V_alpha(R)_at_0:3_tilde}
        \f(z) = \bigO (z^{\alpha + 1/2}).
      \end{equation}
    \end{enumerate}
  \end{enumerate}
\end{defn}

% \subsection{Power series representation of $V_{\alpha}(R)$ and $\V_{\alpha}(R)$}

Let $f(z) \in V_{\alpha}(R)$. We denote the branch of $f(z)$ in the sector $\arg (-z) \in (\frac{-\pi + \gamma}{1 + \theta^{-1}}, \frac{\pi - \gamma}{1 + \theta^{-1}})$ as $f^{(1)}(z)$, and the branch of $f(z)$ in the sector $\arg z \in (\frac{-\theta^{-1} \pi - \gamma}{1 + \theta^{-1}}, \frac{\theta^{-1} \pi + \gamma}{1 + \theta^{-1}})$ as $f^{(R)}(z)$. In the sector $\arg z \in (\frac{\theta^{-1} \pi - \gamma}{1 + \theta^{-1}}, \frac{\theta^{-1} \pi + \gamma}{1 + \theta^{-1}})$, the function
\begin{equation} \label{eq:f^(1)_upper_ext}
  f^{(R)}(z) - e^{-\frac{2\alpha + 3}{\theta + 1} \pi i} f^{(R)}(z e^{-\frac{2\pi}{\theta + 1} i})
\end{equation}
is analytic and continuous on the two boundary rays, and by \eqref{eq:jump_V_alpha(R)} its boundary value along the $-$ side of ray $\{ \arg z = \frac{\theta^{-1} \pi + \gamma}{1 + \theta^{-1}} \}$ is equal to that of $f^{(1)}(z)$ on the $+$ side of the ray. Similarly, in the sector $\arg z \in (\frac{-\theta^{-1} \pi - \gamma}{1 + \theta^{-1}}, \frac{-\theta^{-1} \pi + \gamma}{1 + \theta^{-1}})$, the function
\begin{equation} \label{eq:f^(1)_lower_ext}
  f^{(R)}(z) - e^{\frac{2\alpha + 3}{\theta + 1} \pi i} f^{(R)}(z e^{\frac{2\pi}{\theta + 1} i})
\end{equation}
is analytic and continuous on the two boundary rays, and by \eqref{eq:jump_V_alpha(R)} its boundary value along the $+$ side of ray $\{ \arg z = \frac{-\theta^{-1} \pi - \gamma}{1 + \theta^{-1}} \}$ is equal to that of $f^{(1)}(z)$ on the $-$ side of the ray. Then we see that $f^{(1)}(s)$ can be analytically extended to the sector $\arg (-z) \in (\frac{-\pi - \gamma}{1 + \theta^{-1}}, \frac{\pi + \gamma}{1 + \theta^{-1}})$ by \eqref{eq:f^(1)_upper_ext} and \eqref{eq:f^(1)_lower_ext}. It is clear that for $\arg z \in (\frac{\theta^{-1} \pi - \gamma}{1 + \theta^{-1}}, \frac{\theta^{-1} \pi + \gamma}{1 + \theta^{-1}})$, $f^{(1)}(z)$ satisfies
\begin{equation}
  f^{(1)}(z) = -e^{-\frac{2\alpha + 3}{\theta + 1} \pi i} f^{(1)}(z e^{-\frac{2\pi}{\theta + 1} i}).
\end{equation}
Hence, we find that the function
\begin{equation}
  (-z)^{\frac{1}{2} - \frac{\alpha + 3/2}{\theta + 1}} f^{(1)}((-z)^{\frac{\theta}{\theta + 1}}), \quad \arg (-z) \in (-\pi, \pi),
\end{equation}
can be extended to a holomorphic function on $D^*(0, R^{1 + \theta^{-1}})$, and by \eqref{eq:V_alpha(R)_at_0:1}, \eqref{eq:V_alpha(R)_at_0:2} and \eqref{eq:V_alpha(R)_at_0:3}, as $z \to 0$ in the sector $\arg (-z) \in (-\pi + \gamma, \pi - \gamma)$
\begin{equation}
  (-z)^{\frac{1}{2} - \frac{\alpha + 3/2}{\theta + 1}} f^{(1)}((-z)^{\frac{\theta}{\theta + 1}}) = \bigO(1),
\end{equation}
and it has at most polynomial growth in $z^{-1}$ from any direction. We conclude that $f^{(1)}(z)$ has the power series representation
\begin{equation} \label{eq:f^(1)_in_series}
  f^{(1)}(z) = (-z)^{-\frac{1}{2} + \frac{\alpha + 1}{\theta}} \sum^{\infty}_{k = 0} a_k (-z)^{\frac{\theta + 1}{\theta} k}, \quad \limsup_{k \to \infty} \lvert a_k \rvert^{\frac{1}{k}} \leq R^{-\frac{\theta + 1}{\theta}}.
\end{equation}
Next, with $a_k$ given in \eqref{eq:f^(1)_in_series}, we consider the function ($\log z$ takes the principal branch)
\begin{equation} \label{eq:f^R_f^2}
  f^{(2)}(z)  = f^{(R)}(z) - z^{-\frac{1}{2} + \frac{\alpha + 1}{\theta}} \sum^{\infty}_{k = 0} (-1)^k a_k z^{\frac{\theta + 1}{\theta} k}
  \times
  \begin{cases}
    \frac{1}{2\sin \pi \left( \frac{k + \alpha + 1}{\theta} \right)}, & \frac{k + \alpha + 1}{\theta} \notin \intZ, \\
    \frac{\theta + 1}{2\pi} (-1)^{\frac{k + \alpha + 1}{\theta}} \log z, & \frac{k + \alpha + 1}{\theta} \in \intZ.
  \end{cases}
\end{equation}
We see that $f^{(2)}(z)$ is analytic in the sector $\arg z \in (\frac{-\theta^{-1} \pi - \gamma}{1 + \theta^{-1}}, \frac{\theta^{-1} \pi + \gamma}{1 + \theta^{-1}})$ and continuous on the two boundary rays. $f^{(2)}(z)$ satisfies
\begin{equation}
  f^{(2)}(z) = e^{-\frac{2\alpha + 3}{\theta + 1} \pi i} f^{(2)}(z e^{-\frac{2\pi}{1 + \theta} i}), \quad \arg z \in (\frac{\theta^{-1} \pi - \gamma}{1 + \theta^{-1}}, \frac{\theta^{-1} \pi + \gamma}{1 + \theta^{-1}}),
\end{equation}
due to the expression \eqref{eq:f^(1)_upper_ext} of $f^{(1)}(z)$ in this sector and the series formula \eqref{eq:f^(1)_in_series} of $f^{(1)}(z)$. Hence, we find that the function
\begin{equation}
  z^{\frac{\alpha + 1/2 - \theta}{\theta + 1}} f^{(2)}(z^{\frac{1}{\theta + 1}}), \quad \arg z \in (-\pi, \pi)
\end{equation}
can be extended to a holomorphic function on $D^*(0, R^{\theta + 1})$, and by \eqref{eq:V_alpha(R)_at_0:1}, \eqref{eq:V_alpha(R)_at_0:2} and \eqref{eq:V_alpha(R)_at_0:3}, as $z \to 0$
\begin{equation}
  z^{\frac{\alpha + 1/2 - \theta}{\theta + 1}} f^{(2)}(z^{\frac{1}{\theta + 1}}) =
  \begin{cases}
    \bigO(1), & \alpha > \theta - 1, \\
    \bigO(\log z), & \alpha = \theta - 1, \\
    \bigO(z^{\frac{\alpha + 1}{\theta} - 1}), & -1 < \alpha < \theta - 1,
  \end{cases}
\end{equation}
and then this function is analytic at $0$. Like \eqref{eq:f^(1)_in_series}, we conclude that $f^{(2)}(z)$ has the power series representation
\begin{equation} \label{eq:f^(2)_in_series}
  f^{(2)}(z) = z^{\theta - \alpha - \frac{1}{2}} \sum^{\infty}_{k = 0} b_k z^{(\theta + 1)k}, \quad \limsup_{k \to \infty} \lvert b_k \rvert^{\frac{1}{k}} \leq R^{-(\theta + 1)}.
\end{equation}

By the arguments above, we have the power series representation of $V_{\alpha}(R)$:
\begin{lemma} \label{lem:power_series_repr_V_alpha(R)}
  Each $f(z) \in V_{\alpha}(R)$ has a unique power series representation such that its branch $f^{(1)}(z)$ in sector $\arg (-z) \in (\frac{-\pi + \gamma}{1 + \theta^{-1}}, \frac{\pi - \gamma}{1 + \theta^{-1}})$ is expressed by \eqref{eq:f^(1)_in_series} and its branch $f^{(R)}(z)$ in the sector $\arg z \in (\frac{-\theta^{-1} \pi - \gamma}{1 + \theta^{-1}}, \frac{\theta^{-1} \pi + \gamma}{1 + \theta^{-1}})$ is expressed by \eqref{eq:f^R_f^2} and \eqref{eq:f^(2)_in_series}. Conversely, any pair of power series \eqref{eq:f^(1)_in_series} and \eqref{eq:f^(2)_in_series} define a function $f(z) \in V_{\alpha}(R)$.
\end{lemma}

By a parallel argument, we have the power series representation of $\V_{\alpha}(R)$:
\begin{lemma}
  Each $\f(z) \in \V_{\alpha}(R)$ has a unique power series representation
  \begin{equation} \label{eq:power_series_repr_V_alpha(R)_tilde}
    \f(z) =
    \begin{cases}
      z^{\alpha + \frac{1}{2}} \sum^{\infty}_{k = 0} \a_k z^{(\theta + 1)k}, & \arg z \in (\frac{-\theta^{-1} \pi + \gamma}{1 + \theta^{-1}}, \frac{\theta^{-1} \pi - \gamma}{1 + \theta^{-1}}), \\
      (-z)^{\alpha + \frac{1}{2}} \sum^{\infty}_{k = 0} (-1)^{k + 1} \a_k (-z)^{(\theta + 1) k} & \\
      \phantom{(-z)}%^{\alpha + \frac{1}{2}}}
      \times
      \begin{cases}
        \frac{1}{2\sin \pi (\theta k + \alpha)}, & \theta k + \alpha \notin \intZ, \\
        \frac{1 + \theta^{-1}}{2\pi} (-1)^{\theta k + \alpha} \log (-z), & \theta k + \alpha \in \intZ,
      \end{cases}
      & \\
      + (-z)^{\frac{1}{2} - \frac{\alpha}{\theta}} \sum^{\infty}_{k = 0} \b_k z^{\frac{\theta + 1}{\theta} k}, & \arg (-z) \in (\frac{-\pi - \gamma}{1 + \theta^{-1}}, \frac{\pi + \gamma}{1 + \theta^{-1}}),
    \end{cases} \\
  \end{equation}
  such that
  \begin{align} \label{eq:growth_a_b_tilde}
    \limsup_{k \to \infty} \lvert \a_k \rvert^{\frac{1}{k}} \leq {}& R^{-(\theta + 1)}, & \limsup_{k \to \infty} \lvert \b_k \rvert^{\frac{1}{k}} \leq {}& R^{-\frac{\theta + 1}{\theta}}.
  \end{align}
  Conversely, any pair of sequences $\{ \a_k, \b_k \}$ satisfying \eqref{eq:growth_a_b_tilde} defines a function $\f \in \V_{\alpha}(R)$ by \eqref{eq:power_series_repr_V_alpha(R)_tilde}.
\end{lemma}

\begin{rmk}
  The series representation of $f(z) \in V_{\alpha}(R)$ does not depend on $\gamma$, since both $f^{(1)}(z)$ and $f^{(2)}(z)$ have power series representations independent of $\gamma$. The same property holds for $\f(z) \in \V_{\alpha}(R)$.
\end{rmk}

\subsection{Functions $G^{(\ell), \model}$, $H^{(\ell), \model}$, $\G^{(\ell), \model}$, $\H^{(\ell), \model}$, and the orthogonality}

For $\lambda \in \realR$, let
\begin{equation} \label{eq:D_1(k)_D_2(k)}
  % D_1(\lambda) = {}& \frac{\alpha + 3/2}{\theta + 1} + \frac{1}{2} - \lambda, &
  \translate(\lambda) = -\frac{\alpha + 3/2}{\theta + 1} + \lambda.
\end{equation}
We define, with the functions $I^{(1)}_{\theta, a}(z)$ and $I^{(2)}_{\theta, a}(z)$ defined in \eqref{eq:defn_I_theta_a} and \eqref{eq:defn_J_theta_a} respectively,
\begin{equation} \label{eq:G_divided_model}
  G^{\model}(z; \lambda) = \frac{\sqrt{2\pi} \theta^{\translate(\lambda)}}{\sqrt{\theta + 1}} \times
  \begin{cases}
    \frac{1}{2\pi} I^{(2)}_{\theta, \translate(\lambda)}(z), & \arg z \in (-\frac{\theta^{-1} + \gamma}{1 + \theta^{-1}}, \frac{\theta^{-1} + \gamma}{1 + \theta^{-1}}), \\
   -i e^{\translate(\lambda) \pi i} I^{(1)}_{\theta, \translate(\lambda)}((-z) e^{\frac{\theta}{\theta + 1} \pi i}), & \arg (-z) \in (-\frac{\pi - \gamma}{1 + \theta^{-1}}, 0), \\
   i e^{-\translate(\lambda) \pi i} I^{(1)}_{\theta, \translate(\lambda)}((-z) e^{-\frac{\theta}{\theta + 1} \pi i}), & \arg (-z) \in (0, \frac{\pi - \gamma}{1 + \theta^{-1}}).
  \end{cases}
\end{equation}
It is clear that $G^{\model}(z; \lambda)$ is analytic in $D^*(0, R) \setminus \{ \arg z = \pm \frac{\theta^{-1} + \gamma}{1 + \theta^{-1}} \}$, after analytic extension on $\realR_+$.

By \eqref{eq:G_split}, we have, with the rays oriented from $0$ to $\infty$,
\begin{equation} \label{eq:G^model_jump}
  G^{\model}_+(z; \lambda) - G^{\model}_-(z; \lambda) =
  \begin{cases}
    -e^{-\frac{2\alpha + 3}{\theta + 1} \pi i} G^{\model}(z e^{-\frac{2\pi}{\theta + 1} i}; \lambda) e^{\lambda 2\pi i}, & \arg z = \frac{\theta^{-1} \pi + \gamma}{1 + \theta^{-1}}, \\
    e^{\frac{2\alpha + 3}{\theta + 1} \pi i} G^{\model}(z e^{\frac{2\pi}{\theta + 1} i}; \lambda) e^{-\lambda 2\pi i}, & \arg z = \frac{-\theta^{-1} \pi - \gamma}{1 + \theta^{-1}}.
  \end{cases}
\end{equation}
By \eqref{eq:asy_IJK_at_0}, We have, as $z \to 0$,
\begin{equation} \label{eq:G^l_asy}
  G^{\model}(z; \lambda) =
  \begin{cases}
    \bigO(z^{(1 + \theta^{-1}) (1/2 - \translate(\lambda))}), & \translate(\lambda) > \frac{1}{2(1 + \theta)}, \\
    \begin{cases}
      \bigO(z^{1/2}), & \arg(-z) \in (-\frac{\pi - \gamma}{1 + \theta^{-1}}, \frac{\pi - \gamma}{1 + \theta^{-1}}), \\
      \bigO(z^{1/2} \log z), & \arg z \in (-\frac{\theta^{-1}\pi + \gamma}{1 + \theta^{-1}}, \frac{\theta^{-1}\pi + \gamma}{1 + \theta^{-1}}),
    \end{cases}
    & \translate(\lambda) = \frac{1}{2(1 + \theta)}, \\
    \begin{cases}
      \bigO(z^{(1 + \theta^{-1}) (1/2 - \translate(\lambda))}), & \arg(-z) \in (-\frac{\pi - \gamma}{1 + \theta^{-1}}, \frac{\pi - \gamma}{1 + \theta^{-1}}), \\
      \bigO(z^{(\theta + 1) \translate(\lambda)}), & \arg z \in (-\frac{\theta^{-1}\pi + \gamma}{1 + \theta^{-1}}, \frac{\theta^{-1}\pi + \gamma}{1 + \theta^{-1}}),
    \end{cases}
    & \translate(\lambda) < \frac{1}{2(1 + \theta)}.
  \end{cases}
\end{equation}

Next, we define a ``fractional part'' function such that for all $\ell \in \intZ$,
\begin{equation} \label{eq:round_of_R^lambda}
  R^{\lambda}(\ell) = \frac{\theta}{\theta + 1} \ell + m_{\ell}, \quad m_{\ell} \in \intZ, \text{ such that } R^{\lambda}(\ell) \in (0, 1],
\end{equation}
so that $m_0 = 1$ and as $\ell$ increases, both $\{ \ell + m_{\ell} \}$ and $\{ -m_{\ell} \}$ weakly increase. For $\ell \in \intZ$, we define
\begin{equation} \label{eq:defn_G^ell_model}
  G^{(\ell), \model}(z) = G^{\model}(z; R^{\lambda}(\ell)) z^{\ell}.
\end{equation}
We have, by \eqref{eq:G^model_jump},
\begin{equation} \label{eq:G^ell_model_jump}
  G^{(\ell), \model}_+(z) - G^{(\ell), \model}_-(z) =
  \begin{cases}
    -e^{-\frac{2\alpha + 3}{\theta + 1} \pi i} G^{(\ell), \model}(z e^{-\frac{2\pi}{\theta + 1} i}), & \arg z = \frac{\theta^{-1} \pi + \gamma}{1 + \theta^{-1}}, \\
    e^{\frac{2\alpha + 3}{\theta + 1} \pi i} G^{(\ell), \model}(z e^{\frac{2\pi}{\theta + 1} i}), & \arg z = \frac{-\theta^{-1} \pi - \gamma}{1 + \theta^{-1}}.
  \end{cases}
\end{equation}
Moreover, using the residue formula to the integral formulas \eqref{eq:defn_I_theta_a} and \eqref{eq:defn_J_theta_a} to express $I^{(1)}_{\theta, a}(z)$ and $I^{(2)}_{\theta, a}(z)$ into series, we find that $G^{(\ell), \model}(z)$ has a power series expansion as follows.
\begin{equation} \label{eq:series_G^ell_model}
  G^{(\ell), \model}(z) = 
  \begin{cases}
    (-z)^{-\frac{1}{2} + \frac{\alpha + 1}{\theta}} \sum^{\infty}_{k = 1 - m_{\ell}} a^{(\ell)}_k (-z)^{\frac{\theta + 1}{\theta} k}, & \arg (-z) \in (\frac{-\pi + \gamma}{1 + \theta^{-1}}, \frac{\pi - \gamma}{1 + \theta^{-1}}), \\
    z^{-\frac{1}{2} + \frac{\alpha + 1}{\theta}} \sum^{\infty}_{k = 1 - m_{\ell}} (-1)^k a^{(\ell)}_k z^{\frac{\theta + 1}{\theta} k} & \\
    \phantom{z}%^{-\frac{1}{2} + \frac{\alpha + 1}{\theta}} \sum^{\infty}_{k = 0}}
    \times
    \begin{cases}
      \frac{1}{2\sin \pi \left( \frac{k + \alpha + 1}{\theta} \right)}, & \frac{k + \alpha + 1}{\theta} \notin \intZ, \\
    \frac{\theta + 1}{2\pi} (-1)^{\frac{k + \alpha + 1}{\theta}} \log z, & \frac{k + \alpha + 1}{\theta} \in \intZ,
    \end{cases}
    & \\
    + z^{\theta - \alpha - \frac{1}{2}} \sum^{\infty}_{k = \ell + m_{\ell} - 1} b^{(\ell)}_k z^{(\theta + 1)k}, & \arg z \in (\frac{-\theta^{-1} \pi + \gamma}{1 + \theta^{-1}}, \frac{\theta^{-1} \pi + \gamma}{1 + \theta^{-1}}),
  \end{cases}
\end{equation}
where $a^{(\ell)}_k, b^{(\ell)}_k$ have explicit formulas, such that
\begin{align}
  \limsup_{k \to \infty} \lvert a^{(\ell)}_k \rvert^{\frac{1}{k}} = {}& 0, & \limsup_{k \to \infty} \lvert b^{(\ell)}_k \rvert^{\frac{1}{k}} = {}& 0, & a^{(\ell)}_{1 - m_{\ell}} \neq {}& 0, & b^{(\ell)}_{\ell + m_{\ell} - 1} \neq {}& 0.
\end{align}
As a consequence, if $\ell \in \natN$, $G^{(\ell), \model}(s) \in V_{\alpha}(R)$ for all $R \in (0, +\infty]$.

Let $\beta \in \realR$. We define similar to \eqref{eq:G_divided_model}, with the functions $I^{(2)}_{\theta, a}(z)$ and $I^{(3)}_{\theta, a}(z)$ defined in \eqref{eq:defn_J_theta_a} and \eqref{eq:defn_K_theta_a} respectively,
\begin{equation} \label{eq:H_divided_model}
  H^{\model}(z; \beta) = \frac{\sqrt{2\pi} \theta^{-\translate(\beta)}}{\sqrt{\theta + 1}} \times
  \begin{cases}
    \frac{1}{2\pi} I^{(2)}_{\theta, -\translate(\beta)}(-z), & \arg (-z) \in (-\frac{\pi + \gamma}{1 + \theta^{-1}}, \frac{\pi + \gamma}{1 + \theta^{-1}}), \\
    e^{\translate(\beta) \pi i} I^{(3)}_{\theta, -\translate(\beta)}(z e^{\frac{\pi i}{\theta + 1}}), & \arg z \in (-\frac{\theta^{-1} \pi - \gamma}{1 + \theta^{-1}}, 0), \\
    e^{-\translate(\beta) \pi i} I^{(3)}_{\theta, -\translate(\beta)}(z e^{-\frac{\pi i}{\theta + 1}}), & \arg z \in (0, \frac{\theta^{-1} \pi - \gamma}{1 + \theta^{-1}}).
  \end{cases}
\end{equation}
It is clear that $H^{\model}(z; \beta)$ is analytic in $D^*(0, R) \setminus \{ \arg z = \pm \frac{\theta^{-1} - \gamma}{1 + \theta^{-1}} \}$, after analytic extension on $\realR_+$.

Analogous to \eqref{eq:G^model_jump}, by \eqref{eq:H_split}, we have, with the rays oriented from $0$ to $\infty$,
\begin{equation} \label{eq:jump_H^model}
  H^{\model}_+(z; \beta) - H^{\model}_-(z; \beta) =
  \begin{cases}
    -e^{\frac{2\alpha + 3}{\theta + 1} \pi i} H^{\model}(z e^{-\frac{2\pi}{\theta + 1} i}; \beta) e^{-\beta 2\pi i}, & \arg z = \frac{\theta^{-1} \pi - \gamma}{1 + \theta^{-1}}, \\
    e^{-\frac{2\alpha + 3}{\theta + 1} \pi i} H^{\model}(z e^{\frac{2\pi}{\theta + 1} i}; \beta) e^{\beta 2\pi i}, & \arg z = \frac{-\theta^{-1} \pi + \gamma}{1 + \theta^{-1}}.
  \end{cases}
\end{equation}
Analogous to \eqref{eq:G^l_asy}, by \eqref{eq:asy_IJK_at_0}, we have as $z \to 0$,
\begin{equation} \label{eq:H^l_asy}
  H^{\model}(z; \beta) =
  \begin{cases}
  \bigO(z^{-(\theta + 1) \translate(\beta)}), & \translate(\beta) > \frac{-1}{2(\theta + 1)}, \\
    \begin{cases}
      \bigO(z^{1/2}), & \arg z \in (\frac{-\theta^{-1} \pi + \gamma}{1 + \theta^{-1}}, \frac{\theta^{-1} \pi - \gamma}{1 + \theta^{-1}}), \\
      \bigO(z^{1/2} \log z), & \arg (-z) \in (\frac{-\pi - \gamma}{1 + \theta^{-1}}, \frac{\pi + \gamma}{1 + \theta^{-1}}),
    \end{cases}
    & \translate(\beta) = \frac{-1}{2(\theta + 1)}, \\
    \begin{cases}
      \bigO(z^{-(\theta + 1) \translate(\beta)}), & \arg z \in (\frac{-\theta^{-1} \pi + \gamma}{1 + \theta^{-1}}, \frac{\theta^{-1} \pi - \gamma}{1 + \theta^{-1}}), \\
      \bigO(z^{(1 + \theta^{-1}) (\translate(\beta) + 1/2)}), & \arg (-z) \in (\frac{-\pi - \gamma}{1 + \theta^{-1}}, \frac{\pi + \gamma}{1 + \theta^{-1}}),
    \end{cases}
    & \translate(\beta) < \frac{-1}{2(\theta + 1)}.
  \end{cases}
\end{equation}

Next, we define a ``fractional part'' function similar to \eqref{eq:round_of_R^lambda}, such that for all $\ell \in \intZ$,
\begin{equation} \label{eq:defn_n_ell}
  R^{\beta}(\ell) = -\frac{\theta}{1 + \theta} \ell - n_{\ell}, \quad n_{\ell} \in \intZ, \text{ such that } R^{\beta}(\ell) \in \left( -\frac{\theta}{\theta + 1}, \frac{1}{\theta + 1} \right].
\end{equation}
For $\ell \in \intZ$, analogous to \eqref{eq:defn_G^ell_model}, define 

\begin{equation} \label{eq:defn_H^ell_model}
  H^{(\ell), \model}(z) = H^{\model}(z; R^{\beta}(\ell)) z^{\ell}.
\end{equation}
We have, by \eqref{eq:jump_H^model}, analogous to \eqref{eq:G^ell_model_jump}
\begin{equation} \label{eq:H^ell_model_jump}
  H^{(\ell), \model}_+(z) - H^{(\ell), \model}_-(z) =
  \begin{cases}
    -e^{\frac{2\alpha + 3}{\theta + 1} \pi i} H^{(\ell), \model}(z e^{-\frac{2\pi}{\theta + 1} i}), & \arg z = \frac{\theta^{-1} \pi - \gamma}{1 + \theta^{-1}}, \\
    e^{-\frac{2\alpha + 3}{\theta + 1} \pi i} H^{(\ell), \model}(z e^{\frac{2\pi}{\theta + 1} i}), & \arg z = \frac{-\theta^{-1} \pi + \gamma}{1 + \theta^{-1}}.
  \end{cases}
\end{equation}

Similar to \eqref{eq:series_G^ell_model}, $H^{(\ell), \model}(z)$ also has the power series representation
\begin{equation} \label{eq:series_H^ell_model}
  H^{(\ell), \model}(z) =
  \begin{cases}
    z^{\alpha + \frac{3}{2}} \sum^{\infty}_{k = \ell + n_{\ell}} c^{(\ell)}_k z^{(\theta + 1) k}, & \arg z \in (-\frac{\theta^{-1} \pi - \gamma}{1 + \theta^{-1}}, \frac{\theta^{-1} \pi - \gamma}{1 + \theta^{-1}}), \\
    (-z)^{\alpha + \frac{3}{2}} \sum^{\infty}_{k = \ell + n_{\ell}} (-1)^k c^{(\ell)}_k (-z)^{(\theta + 1) k} & \\
    \phantom{(-z)^{\alpha + \frac{3}{2}}}
    \times
    \begin{cases}
      \frac{1}{2\sin \pi(\theta k + \alpha)}, & \theta k + \alpha \notin \intZ, \\
      \frac{1 + \theta^{-1}}{2\pi} (-1)^{\theta k + \alpha} \log z, & \theta k + \alpha \notin \intZ,
    \end{cases}
                                                                                                  & \\
  + (-z)^{\frac{1}{2} - \frac{\alpha + 1}{\theta}} \sum^{\infty}_{k = -n_{\ell}} d^{(\ell)}_k (-z)^{\frac{\theta + 1}{\theta} k}, & \arg (-z) \in (-\frac{\pi + \gamma}{1 + \theta^{-1}}, \frac{\pi + \gamma}{1 + \theta^{-1}}),
  \end{cases}
\end{equation}
such that
\begin{align} \label{eq:properties_H_coeff}
  \limsup_{k \to \infty} \lvert c^{(\ell)}_k \rvert^{\frac{1}{k}} = {}& 0, & \limsup_{k \to \infty} \lvert d^{(\ell)}_k \rvert^{\frac{1}{k}} = {}& 0, & c^{(\ell)}_{\ell + n_{\ell}} \neq {}& 0, & d^{(\ell)}_{-n_{\ell}} \neq {}& 0.
\end{align}

For any $r > 0$, we estimate the values of $G^{(\ell), \model}(z)$ and $H^{(\ell), \model}(z)$ on the circle $\lvert z \rvert = r$. By the definition formulas \eqref{eq:defn_G^ell_model} and \eqref{eq:G_divided_model} of $G^{(\ell), \model}(z)$, \eqref{eq:defn_H^ell_model} and \eqref{eq:H_divided_model} of $H^{(\ell), \model}(z)$, the range of $R^{\lambda}(\ell)$ in \eqref{eq:round_of_R^lambda}, the range of $R^{\beta}(\ell)$ in \eqref{eq:defn_n_ell}, and the integral formulas, \eqref{eq:defn_I_theta_a}, \eqref{eq:defn_J_theta_a} and \eqref{eq:defn_K_theta_a} of $I^{(1)}_{\theta, a}(z)$, $I^{(2)}_{\theta, a}(z)$ and $I^{(3)}_{\theta, a}(z)$, we have that there is $C_r > 0$ independent of $\ell$, such that
\begin{align} \label{eq:bound_G^ell_model}
  \lvert G^{(\ell), \model}(z) \rvert < {}& C_r r^{\ell}, & \lvert H^{(\ell), \model}(z) \rvert < {}& C_r r^{\ell}, && \text{for all $\lvert z \rvert = r$ and $\ell \in \intZ$}.
\end{align}
As $r \to \infty$, by the definition formulas of $G^{(\ell), \model}(z)$ and $H^{(\ell), \model}(z)$, and the limit formulas \eqref{eq:asy_I_theta_a_infty}, \eqref{eq:asy_J_theta_a_infty} and \eqref{eq:asy_K_theta_a_infty} respectively for $I^{(1)}_{\theta, a}(z)$, $I^{(2)}_{\theta, a}(z)$ and $I^{(3)}_{\theta, a}(z)$, we have uniformly in $\ell \in \intZ$
\begin{align} \label{eq:G_divided}
  e^z G^{(\ell), \model}(z) = {}& z^{\ell} (1 + \bigO(z^{-1})), & e^{-z} H^{(\ell), \model}(z) = {}& z^{\ell} (1 + \bigO(z^{-1})), & z \to {}& \infty. % \quad \text{where} \quad G^{(\ell), \pre}(z) = e^z G^{(\ell), \model}(z).
\end{align}

In parallel to $G^{(\ell), \model}(z)$ and $H^{(\ell), \model}(z)$ We also define 
\begin{align}
  \R^{\beta}(\ell) = {}& \frac{1}{\theta + 1} - \frac{\theta}{\theta + 1} \ell - \n_{\ell}, && \n_{\ell} \in \intZ, \text{ such that } \R^{\beta}(\ell) \in \left( -\frac{\theta}{\theta + 1}, \frac{1}{\theta + 1} \right], \label{eq:defn_n_ell_tilde} \\
  \R^{\lambda}(\ell) = {}& \frac{1}{\theta + 1} + \frac{\theta}{\theta + 1} \ell + \m_{\ell}, && \m_{\ell} \in \intZ, \text{ such that } \R^{\lambda}(\ell) \in (0, 1].
\end{align}
and
\begin{align} \label{eq:G_H_ell_tilde}
  \G^{(\ell), \model}(z) = {}& H^{\model}(z; \R^{\beta}(\ell)) z^{\ell}, & \H^{(\ell), \model}(z) = {}& G^{\model}(z; \R^{\lambda}(\ell)) z^{\ell}.
\end{align}
We have that $\G^{(\ell), \model}(z)$ and $\H^{(\ell), \model}(z)$ satisfy jump conditions similar to \eqref{eq:G^ell_model_jump} and \eqref{eq:H^ell_model_jump}, have series representations similar to \eqref{eq:series_G^ell_model} and \eqref{eq:series_H^ell_model}, and have estimates similar to \eqref{eq:bound_G^ell_model} and \eqref{eq:G_divided}. We omit the details. At last, we note that if $\ell \in \natN$, then $\G^{(\ell), \model}(z) \in \V_{\alpha}(R)$ for all $R \in (0, +\infty]$.

% \subsection{Orthogonality between $G^{(\ell), \model}(z)$ and $H^{(\ell), \model}(z)$, and between $\G^{(\ell), \model}(z)$ and $\H^{(\ell), \model}(z)$}

$H^{(\ell), \model}$ (resp.~$\H^{(\ell), \model}$) acts on $V_{\alpha}(R)$ (resp.~$\V_{\alpha}(R)$) as a linear functional by the integral formula in the following lemma:
\begin{lemma} \label{lem:inner_prod_defn}
  Suppose $f \in V_{\alpha}(R)$ and $\f \in \V_{\alpha}(R)$. For all $\ell \in \intZ$ and $R' \in (0, R)$, the inner products
  \begin{align}
    \langle f, H^{(\ell), \model} \rangle_{V_{\alpha}(R)} = {}& \frac{1}{2\pi i} \oint_{\lvert z \rvert = R'} H^{(\ell), \model}(z) f(z) \frac{dz}{z}, \label{eq:defn_inner_prod_with_H} \\
    \langle \f, \H^{(\ell), \model} \rangle_{\V_{\alpha}(R)} = {}& \frac{1}{2\pi i} \oint_{\lvert z \rvert = R'} \H^{(\ell), \model}(z) \f(z) \frac{dz}{z}, \label{eq:defn_inner_prod_with_H_tilde}
  \end{align}
  are well defined, in the sense that they are independent of $R'$. Furthermore, they are also independent of $\gamma$, in the sense that if $f$ is given by $\{ a_k \}, \{ b_k \}$ as in \eqref{eq:f^(1)_in_series}, \eqref{eq:f^R_f^2} and \eqref{eq:f^(2)_in_series}, then the integral in \eqref{eq:defn_inner_prod_with_H} only depends on $\{ a_k \}, \{ b_k \}$ but not $\gamma$, and if $\f$ is given by $\{ \a_k \}, \{ \b_k \}$ as in \eqref{eq:power_series_repr_V_alpha(R)_tilde}, then the integral in \eqref{eq:defn_inner_prod_with_H_tilde} only depends on $\{ \a_k \}, \{ \b_k \}$ but not $\gamma$. In particular, we can take $\gamma = 0$ in the integral formulas in \eqref{eq:defn_inner_prod_with_H} and \eqref{eq:defn_inner_prod_with_H_tilde}.
\end{lemma}

\begin{proof}
  We prove the well-definedness and $\gamma$-independence of $\langle f, H^{(\ell), \model} \rangle_{V_{\alpha}(R)}$ in \eqref{eq:defn_inner_prod_with_H}, and the well-definedness and $\gamma$-independence of $\langle \f, \H^{(\ell), \model} \rangle_{\V_{\alpha}(R)}$ in \eqref{eq:defn_inner_prod_with_H_tilde} is analogous.
  
  We denote the integral on the right-hand side of \eqref{eq:defn_inner_prod_with_H} $I_{\gamma}(R')$. To see that $I_{\gamma}(R')$ is independent of $R'$, we take $0 < R_1 < R_2$, and define the following four arcs, with $k = 1, 2$ and with counterclockwise orientation
  \begin{equation}
    \begin{aligned}
      C^{(1)}_k = {}& \{ R_k e^{it} : \frac{-\theta^{-1} \pi + \gamma}{1 + \theta^{-1}} < t < \frac{\theta^{-1} \pi - \gamma}{1 + \theta^{-1}} \}, & C^{(2)}_k = {}& \{ -R_k e^{it} : \frac{-\pi + \gamma}{1 + \theta^{-1}} < t < \frac{\pi - \gamma}{1 + \theta^{-1}} \}, \\
      C^{(3)}_k = {}& \{ R_k e^{it} : \frac{\theta^{-1} \pi - \gamma}{1 + \theta^{-1}} < t < \frac{\theta^{-1} \pi + \gamma}{1 + \theta^{-1}} \}, & C^{(4)}_k = {}& \{ R_k e^{it} : \frac{-\theta^{-1} \pi - \gamma}{1 + \theta^{-1}} < t < \frac{-\theta^{-1} \pi + \gamma}{1 + \theta^{-1}} \}.
    \end{aligned}
  \end{equation}
  Then let
  \begin{equation} \label{eq:integral_r1_r2}
    \quad J^{(j)}_k = \frac{1}{2\pi i} \int_{C^{(j)}_k} f(w) H^{(\ell), \model}(w) \frac{dw}{w},
  \end{equation}
  and
  \begin{equation}
    \begin{aligned}
      I^{(1)}_1 = {}& \int^{R_2 e^{\frac{\theta^{-1} \pi - \gamma}{1 + \theta^{-1}} i}}_{R_1 e^{\frac{\theta^{-1} \pi - \gamma}{1 + \theta^{-1}} i}} f(w) H^{(\ell), \model}_-(w) \frac{dw}{w}, & I^{(1)}_2 = {}& \int^{R_2 e^{\frac{-\theta^{-1} \pi + \gamma}{1 + \theta^{-1}} i}}_{R_1 e^{\frac{-\theta^{-1} \pi + \gamma}{1 + \theta^{-1}} i}} f(w) H^{(\ell), \model}_+(w) \frac{dw}{w}, \\
      I^{(2)}_1 = {}& \int^{R_2 e^{\frac{-\theta^{-1} \pi - \gamma}{1 + \theta^{-1}} i}}_{R_1 e^{\frac{-\theta^{-1} \pi - \gamma}{1 + \theta^{-1}} i}} f_-(w) H^{(\ell), \model}(w) \frac{dw}{w}, & I^{(2)}_2 = {}&  \int^{R_2 e^{\frac{\theta^{-1} \pi + \gamma}{1 + \theta^{-1}} i}}_{R_1 e^{\frac{\theta^{-1} \pi + \gamma}{1 + \theta^{-1}} i}} f_+(w) H^{(\ell), \model}(w) \frac{dw}{w}, \\
      I^{(3)}_1 = {}& \int^{R_2 e^{\frac{\theta^{-1} \pi + \gamma}{1 + \theta^{-1}} i}}_{R_1 e^{\frac{\theta^{-1} \pi + \gamma}{1 + \theta^{-1}} i}} f_-(w) H^{(\ell), \model}(w) \frac{dw}{w}, & I^{(3)}_2 = {}& \int^{R_2 e^{\frac{\theta^{-1} \pi - \gamma}{1 + \theta^{-1}} i}}_{R_1 e^{\frac{\theta^{-1} \pi - \gamma}{1 + \theta^{-1}} i}} f(w) H^{(\ell), \model}_+(w) \frac{dw}{w}, \\
      I^{(4)}_1 = {}& \int^{R_2 e^{\frac{-\theta^{-1} \pi + \gamma}{1 + \theta^{-1}} i}}_{R_1 e^{\frac{-\theta^{-1} \pi + \gamma}{1 + \theta^{-1}} i}} f(w) H^{(\ell), \model}_-(w) \frac{dw}{w}, & I^{(4)}_2 = {}& \int^{R_2 e^{\frac{-\theta^{-1} \pi - \gamma}{1 + \theta^{-1}} i}}_{R_1 e^{\frac{-\theta^{-1} \pi - \gamma}{1 + \theta^{-1}} i}} f_+(w) H^{(\ell), \model}(w) \frac{dw}{w}.
    \end{aligned}
  \end{equation}
  By the Cauchy integral formula, we have for all $j = 1, 2, 3, 4$, $J^{(j)}_2 - J^{(j)}_1 + I^{(j)}_2 - I^{(j)}_1 = 0$. Using jump condition \eqref{eq:H^ell_model_jump} for $H^{(\ell), \model}$ and jump condition \eqref{eq:jump_V_alpha(R)} for $f$, we have that
  \begin{equation}
    \begin{split}
      I^{(3)}_2 - I^{(1)}_1 = {}& \int^{R_2}_{R_1} f(r e^{\frac{\theta^{-1} \pi - \gamma}{1 + \theta^{-1}} i}) \left( H^{(\ell), \model}_+(r e^{\frac{\theta^{-1} \pi - \gamma}{1 + \theta^{-1}} i}) - H^{(\ell), \model}_+(r e^{\frac{\theta^{-1} \pi - \gamma}{1 + \theta^{-1}} i}) \right) \frac{dr}{r} \\
      = {}& -\int^{R_2}_{R_1} f(r e^{\frac{\theta^{-1} \pi - \gamma}{1 + \theta^{-1}} i}) e^{\frac{2\alpha + 3}{\theta + 1} \pi i} H^{(\ell), \model}(r e^{\frac{-\theta^{-1} \pi - \gamma}{1 + \theta^{-1}} i}) \frac{dr}{r} \\
      = {}& -\int^{R_2}_{R_1} \left( f_+(r e^{\frac{-\theta^{-1} \pi - \gamma}{1 + \theta^{-1}} i}) - f_+(r e^{\frac{-\theta^{-1} \pi - \gamma}{1 + \theta^{-1}} i}) \right) H^{(\ell), \model}(r e^{\frac{-\theta^{-1} \pi - \gamma}{1 + \theta^{-1}} i}) \frac{dr}{r} \\
      = {}& I^{(2)}_1 - I^{(4)}_2.
    \end{split}
  \end{equation}
  Similarly, we have $I^{(2)}_2 - I^{(3)}_1 = I^{(4)}_1 - I^{(1)}_2$. Hence, we derive that $I_{\gamma}(R_2) = \frac{1}{2\pi i} \sum^{4}_{j = 1} J^{(j)}_2 = \frac{1}{2\pi i} \sum^{4}_{j = 1} J^{(j)}_1 = I_{\gamma}(R_2)$, and prove the well-definedness.
  
  To see the $I_{\gamma}(R')$ is independent of $\gamma$, we compare $I_{\gamma}(R')$ for a positive $\gamma$ with the degenerate case with $\gamma = 0$. Since both $I_{\gamma}(R')$ and $I_0(R')$ are integrals over $\{ \lvert z \rvert = R' \}$, the difference stems from the integrals over the sectors $\arg z \in (\frac{\theta^{-1} \pi}{1 + \theta^{-1}}, \frac{\theta^{-1} \pi + \gamma}{1 + \theta^{-1}})$, $\arg z \in (\frac{\theta^{-1} \pi - \gamma}{1 + \theta^{-1}}, \frac{\theta^{-1} \pi}{1 + \theta^{-1}})$, $\arg z \in (\frac{-\theta^{-1} \pi}{1 + \theta^{-1}}, \frac{-\theta^{-1} \pi + \gamma}{1 + \theta^{-1}})$ and $\arg z \in (\frac{-\theta^{-1} \pi - \gamma}{1 + \theta^{-1}}, \frac{-\theta^{-1} \pi}{1 + \theta^{-1}})$. To be precise, we have
  \begin{subequations}
    \begin{align}
      I_{\gamma}(R') - I_0(R') = {}& \frac{1}{2\pi i} \int^{R' e^{\frac{\theta^{-1} \pi + \gamma}{1 + \theta^{-1}} i}}_{R' e^{\frac{\theta^{-1} \pi}{1 + \theta^{-1}} i}} e^{-\frac{2\alpha + 3}{\theta + 1} \pi i} f(z e^{-\frac{2\pi}{1 + \theta} i}) H^{(\ell), \model}(z) \frac{dz}{z} \label{eq:indpt_gamma_1} \\
                                 & - \frac{1}{2\pi i} \int^{R' e^{\frac{\theta^{-1} \pi}{1 + \theta^{-1}} i}}_{R' e^{\frac{\theta^{-1} \pi - \gamma}{1 + \theta^{-1}} i}} f(z) e^{\frac{2\alpha + 3}{\theta + 1} \pi i} H^{(\ell), \model}(z e^{-\frac{2\pi}{1 + \theta} i}) \frac{dz}{z} \label{eq:indpt_gamma_2} \\
                                 & - \frac{1}{2\pi i} \int^{R' e^{\frac{-\theta^{-1} \pi + \gamma}{1 + \theta^{-1}} i}}_{R' e^{\frac{-\theta^{-1} \pi}{1 + \theta^{-1}} i}} f(z) e^{-\frac{2\alpha + 3}{\theta + 1} \pi i} H^{(\ell), \model}(z e^{\frac{2\pi}{1 + \theta} i}) \frac{dz}{z} \label{eq:indpt_gamma_3} \\
                                 & + \frac{1}{2\pi i} \int^{R' e^{\frac{\theta^{-1} \pi}{1 + \theta^{-1}} i}}_{R' e^{\frac{\theta^{-1} \pi - \gamma}{1 + \theta^{-1}} i}} e^{\frac{2\alpha + 3}{\theta + 1} \pi i} f(z e^{\frac{2\pi}{1 + \theta} i}) H^{(\ell), \model}(z) \frac{dz}{z}. \label{eq:indpt_gamma_4}
    \end{align}
  \end{subequations}
  We see that the integral in \eqref{eq:indpt_gamma_1} is equal to the integral in \eqref{eq:indpt_gamma_3}, and the integral in \eqref{eq:indpt_gamma_2} is equal to the integral in \eqref{eq:indpt_gamma_4}. Hence $I_{\gamma}(R') - I_0(R')$ vanishes due to cancellation.
\end{proof}

\begin{lemma} \label{lem:GH_orthogonal}
  For all $j, k \in \natN$, $\langle G^{(j), \model}, H^{(-k), \model} \rangle_{V_{\alpha}(R)} = \delta_{j, k}$ and $\langle \G^{(j), \model}, \H^{(-k), \model} \rangle_{\V_{\alpha}(R)} = \delta_{j, k}$.
\end{lemma}

\begin{proof}
  We prove the result for $\langle G^{(j), \model}, H^{(-k), \model} \rangle_{V_{\alpha}(R)}$, and that for $\langle \G^{(j), \model}, \H^{(-k), \model} \rangle_{\V_{\alpha}(R)}$ is analogous.

  Suppose $j \leq k$. We write
  \begin{equation}
    \langle G^{(j), \model}, H^{(-k), \model} \rangle_{V_{\alpha}(R)} = \frac{1}{2\pi i} \oint_{\lvert z \rvert = R'} \left( e^z G^{(j), \model}(z) \right) \left( e^{-z} H^{(-k), \model}(z) \right) \frac{dz}{z}.
  \end{equation}
  By lemma \ref{lem:inner_prod_defn}, we know that the integral is independent of $R' \in (0, +\infty)$ and $\gamma$. We take the limit $R' \to \infty$ and fix $\gamma$ as a small positive number. By the limit formulas in \eqref{eq:G_divided}, we have that as $R' \to \infty$, $\langle G^{(j), \model}, H^{(-k), \model} \rangle_{V_{\alpha}(R)} = \delta_{j, k} + \bigO(1/R')$. Since the inner product is independent of $R'$, we conclude that it is $\delta_{j, k}$. 

  If $k < j$, by \eqref{eq:defn_G^ell_model} and \eqref{eq:defn_H^ell_model}, we write
  \begin{equation}
    \langle G^{(j), \model}, H^{(-k), \model} \rangle_{V_{\alpha}(R)} = \frac{1}{2\pi i} \oint_{\lvert z \rvert = R'} G^{\model}(z; R^{\lambda}(j)) H^{\model}(z; R^{\beta}(-k)) z^{j - k} \frac{dz}{z},
  \end{equation}
  and, by Lemma \ref{lem:inner_prod_defn}, let $\gamma = 0$ and take the limit $R' \to 0$. By the limit formulas \eqref{eq:G^l_asy} and \eqref{eq:H^l_asy},
  \begin{multline}
    G^{\model}(z; R^{\lambda}(j)) H^{\model}(z; R^{\beta}(-k)) = \\
    \begin{cases}
      \bigO(z^{-(\theta + 1) \translate(R^{\beta}(-k))}) & \\
      \times
    \begin{cases}
      \bigO(z^{(1 + \theta^{-1}) (1/2 - \translate(R^{\lambda}(j)))}), & \translate(R^{\lambda}(j)) > \frac{1}{2(\theta + 1)}, \\
      \bigO(z^{\frac{1}{2}} \log z), & \translate(R^{\lambda}(j)) = \frac{1}{2(\theta + 1)}, \\
      \bigO(z^{(\theta + 1) \translate(R^{\lambda}(j))}), & \translate(R^{\lambda}(j)) < \frac{1}{2(\theta + 1)},
    \end{cases}
      & \arg z \in (\frac{-\theta^{-1} \pi}{1 + \theta^{-1}}, \frac{\theta^{-1} \pi}{1 + \theta^{-1}}), \\
  % \end{multline}
  % and on the arc $\arg (-z) \in (\frac{-\pi}{1 + \theta^{-1}}, \frac{\pi}{1 + \theta^{-1}})$
  % \begin{multline}
  %   G^{\model}(z; R^{\lambda}(j)) H^{\model}(z; R^{\beta}(-k)) = \\
      \bigO(z^{(1 + \theta^{-1})(1/2 - \translate(R^{\lambda}(j)))} & \\
      \times
      \begin{cases}
       \bigO(z^{-(\theta + 1)\translate(R^{\beta}(-k))}), & \translate(R^{\beta}(-k)) > \frac{-1}{2(\theta + 1)}, \\
      \bigO(z^{\frac{1}{2}} \log z), & \translate(R^{\beta}(-k)) = \frac{-1}{2(\theta + 1)}, \\
      \bigO(z^{(1 + \theta^{-1})(\translate(R^{\beta}(-k)) + 1/2)}), & \translate(R^{\beta}(-k)) < \frac{-1}{2(\theta + 1)},
      \end{cases}
      & \arg (-z) \in (\frac{-\pi}{1 + \theta^{-1}}, \frac{\pi}{1 + \theta^{-1}}).
    \end{cases}
  \end{multline}
  Since by \eqref{eq:round_of_R^lambda} and \eqref{eq:defn_n_ell}, $R^{\lambda}(j) \in (0, 1]$ and $R^{\beta}(-k) \in (-\theta/(\theta + 1), 1/(\theta + 1)]$, we have that $G^{\model}(z; R^{\lambda}(j)) H^{\model}(z; R^{\beta}(-k)) = o(z^{-1})$ on the contour $\{ \lvert z \rvert = R' \}$, and so \linebreak[4] $\langle G^{(j), \model}, H^{(-k), \model} \rangle_{V_{\alpha}(R)} = o((R')^{j - k - 1})$. Since the inner product is independent of $R'$, we conclude that it is $0$.
\end{proof}

\begin{lemma} \label{lem:GH_complete}
  Suppose $f \in V_{\alpha}(R)$ and $\f \in \V_{\alpha}(R)$. If $\langle f, H^{(-k), \model} \rangle_{V_{\alpha}(R)} = 0$ for all $k \in \natN$, then $f(z) = 0$; if $\langle \f, \H^{(-k), \model} \rangle_{\V_{\alpha}(R)} = 0$ for all $k \in \natN$, then $\f(z) = 0$.
\end{lemma}

\begin{proof}
  We prove the result for $f$, and that for $\f$ is analogous.
  
  Suppose $f(z) \in V_{\alpha}(R)$ and $f(z)$ has the power series representation as in \eqref{lem:power_series_repr_V_alpha(R)}. Let $k^{(a)} \in \natN \cup \{ \infty \}$ be the index of the smallest nonzero coefficient among $\{ a_k \}$ and $k^{(b)} \in \natN \cup \{ \infty \}$ be the index of the smallest nonzero coefficient among $\{ b_k \}$. Below we show that if the condition of the lemma is satisfied, then both $k^{(a)}$ and $k^{(b)}$ are $\infty$, that is, $f(z) = 0$.

  We can check that integers $n_{\ell}$ defined in \eqref{eq:defn_n_ell} have the property that both $\{ -n_{\ell} \}$ and $\{ \ell + n_{\ell} \}$ are weakly increasing in $\ell \in \intZ$. More precisely, $n_{-\ell - 1} - n_{-\ell}$ is $0$ or $1$, $n_{-\ell} \to +\infty$ and $\ell - n_{-\ell} \to +\infty$ as $\ell \to \infty$, and $n_0 = 0$.

  Suppose $k^{(a)}$ and $k^{(b)}$ are not both $\infty$, then there is a largest $\ell^{\max} \in \natN$ such that
  \begin{align} \label{eq:two_ineq_f_H}
    k^{(a)} - n_{-\ell^{\max}} \geq {}& 0, & \text{and} & & k^{(b)} + 1 + n_{-\ell^{\max}} - \ell^{\max} \geq {}& 0.
  \end{align}
  We claim that $k^{(a)} - n_{-\ell^{\max}}$ and $k^{(b)} + 1 + n_{-\ell^{\max}} - \ell^{\max}$ cannot be both greater than $0$, otherwise in the two inequalities in \eqref{eq:two_ineq_f_H}, $\ell^{\max}$ can be replaced by $\ell^{\max} + 1$, which is contradictory to the maximality of $\ell^{\max}$. In case both of $k^{(a)} - n_{-\ell^{\max}}$ and $k^{(b)} + 1 + n_{-\ell^{\max}} - \ell^{\max}$ are $0$, we find that the two values of $k^{(a)} - n_{-(\ell^{\max} - 1)}$ and $k^{(b)} + 1 + n_{-(\ell^{\max} - 1)} - (\ell^{\max} - 1)$ are $0$ and $1$. Below let $\ell^* = \ell^{\max}$ if only one between $k^{(a)} - n_{-\ell^{\max}}$ and $k^{(b)} + 1 + n_{-\ell^{\max}} - \ell^{\max}$ is $0$, and let $\ell^* = \ell^{\max} - 1$ if both of them are $0$. Without loss of generality, we assume $k^{(a)} - n_{\ell^*} = 0$ and $k^{(b)} + 1 + n_{-\ell^*} - \ell^* > 0$.

  Below we take $\gamma = 0$ and by the series representation \eqref{eq:series_H^ell_model} of $H^{(-\ell^{\star}), \model}$, we have
  \begin{subequations}
    \begin{align}
      0 = {}& \langle f, H^{(-\ell^*), \model} \rangle_{V_{\alpha}(R)} \notag \\
      = {}& \frac{1}{2\pi i} \int_{\lvert z \rvert = R', \ \arg (-z) \in (\frac{-\pi}{1 + \theta^{-1}}, \frac{\pi}{1 + \theta^{-1}})} (a_{k^{(a)}} d^{(-\ell^*)}_{-n_{-\ell^{\star}}} + \bigO(z^{-\frac{\theta + 1}{\theta}})) (-z)^{-\frac{\theta + 1}{\theta} (k^{(a)} - n_{-\ell^*})} \frac{dz}{z} \label{eq:integral_orthogonal_leftarc} \\
            & + \frac{1}{2\pi i} \int_{\lvert z \rvert = R', \ \arg z \in (\frac{-\theta^{-1} \pi}{1 + \theta^{-1}}, \frac{\theta^{-1} \pi}{1 + \theta^{-1}})} (b_{k^{(b)}} c^{(-\ell^*)}_{-\ell^{\star} + n_{-\ell^{\star}}} + \bigO(z^{\theta + 1})) z^{(\theta + 1)(k^{(b)} + 1 + n_{\ell^*} - \ell^*)} \frac{dz}{z}. \label{eq:integral_orthogonal_rightarc}
    \end{align}
  \end{subequations}
  Letting $R' \to 0$, we find the integral in \eqref{eq:integral_orthogonal_rightarc} vanishes and the integral in \eqref{eq:integral_orthogonal_leftarc} is $\frac{\theta}{\theta + 1} a_{k^{(a)}} d^{(-\ell^*)}_{-n_{-\ell^{\star}}}$. Since $d^{(-\ell^{\star})}_{-n_{-\ell^{\star}}} \neq 0$ (see \eqref{eq:properties_H_coeff}), we derive that $a_{k^{(a)}} = 0$, which is a contradiction.
\end{proof}

From Lemmas \ref{lem:GH_orthogonal} and \ref{lem:GH_complete}, we have the following series representation of functions in $V_{\alpha}(R)$ and $\V_{\alpha}(R)$:
\begin{lemma} \label{lem:V_alpha(R)_series_rep}
  Suppose $f(z) \in V_{\alpha}(R)$ and $\f(z) \in \V_{\alpha}(R)$. We have that $f(z)$ and $\f(z)$ have unique series representations as follows:
  \begin{align}
    f(z) = {}& \sum^{\infty}_{\ell = 0} c_{\ell} G^{(\ell), \model}(z), & \limsup_{\ell \to \infty} \lvert c_{\ell} \rvert^{\frac{1}{\ell}} \leq {}& \frac{1}{R}, \label{eq:f_V_alpha(R)_series_rep} \\
    \f(z) = {}& \sum^{\infty}_{\ell = 0} \c_{\ell} \G^{(\ell), \model}(z), & \limsup_{\ell \to \infty} \lvert \c_{\ell} \rvert^{\frac{1}{\ell}} \leq {}& \frac{1}{R}. \label{eq:f_V_alpha(R)_series_rep_tilde}
  \end{align}
\end{lemma}

\begin{proof}
  We prove the existence and uniqueness of the series representation of $f(z)$ in \eqref{eq:f_V_alpha(R)_series_rep}, and the proof for existence and uniqueness of the series representation of $\f(z)$ in \eqref{eq:f_V_alpha(R)_series_rep_tilde} is analogous.
  
  By the property \eqref{eq:bound_G^ell_model}, we find that the series in \eqref{eq:f_V_alpha(R)_series_rep} converges for all $z \in D^*(0, R)$ except for the two rays $\{ \arg z = \pm \frac{\theta^{-1} \pi + \gamma}{1 + \theta^{-1}} \}$, and the series converges to a function in $V_{\alpha}(R)$.

  For $f(z) \in V_{\alpha}(R)$, we let
  \begin{equation}
    c^{\star}_{\ell} = \frac{1}{2\pi i} \oint_{\lvert z \rvert = R'} f(z) H^{(-\ell), \model}(z) \frac{dz}{z}, \quad \text{and} \quad f^{\star}(z) = \sum^{\infty}_{\ell = 0} c^{\star}_{\ell} G^{(\ell), \model}(z).
  \end{equation}
  By estimate \eqref{eq:bound_G^ell_model}, We also have $\limsup_{\ell \to \infty} \lvert c^{\star}_{\ell} \rvert^{1/\ell} \leq 1/R$ and then $f^{\star}(z) \in V_{\alpha}(R)$. By Lemma \ref{lem:GH_orthogonal}, it is straightforward to check that if the series representation \eqref{eq:f_V_alpha(R)_series_rep} exists, then $c_{\ell} = c^{\star}_{\ell}$. On the other hand, for all $\ell \in \natN$,
\begin{equation}
  \langle f(z) - f^{\star}(z), H^{(-\ell), \model} \rangle_{V_{\alpha}(R)} = 0.
\end{equation}
We thus have, by Lemma \ref{lem:GH_complete}, that $f(z) - f^{\star}(z) = 0$, and construct the series representation \eqref{eq:f_V_alpha(R)_series_rep} by taking $c_{\ell} = c^{\star}_{\ell}$.
\end{proof}

\subsection{Operators $P^{\model}$, $Q^{\model}$, $\P^{\model}$, and $\Q^{\model}$}

It is worth comparing the function spaces $V_{\alpha}(R)$ and $\V_{\alpha}(R)$ with the function space $H(R)$ defined in \eqref{eq:func_space_H}. It is a basic fact in complex analysis that any function $h(z) \in H(R)$ in has a Taylor series representation
\begin{equation} \label{eq:Taylor_H(R)}
  h(z) = \sum^{\infty}_{\ell = 0} a_{\ell} z^{\ell}, \quad \limsup_{\ell \to \infty} \lvert a_\ell \rvert^{\frac{1}{\ell}} \leq \frac{1}{R}.
\end{equation}
In $V_{\alpha}(R)$ (resp.~$\V_{\alpha}(R)$), the series representation \eqref{eq:f_V_alpha(R)_series_rep} (resp.~\eqref{eq:f_V_alpha(R)_series_rep_tilde}) is the analogue of \eqref{eq:Taylor_H(R)}.

We define the operator $P^{\model}$ from $H(R)$ to $V_{\alpha}(R)$, and the operator $\P^{\model}$ from $H(R)$ to $\V_{\alpha}(R)$, as follows. Let $h(z) \in H(R)$ with Taylor expansion \eqref{eq:Taylor_H(R)}. For any $z \in D(0, R)$ and $\lvert z \rvert < R' < R$, let
\begin{align}
  P^{\model}(h)(z) = {}& \frac{1}{2\pi i} \oint_{\lvert w \rvert = R'} h(w) \sum^{\infty}_{\ell = 0} w^{-\ell} G^{(\ell), \model}(z) \frac{dw}{w} = \sum^{\infty}_{\ell = 0} a_{\ell} G^{(\ell), \model}(z), \\
  \P^{\model}(h)(z) = {}& \frac{1}{2\pi i} \oint_{\lvert w \rvert = R'} h(w) \sum^{\infty}_{\ell = 0} w^{-\ell} \G^{(\ell), \model}(z) \frac{dw}{w} = \sum^{\infty}_{\ell = 0} a_{\ell} \G^{(\ell), \model}(z).
\end{align}
We also define the operator $Q^{\model}$ from $V_{\alpha}(R)$ to $H(R)$, and the operator $\Q^{\model}$ from $\V_{\alpha}(R)$ to $H(R)$ as follows. Let $f(z) \in V_{\alpha}(R)$ and $\f(z) \in \V_{\alpha}(R)$, and they have the series representations in \eqref{eq:f_V_alpha(R)_series_rep} and \eqref{eq:f_V_alpha(R)_series_rep_tilde} by Lemma \ref{lem:V_alpha(R)_series_rep}. For any $z \in D(0, R)$ and $\lvert z \rvert < R' < R$,
\begin{align}
  Q^{\model}(f)(z) = {}& \frac{1}{2\pi i} \oint_{\lvert w \rvert = R'} f(w) \sum^{\infty}_{\ell = 0} H^{(-\ell), \model}(w) z^{\ell} \frac{dw}{w} = \sum^{\infty}_{{\ell} = 0} c_{\ell} z^{\ell}, \\
  \Q^{\model}(\f)(z) = {}& \frac{1}{2\pi i} \oint_{\lvert w \rvert = R'} \f(w) \sum^{\infty}_{\ell = 0} \H^{(-\ell), \model}(w) z^{\ell} \frac{dw}{w} = \sum^{\infty}_{{\ell} = 0} \c_{\ell} z^{\ell}.
\end{align}

\begin{lemma} \label{lem:PQ_id}
  Suppose $f(z) \in V_{\alpha}(R)$, $f(z) \in \V_{\alpha}(R)$ and $h(z) \in H(R)$. Then we have
  \begin{align}
    Q^{\model}(P^{\model}(h))(z) = {}& h(z) && \text{and} & \Q^{\model}(\P^{\model}(h))(z) = {}& h(z), \label{eq:lem:PQ_id_1} \\
    P^{\model}(Q^{\model}(f))(z) = {}& f(z) && \text{and} & \P^{\model}(\Q^{\model}(\f))(z) = {}& \f(z). \label{eq:lem:PQ_id_2}
  \end{align}
  Hence, $Q^{\model}P^{\model} = \Q^{\model} \P^{\model} = I$ as an operator on $H(R)$, $P^{\model}Q^{\model} = I$ as an operator on $V_{\alpha}(R)$ and $\P^{\model} \Q^{\model} = I$ as an operator on $\V_{\alpha}(R)$. Moreover, $P^{\model}Q^{\model}$ and $\P^{\model} \Q^{\model}$ have the reproducing kernel representation such that with $\lvert z \rvert < R' < R$,
  \begin{align}
    f(z) = P^{\model}(Q^{\model}(f))(z) = {}& \frac{1}{2\pi i} \oint_{\lvert w \rvert = R'} f(w) \sum^{\infty}_{k = 0} G^{(\ell), \model}(z) H^{(-\ell), \model}(z) \frac{dw}{w}, \\
    \f(z) = \P^{\model}(\Q^{\model}(\f))(z) = {}& \frac{1}{2\pi i} \oint_{\lvert w \rvert = R'} \f(w) \sum^{\infty}_{k = 0} \G^{(\ell), \model}(z) \H^{(-\ell), \model}(z) \frac{dw}{w}.
  \end{align}
\end{lemma}

\begin{proof}
  We check the $Q^{\model}(P^{\model}(h))(z)$ part of \eqref{eq:lem:PQ_id_1}, and the $\Q^{\model}(\P^{\model}(h))(z)$ is analogous. We omit \eqref{eq:lem:PQ_id_2} since it is parallel to \eqref{eq:lem:PQ_id_1}. We consider,
  \begin{equation}
    \begin{split}
      Q^{\model}(P^{\model}(h))(z) = {}& \frac{1}{2\pi i} \oint_{\lvert w \rvert = R'} \sum^{\infty}_{k = 0} a_k G^{(k)}(w) \sum^{\infty}_{\ell = 0} H^{(-\ell)}(w) z^{\ell} \frac{dw}{w} \\
      = {}& \sum^{\infty}_{k = 0} a_k \left( \frac{1}{2\pi i} \oint_{\lvert w \rvert = R'} G^{(k)}(w) \sum^{\infty}_{k = 0} H^{(-\ell)}(w) z^{\ell} \frac{dw}{w} \right) \\
      = {}& \sum^{\infty}_{k = 0} a_k z^k = h(z).
    \end{split}
  \end{equation}
\end{proof}

In the special case that $\theta = p/q$ is rational, the operators $P^{\model}, Q^{\model}, \P^{\model}, \Q^{\model}$ are finite rank operators. To see it, we note that
\begin{align}
  X^{(\ell + k(p + q)), \model}(z) = {}& X^{(\ell), \model}(z) z^{k(p + q)}, & X = {}& G, H, \G, \H.
\end{align}
For $h(z) \in H(R)$, with $\omega = e^{2\pi i/(p + q)}$,
\begin{equation} \label{eq:P^(0)_pre_rational}
  \begin{split}
    P^{\model}(h)(z) = {}& \frac{1}{2\pi i} \oint_{\lvert w \rvert = R'} h(w) \sum^{p + q - 1}_{\ell = 0} w^{-\ell} G^{(\ell), \model}(z) \left( \sum^{\infty}_{k = 0} \frac{z^{k(p + q)}}{w^{k(p + q)}} \right) \frac{dw}{w} \\
    = {}& \frac{1}{2\pi i} \oint_{\lvert w \rvert = R'} h(w) \sum^{p + q - 1}_{\ell = 0} w^{-\ell} G^{(\ell), \model}(z) \prod^{p + q - 1}_{k = 0} \frac{w}{w - \omega^k z} \frac{dw}{w} \\
    = {}& \frac{1}{p + q} \sum^{p + q - 1}_{k = 0} h(\omega^k z) \sum^{p + q - 1}_{\ell = 0} (\omega^k z)^{-\ell} G^{(\ell), \model}(\omega^k z).
  \end{split}
\end{equation}
In the last identity, we make use of the residue calculation and the identity that $\prod^{p + q - 1}_{k = 1} (1 - \omega^k) = p + q$. Also for $f(z) \in V_{\alpha}(R)$,
\begin{equation} \label{eq:Q^(0)_pre_rational}
  \begin{split}
    Q^{\model}(f)(z) = {}& \frac{1}{2\pi i} \oint_{\lvert w \rvert = R'} f(w) \sum^{p + q - 1}_{\ell = 0} H^{(-\ell), \model}(w) z^{\ell} \left( \sum^{\infty}_{k = 0} \frac{z^{k(p + q)}}{w^{k(p + q)}} \right) \frac{dw}{w} \\
    % = {}& \frac{1}{2\pi i} \oint_{\lvert w \rvert = R'} f(w) \sum^{p + q - 1}_{\ell = 0} H^{(-\ell), \model}(w) z^{\ell} \prod^{p + q - 1}_{k = 0} \frac{w}{w - \omega^k z} \frac{dw}{w} \\
    = {}& \frac{1}{p + q} \sum^{p + q - 1}_{k = 0} f(\omega^k z) \sum^{p + q - 1}_{\ell = 0} z^{\ell} H^{(-\ell), \model}(\omega^k z).
  \end{split}
\end{equation}
Similar formulas hold for $\P^{\model}$ and $\Q^{\model}$.

\section{Asymptotic analysis of the RH problem for $Y$} \label{sec:asy_RH_Y}

To derive the asymptotic behaviour of  $p_n$ as $n \to \infty$, one can perform a Deift/Zhou steepest descent analysis of RH problem \ref{RHP:original_p} for $Y$. This method is carried out in \cite[Section 3]{Wang-Zhang21} when $\theta$ is an integer, and in this section we will carry it out for all real $\theta > 0$. Several steps of this method involve only in the $g$-functions, the global parametrix, and the local parametrix around $b$. These steps are applicable for all $\theta > 0$, and they are accomplished in \cite[Sections 3.1 -- 3.5]{Wang-Zhang21}, although \cite{Wang-Zhang21} focuses on integer valued $\theta$. we only recall the formulas there that are to be used in later computations in the current paper and omit the detailed derivations. The construction of the local parametrix around $0$, however, is new.

Later in this paper, we denote
\begin{align} \label{eq:value_r_n}
  r_n := {}& n^{-\frac{2}{m_{\theta} + 1}}, && \text{and} & \varrho = {}& (\theta + 1)\rho = \left. (1 + \theta^{-1}) d_1 \pi \middle/ \sin\left(\frac{\pi}{1 + \theta}\right), \right.
\end{align}
where $m_{\theta}$ and $\rho$ are given in \eqref{eq:defn_rho}, and let $\rb$ be a small constant such that $f_b(z)$ in \eqref{def:fb} maps $D(b, \rb)$ conformally to an open region $D_b$ containing $0$. 

\subsection{Transforms of RH problem for $Y$ from \cite{Wang-Zhang21} that are valid for all $\theta > 0$} \label{sec:asy_RH_Y_away_0}

\paragraph{Transform from $Y$ to $Q$}

Let $\Sigma_1 \subseteq \compC_+ \cap \halfH$ be a contour from $0$ to $b$ whose shape to be fixed as follows:
\begin{enumerate}
\item 
  In the disk $D(0, (10 r_n)^{1 + \theta^{-1}})$ where $r_n$ is defined in \eqref{eq:value_r_n}, $\Sigma_1$ is the straight line segment connecting $0$ and $(10 r_n)^{1 + \theta^{-1}} e^{\gamma i}$, where $\gamma$ is a small constant (cf.~Section \ref{sec:V_alpha_R}).
\item
  Let $\sigma_b \in \partial D(b, r^{(b)})$ such that $f_b(\sigma_b)$ is the intersection of the ray $\{ \arg z = \frac{2\pi}{3} \}$ with $\partial D_b$. In the disk $D(b, \rb)$, $\Sigma_1$ is $f^{-1}_b(\{ \arg z = \frac{2\pi}{3} \} \cap D_b)$, a curve connecting $\sigma_b$ and $b$.
\item
  We denote the part of $\Sigma_1$ between $(10 r_n)^{1 + \theta^{-1}} e^{\gamma i}$ and $\sigma_b$ as $\Sigma^R_1(10 r_n)$ (see \eqref{def:SigmaiR}). We require that $\Sigma^R_1(10 r_n)$ is a curve lying in $\compC_+ \cap \halfH$, such that
  \begin{equation} \label{eq:Sigma_shape_fixed}
    \Re \phi(z) > \epsilon r_n \quad \text{for all $z \in \Sigma^R_1(10 r_n)$}
  \end{equation}
  For some $\epsilon > 0$. The existence of such $\Sigma^R_1(10 r_n)$ is guaranteed by \eqref{eq:ineq:psi} that controls $\Re \phi(z)$ as $z$ is above and close to $(0, b)$, and the estimates \eqref{eq:g_error} and \eqref{eq:g_tilde_error} that control $\Re \phi(z)$ as $z$ is close to $0$.
\end{enumerate}
Then let $\Sigma_2 = \overline{\Sigma_1} \subseteq \compC_- \cap \halfH$. We call the region enclosed by $\Sigma_1$ and $\Sigma_2$ the ``lens'', such that $\realR$ divides it into upper and lower parts. Then we denote the contour
\begin{equation}\label{def:Sigma}
  \Sigma:= [0,+\infty) \cup \Sigma_1 \cup \Sigma_2,
\end{equation}

as shown in Figure \ref{fig:Sigma}, and the orientation of $\Sigma$ is also specified there.

Let $Y = (Y_1, Y_2)$ be the solution to RH problem \ref{RHP:original_p}. As shown in \cite[Sections 3.1--3.4]{Wang-Zhang21}, after the explicit and invertible transformations $Y \to T \to S \to Q$, we derive (see \cite[Equations (3.1), (3.6) and (3.29)]{Wang-Zhang21})
\begin{multline} \label{def:thirdtransform}
  Q(z) = (Q_1(z), Q_2(z)) = (Y_1(z) e^{-n g(z)}, Y_2(z) e^{n (\g(z) - \ell)}) \\
  \times
  \begin{cases}
    \begin{pmatrix}
      P^{(\infty)}_1(z) & 0 \\
      0 & P^{(\infty)}_2(z)
    \end{pmatrix}^{-1},
    & \text{$z$ outside the lens}, \\
    \begin{pmatrix}
      1 & 0 \\
      \frac{\theta}{z^{\alpha+1-\theta}} e^{-n \phi(z)} & 1
    \end{pmatrix}
    \begin{pmatrix}
      P^{(\infty)}_1(z) & 0 \\
      0 & P^{(\infty)}_2(z)
    \end{pmatrix}^{-1},
    & \text{$z$ in the lower part of the lens}, \\
    \begin{pmatrix}
      1 & 0 \\
      -\frac{\theta}{z^{\alpha+1-\theta}} e^{-n \phi(z)} & 1
    \end{pmatrix}
    \begin{pmatrix}
      P^{(\infty)}_1(z) & 0 \\
      0 & P^{(\infty)}_2(z)
    \end{pmatrix}^{-1},
    & \text{$z$ in the upper part of the lens}.
  \end{cases}
\end{multline}
Here $g(z)$, $\g(z)$ and $\phi(z)$ are defined in \eqref{def:g}, \eqref{def:tildeg} and \eqref{def:phi}, and $P^{(\infty)}_1(z)$ and $P^{(\infty)}_2(z)$ are defined in \eqref{eq:P1} and \eqref{eq:P2}.

\begin{figure}[t]
  \centering
  \includegraphics{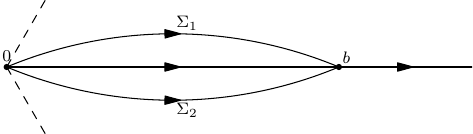}
  \caption{The contour $\Sigma$. The dashed lines indicate $\halfH$.}
  \label{fig:Sigma}
\end{figure}

\paragraph{Local parametrix around $b$}

Let, with $P^{(\infty)}_1$ and $P^{(\infty)}_1$ defined in \eqref{eq:P1} and \eqref{eq:P2}, (\cite[Equation (3.37)]{Wang-Zhang21})
\begin{equation}\label{def:gib}
  g^{(b)}_1(z) = \frac{z^{(\theta - \alpha - 1)/2}}{P_1^{(\infty)}(z)}, \qquad g^{(b)}_2(z) = \frac{z^{(\alpha + 1 - \theta)/2}}{\theta P_2^{(\infty)}(z)},
\end{equation}
and let, for $z \in D(0, \rb) \setminus \Sigma$, (\cite[Equation (3.43)]{Wang-Zhang21})
\begin{equation}\label{def:Eb}
  E^{(b)}(z)=\frac{1}{\sqrt{2}}
  \begin{pmatrix}
    g^{(b)}_1(z) & 0  \\
    0 & g^{(b)}_2(z)
  \end{pmatrix}^{-1}
  e^{\frac{\pi i}{4} \sigma_3}
  \begin{pmatrix}
    1 & -1 \\
    1 & 1
  \end{pmatrix}
  \begin{pmatrix}
    n^{\frac16} f_b(z)^{\frac14} & 0 \\
    0 & n^{-\frac16} f_b(z)^{-\frac14}
  \end{pmatrix}.
\end{equation}
The local parametrix $P^{(b)}(z)$ is the $2 \times 2$ matrix-valued function (\cite[Equations (3.38) and (3.42)]{Wang-Zhang21})
\begin{equation} \label{def:Pb}
  P^{(b)}(z) := E^{(b)}(z) \Psi^{(\Ai)}(n^{\frac23}f_b(z))
  \begin{pmatrix}
    e^{-\frac{n}{2} \phi(z)} g^{(b)}_1(z) & 0
    \\
    0 & e^{\frac{n}{2} \phi(z)} g^{(b)}_2(z)
  \end{pmatrix},  \qquad
  z \in D(b,\rb) \setminus \Sigma,
\end{equation}
where $\Psi^{(\Ai)}$ is the Airy parametrix defined in \eqref{eq:defn_Psi_Airy}. We have the following properties of $P^{(b)}(z)$ (\cite[RH problem 3.9(3), (4)]{Wang-Zhang21}):
\begin{prop} \label{prop:Pb_asy}
  As $z \to b$, $(P^{(b)}(z))_{ij} = \bigO((z - b)^{-1/4})$ and $(P^{(b)}(z)^{-1})_{ij} = \bigO((z - b)^{-1/4})$, for $i, j = 1, 2$. For $z$ on the boundary $\partial D(b,\rb)$ (except for the intersecting points with $\Sigma$), we have, as $n \to \infty$, $P^{(b)}(z) = I + \bigO(n^{-1})$ uniformly.
\end{prop}
This proposition follows the limit properties of $P^{(\infty)}_1$ and $P^{(\infty)}_2$ in \eqref{eq:Pinfty_asy_b}, and of the Airy function (see \cite[Chapter 9]{Boisvert-Clark-Lozier-Olver10}). We omit the proof here, since the detail is given in \cite[Section 3.5]{Wang-Zhang21}.

Next, from $Q(z)$ defined in \eqref{def:thirdtransform}, we define a vector-valued function $V^{(b)}(z) = (V^{(b)}_1(z), V^{(b)}_1(z))$ by (\cite[Equation (3.46)]{Wang-Zhang21})
\begin{equation} \label{eq:defn_V^(b)}
  V^{(b)}(z) = Q(z) P^{(b)}(z)^{-1}, \quad z \in D(b,\rb) \setminus \Sigma.
\end{equation}

\paragraph{Definition and properties of $R^{\pre}$}

We define, for $r \geq 0$, 
\begin{equation}\label{def:SigmaiR}
  \Sigma^R_i(r) :=\Sigma_i \setminus \{D(0, r^{1 + \theta^{-1}}) \cup D(b, \rb) \} \quad \text{and} \quad \Sigma^R_i = \Sigma^R_i(r^{1 + \theta^{-1}}_n), \quad i=1, 2,
\end{equation}
where $D(0, 0) = \emptyset$, and
\begin{equation}\label{def:sigmapre}
 \Sigma^{\pre}:=[0,b]\cup [b+\rb, \infty) \cup \partial D(b,\rb) \cup \Sigma_1^R(0) \cup \Sigma_2^R(0).
\end{equation}
With $Q = (Q_1, Q_2)$ defined in \eqref{def:thirdtransform} and $V^{(b)} = (V^{(b)}_1, V^{(b)}_2)$ in \eqref{eq:defn_V^(b)}, we define the $1 \times 2$ array of functions $R^{\pre} = (R^{\pre}_1, R^{\pre}_2)$ on $(\compC \setminus \Sigma^{\pre}, \halfH \setminus \Sigma^{\pre})$ by
\begin{equation} \label{eq:defn_R^pre}
  \begin{aligned}
    R^{\pre}_1(z) = {}& Q_1(z), && z \in \compC \setminus (\Sigma^{\pre} \cup D(b, \rb)), \\
    R^{\pre}_2(z) = {}& Q_2(z), && z \in \halfH \setminus (\Sigma^{\pre} \cup D(b, \rb)), \\
    (R^{\pre}_1, R^{\pre}_2) = {}& (V^{(b)}_1, V^{(b)}_2), && \text{$R_1(z)$ and $R_2(z)$ on $D^*(b, \rb) \setminus (b - \rb, b)$}.
  \end{aligned}
\end{equation}
Since $R^{\pre}$ is transformed from $Y$, it satisfies the following RH problem that is derived from RH problem \ref{RHP:original_p}:

\begin{RHP} \label{RHP:Svar} \hfill
  \begin{enumerate}[label=\emph{(\arabic*)}, ref=(\arabic*)]
  \item \label{enu:RHP:Svar:1}
    $R^{\pre} = (R^{\pre}_1, R^{\pre}_2)$ is analytic in $(\compC \setminus \Sigma^{\pre}, \mathbb{H}_\theta \setminus \Sigma^{\pre})$, and is continuous up to the boundary, except for $0$.
  \item \label{enu:RHP:Svar:2}
    For $z \in (0, b) \cup \Sigma^R_1(0) \cup \Sigma^R_2(0) \cup (b + \rb, +\infty)$, we have that $(R^{\pre}_1)_{\pm}(z)$ and $(R^{\pre}_2)_{\pm}(z)$ are bounded for $z$ away from $0$, and
    \begin{equation} \label{def:Jcals}
      R^{\pre}_+(z) = R^{\pre}_-(z) J_Q(z),
    \end{equation}
    where (\cite[Equation (3.31)]{Wang-Zhang21})
    \begin{equation} \label{def:JQ}
      J_Q(z) =
      \begin{cases}
        \begin{pmatrix}
          1 & 0 \\
          \frac{\theta P^{(\infty)}_2(z)}{z^{\alpha+1-\theta}P^{(\infty)}_1(z)}e^{-n \phi(z)} & 1
        \end{pmatrix},
        & z \in \Sigma^R_1(0) \cup \Sigma^R_2(0), \\
        \begin{pmatrix}
          0 & 1 \\
          1 & 0
        \end{pmatrix},
        & z \in (0, b), \\
        \begin{pmatrix}
          1 & \frac{z^{\alpha+1-\theta}P^{(\infty)}_1(z)}{\theta P^{(\infty)}_2(z)}e^{n \phi(z)}  \\
          0 & 1
        \end{pmatrix},
        & z \in (b + \rb, +\infty),
      \end{cases}
    \end{equation}
    and for $z \in \partial D(b, \rb)$,
    \begin{equation}
      R^{\pre}_+(z) = R^{\pre}_-(z) P^{(b)}(z).
    \end{equation}
  \item \label{enu:RHP:Svar:3}
    As $z \to \infty$ in $\compC$, $R^{\pre}_1$ behaves as $R^{\pre}_1(z)=1+\bigO(z^{-1})$.

  \item \label{enu:RHP:Svar:4}
    As $z \to \infty$ in $\mathbb{H}_\theta $, $R^{\pre}_2$ behaves as $R^{\pre}_2(z)=\bigO(1)$.

  \item
  \label{enu:RHP:Svar:5}
    As $z \to 0$ in $\compC \setminus \Sigma^{\pre}$, we have
    \begin{equation} \label{eq:Q_to_0:1}
      R^{\pre}_1(z) =
      \begin{cases}
        \bigO ( z^{\frac{\theta(\theta - \alpha - 1/2)}{1 + \theta}} ), & \text{$\alpha > \theta - 1$ and $z$ inside the lens}, \\
        \bigO ( z^{\frac{\theta}{2(1+\theta)}} \log z ), & \text{$\alpha = \theta - 1$ and $z$ inside the lens}, \\
        \bigO ( z^{\frac{\alpha + 1 - \theta/2}{1+\theta}} ), & \text{$z$ outside the lens or $-1 < \alpha < \theta - 1$}.
      \end{cases}
    \end{equation}

  \item
    As $z \to 0$ in $\halfH \setminus \Sigma^{\pre}$, we have
    \begin{equation} \label{eq:Q_to_0:2}
      R^{\pre}_2(z)=
      \begin{cases}
        \bigO ( z^{\frac{\theta(\theta-\alpha-1/2)}{1+\theta}} ), & \alpha > \theta - 1, \\
        \bigO ( z^{\frac{\theta}{2(1+\theta)}} \log z ), & \alpha = \theta - 1, \\
        \bigO (z^{\frac{\alpha+1-\theta/2}{1+\theta}} ), & -1 < \alpha < \theta - 1.
      \end{cases}
    \end{equation}
  \item
    As $z \to b$, we have $R^{\pre}_1(z) = \bigO(1)$ and $R^{\pre}_2(z) = \bigO(1)$.
  \item \label{enu:RHP:Svar:7}
    For $x>0$, we have the boundary condition $R^{\pre}_2(e^{\pi i/\theta}x) = R^{\pre}_2(e^{-\pi i/\theta}x)$.
  \end{enumerate}
\end{RHP}
By the regularity assumption in Section \ref{subsec:regularity}, we have that that there exists $\epsilon > 0$ such that for all large enough $n$,
\begin{align}
  \lvert (J_Q)_{12}(z) \rvert < {}& e^{-\epsilon n}, & z \in {}& (b + \rb, +\infty), \label{eq:est_JQ_12} \\
  \lvert (J_Q)_{21}(z) \rvert < {}& e^{-\epsilon n r_n}, & z \in {}& (\Sigma_1 \cup \Sigma_2) \setminus (D(0, r^{1 + \theta^{-1}}_n) \cup D(0, \rb)). \label{eq:est_JQ_21} \\
  \lvert (J_Q)_{21}(z) \rvert \to {}& 0 \text{ exponentially fast}, & z \to {}& \infty \text{ from $\realR_+$}. \label{eq:est_JQ_21_vanishing}
\end{align}
(More specifically, \eqref{eq:est_JQ_12} is from \eqref{eq:gequal2}, and \eqref{eq:est_JQ_21} and \eqref{eq:est_JQ_21_vanishing} are from \eqref{eq:Sigma_shape_fixed} and the estimates of $g(z)$ and $\g(z)$ in \eqref{eq:g_error}, \eqref{eq:g_tilde_error}, \eqref{eq:asy_formula_for_g} and \eqref{eq:asy_formula_for_g_tilde}.) We also have, from the estimates of entries of $P^{(b)}(z)$ that are based on the asymptotics of Airy function and the limit formulas \eqref{eq:g_error}, \eqref{eq:g_tilde_error}, \eqref{eq:asy_formula_for_g} and \eqref{eq:asy_formula_for_g_tilde} of $g(z)$ and $\g(z)$, and the limit formulas \eqref{eq:P_and_P^pre}, \eqref{eq:asy_P^infty_1} and \eqref{eq:asy_P^infty_2} of $P^{(\infty)}_1(z)$ and $P^{(\infty)}_2(z)$, the estimate that is uniform for all large enough $n$ that for $z \in \partial D(b, \rb)$
\begin{align}
  \lvert (P^{(b)})_{11}(z) - 1 \rvert = {}& \bigO(n^{-1}), & \lvert (P^{(b)})_{12}(z) \rvert = {}& \bigO(n^{-1}), \label{eq:est_Pb_row1} \\
  \lvert (P^{(b)})_{21}(z) \rvert = {}& \bigO(n^{-1}), & \lvert (P^{(b)})_{22}(z) - 1 \rvert = {}& \bigO(n^{-1}). \label{eq:est_Pb_row2}
\end{align}
Since the derivation of \eqref{eq:est_Pb_row1} and \eqref{eq:est_Pb_row2} is given in \cite[Section 3.5]{Wang-Zhang21}, we omit further details.

\begin{prop} \label{prop_uniqueness}
  RH problem \ref{RHP:Svar} has a unique solution.
\end{prop}
The proof of Proposition \ref{prop_uniqueness} is contained in \cite[Proofs of Propositions 3.3 and 3.20]{Wang-Zhang21}. For completeness, we give a proof below.
\begin{proof}[Proof of Proposition \ref{prop_uniqueness}]
  The biorthogonal polynomial $p_n(z) = p^{(V)}_{n, n}(z)$ always exists since it is the average characteristic function of the Muttalib-Borodin ensemble, as explained in \cite[Equation (1.14)]{Wang-Zhang21}. It is straightforward to check that $Y = (p_n, Cp_n)$ is a solution of RH problem \ref{RHP:original_p}. Since RH problem \ref{RHP:Svar} is transformed from RH problem \ref{RHP:original_p}, we conclude that at least it has one solution that is transformed from $Y = (p_n, Cp_n)$.

  On the other hand, the uniqueness of RH problem \ref{RHP:Svar} is, by taking the transform \eqref{def:thirdtransform} reversely from $Q$ to $Y$, is equivalent to the uniqueness of RH problem \ref{RHP:original_p} with the boundary condition ``$Y$ has continuous boundary values $Y_{\pm}$ when approaching $(0, +\infty)$ from above $(+)$ and below $(-)$'' replaced by the weaker one ``$Y$ has continuous boundary values $Y_{\pm}$ when approaching $(0, b) \cup (b, +\infty)$ from above $(+)$ and below $(-)$, and $Y_1(z) = \bigO((z - b)^{-1/4})$ and $Y_2(z) = \bigO((z - b)^{-1/4})$ as $z \to b$'', and the jump condition \eqref{eq:jump_Y_realR} holds for $x \in (0, b) \cup (b, +\infty)$. Since $Y_1(z)$ has a trivial jump along $(0, b) \cup (b, +\infty)$, and it can be extended to an analytic function on $\compC \setminus \{ 0, b \}$. Then since $Y_1(z) = \bigO(1)$ as $z \to 0$ and $Y_1(z) = \bigO((z - b)^{-1/4})$ as $z \to b$, we conclude that $Y_1(z)$ extends analytically to a holomorphic function on $\compC$, and it is a monic polynomial of degree $n$ due to Item \ref{enu:RHP:original_p:3} of RH problem \ref{RHP:original_p}. Similarly, we consider $Y_3(z) := Y_2(z^{1/\theta}) - CY_1(z^{1/\theta})$, and find that it has a trivial jump along $(-\infty, 0) \cup (0, b) \cup (b, +\infty)$, so it extends to an analytical function on $\compC \setminus \{ 0, b^{1/\theta} \}$. Also we have $Y_3(z) = \bigO((z - b^{1/\theta})^{-1/4})$ as $z \to b^{1/\theta}$, and from \eqref{eq:Y_2_limit_at_0} we have
  \begin{equation}
    Y_3(z) =
    \begin{cases}
      \bigO(1), & \text{$\alpha+1-\theta > 0$,} \\
      \bigO(\log z), & \text{$\alpha+1-\theta = 0$,} \\
      \bigO(z^{(\alpha+1)/\theta - 1}), & \text{$\alpha+1-\theta < 0$},
    \end{cases}
  \end{equation}
  as $z \to 0$. Hence $Y_3(z)$ is holomorphic on $\compC$, and by Item \ref{enu:RHP:original_p:4} of RH problem \ref{RHP:original_p} we have that $Y_3(z) = 0$, that is, $Y_2(z) = CY_1(z)$. Now we see from Item \ref{enu:RHP:original_p:4} of RH problem \ref{RHP:original_p} that $Y_1(z)$ is the biorthogonal polynomial $p_n$, and then $Y = (p_n, Cp_n)$. We conclude that the solution of RH problem \ref{RHP:Svar} is unique.
\end{proof}

\subsection{Local parametrix around $0$} \label{subsec:local_para_0}

\subsubsection{Transformation of $Q$ into function space $V^{(n), \dressing}_{\alpha}(r)$}

We define a transform $\transform: (X_1, X_2) \mapsto Y$, that maps a pair of functions $(X_1, X_2)$ where $X_1(z)$ is defined on $\compC \setminus \realR_+$ and $X_2(z)$ is defined on $\halfH \setminus \realR_+$, to a function $Y(z)$ defined on $\compC \setminus (\{ 0 \} \cup \{ \arg z = 0,  \pm \frac{\pi}{\theta + 1} \})$, in the way that
\begin{equation} \label{eq:defn_trans_T}
  Y(z) =
  \begin{cases}
    X_1(-(-z)^{1 + \theta^{-1}}), & \arg z \in (-\frac{\theta \pi}{1 + \theta}, \frac{\theta \pi}{1 + \theta}), \\
    X_2(z^{1 + \theta^{-1}} e^{\frac{\pi}{\theta} i}), & \arg z \in (\frac{-\pi}{1 + \theta}, 0), \\
    X_2(z^{1 + \theta^{-1}} e^{-\frac{\pi}{\theta}i}), & \arg z \in (0, \frac{\pi}{1 + \theta}).
  \end{cases}
\end{equation}
It is clear that $\transform^{-1}$ is also well defined. For instance, let $\id: z \to z$ denote the identity function on $\compC$ (and also on $\compC \setminus (\{ 0 \} \cup \{ \arg z = 0,  \pm \frac{\pi}{\theta + 1} \})$), then
\begin{equation} \label{eq:defn_X^id}
  \transform^{-1}(\id) = (X^{\id}_1, X^{\id}_2), \quad \text{where} \quad X^{\id}_1(z) = -(-z)^{\frac{\theta}{\theta + 1}}, \quad X^{\id}_2(z) =
  \begin{cases}
    z^{\frac{\theta}{\theta + 1}} e^{\frac{-\pi i}{\theta + 1}}, & \arg z \in (0, \frac{\pi}{\theta}), \\
    z^{\frac{\theta}{\theta + 1}} e^{\frac{\pi i}{\theta + 1}}, & \arg z \in (-\frac{\pi}{\theta}, 0).
  \end{cases}
\end{equation}

From the function $Q=(Q_1,Q_2)$ in \eqref{def:thirdtransform} that satisfies RH problem \ref{RHP:Svar}, we define a function $U(z)$ as
\begin{equation} \label{eq:defn_U_from_Q}
  U = \transform(Q_1, Q_2).
\end{equation}
Throughout this paper, we only consider $U(z)$ with $\lvert z \rvert < 10 r_n$. We note that in this region, due to the properties of $Q_1, Q_2$ that are summarized in Items \ref{enu:RHP:Svar:2} and \ref{enu:RHP:Svar:7} of RH problem \ref{RHP:Svar}, $U(z)$ extends analytically on rays $\{ \arg z = 0 \}$ and $\{ \arg z = \pm\pi/(1 + \theta) \}$, while it has jumps along the rays $\{ \arg z = \pm \frac{\theta^{-1} \pi + \gamma}{1 + \theta^{-1}} \}$. See Figure \ref{fig:U_in_Q}.  Moreover, $U(z)$ is continuous up to the boundary at the two rays, and satisfies

\begin{figure}[htb]
  \begin{minipage}{0.45\linewidth}
    \centering
    \includegraphics{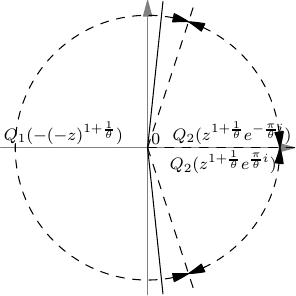}
    \caption[$U(z)$ in $Q_1, Q_2$.]{Expression of $U(z)$ in $Q_1(z)$ and $Q_2(z)$. $U(z)$ is discontinuous along the two solid rays.}
    \label{fig:U_in_Q}
  \end{minipage}
  \hspace{\stretch{1}}
  \begin{minipage}{0.45\linewidth}
    \centering
    \includegraphics{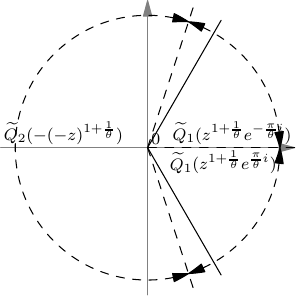}
    \caption[$\U(z)$ in $\Q_1, \Q_2$.]{Expression of $\U(z)$ in $\Q_1(z)$ and $\Q_2(z)$. $\U(z)$ is discontinuous along the two solid rays.}
    \label{fig:U_tilde_in_Q_tilde}
  \end{minipage}
\end{figure}

\begin{equation} \label{eq:U_jump}
  U_+(z) - U_-(z) =
  \begin{cases}
    J_U(z) U(z e^{\frac{2\pi}{1 + \theta^{-1}} i}), & \arg z = \frac{\theta^{-1} \pi + \gamma}{1 + \theta^{-1}}, \\
    J_U(z) U(z e^{-\frac{2\pi}{1 + \theta^{-1}} i}), & \arg z = \frac{-\theta^{-1}\pi - \gamma}{1 + \theta^{-1}}, 
  \end{cases}
\end{equation}
where, with $J_Q$ defined in \eqref{def:JQ},
\begin{equation} \label{def:JU}
  J_U(z) =
  \begin{cases}
    (J_Q(z^{1 + \theta^{-1}} e^{-\frac{\pi}{\theta} i}))_{21}, & \arg z = \frac{\theta^{-1} \pi + \gamma}{1 + \theta^{-1}}, \\
    (J_Q(z^{1 + \theta^{-1}} e^{\frac{\pi}{\theta} i}))_{21}, & \arg z = \frac{-\theta^{-1}\pi - \gamma}{1 + \theta^{-1}}.
  \end{cases}
\end{equation}
The limit behaviour of $U(z)$ as $z \to 0$ can be derived from that of $Q = R^{\pre}$ in \eqref{eq:Q_to_0:1} and \eqref{eq:Q_to_0:2}.

Recall the functions $g(z)$ and $\g(z)$ defined in \eqref{def:g} and \eqref{def:tildeg}, functions $g^{\pre}(z)$ and $\g^{\pre}(z)$ defined in \eqref{eq:asy_formula_for_g} and \eqref{eq:asy_formula_for_g_tilde}, functions $P_1^{(\infty)}(z)$ and $P_2^{(\infty)}(z)$ defined in \eqref{eq:P1} and \eqref{eq:P2}, and functions $P^{(\infty), \pre}_1(z)$ and $P^{(\infty), \pre}_2(z)$ defined in \eqref{eq:asy_P^infty_1} and \eqref{eq:asy_P^infty_2}. We denote
\begin{align} \label{eq:defn_n(z)}
  n_1(z) = {}& \frac{P^{(\infty), \pre}_1(z)}{P^{(\infty)}_1(z)} \frac{e^{n(-g(z) + V(z) + \ell - \g_-(0))}}{e^{-n g^{\pre}(z)}}, &
  n_2(z) = {}& \frac{P^{(\infty), \pre}_2(z)}{P^{(\infty)}_2(z)} \frac{e^{n(\g(z) - \g_+(0))}}{e^{n \g^{\pre}(z)}}.
  % \sqrt{\frac{\theta}{1 + \theta}} c^{\frac{2(\alpha + 1) - \theta}{2(1 + \theta)}} z^{\frac{\theta - 2(\alpha + 1)}{2(1 + \theta)}} \frac{e^{n(-g(z) + V(z) + \ell - \g_-(0))}}{P^{(\infty)}_1(z)} \\
  % \times
  % \begin{cases}
  %   e^{\frac{2(\alpha + 1) - \theta}{2(1 + \theta)} \pi i} e^{-(1 + \theta^{-1}) d_1 \pi \frac{e^{\frac{1}{1 + \theta} \pi i}}{\sin(\frac{\pi}{1 + \theta})} z^{\frac{\theta}{1 + \theta}}}, & \arg z \in (0, \pi), \\
  %   e^{\frac{\theta - 2(\alpha + 1)}{2(1 + \theta)} \pi i} e^{-(1 + \theta^{-1}) d_1 \pi \frac{e^{\frac{-1}{1 + \theta} \pi i}}{\sin(\frac{\pi}{1 + \theta})} z^{\frac{\theta}{1 + \theta}}}, & \arg z \in (-\pi, 0), 
  % \end{cases}
\end{align}
It is clear that $n_1(z)$ is well defined on $D^*(0, b) \setminus \realR_+$ and $n_2(z)$ is well defined on $(D^*(0, b) \cap \halfH) \setminus \realR_+$. By \eqref{eq:P^infty_jump}, \eqref{eq:gequal} and the properties of $P^{(\infty)}_1, P^{(\infty)}_2, P^{(\infty), \pre}_1, P^{(\infty), \pre}_2, \g, \g^{\pre}$, we have $n_1(z)$ and $n_2(z)$ are continuous up to the boundary rays and
\begin{align} \label{eq:m_n}
  (n_1)_+(x) = {}& (n_2)_-(x), & (n_2)_-(x) = {}& (n_1)_+(x), & n_2(x e^{\pi i/\theta}) = {}& n_2(x e^{-\pi i/\theta}), & \text{for $x \in (0, b)$}. 
\end{align}
Then we define the function $n(z)$ on $D^*(0, b^{\theta/(\theta + 1)}) \setminus \{ \arg z = 0, \pm \frac{\theta^{-1} \pi}{1 + \theta^{-1}} \}$ by
\begin{equation} \label{eq:func_n}
  n(z) = 
  \begin{cases}
    n_2(z^{1 + \theta^{-1}} e^{\frac{\pi i}{\theta}}), & \arg z \in (-\frac{\theta^{-1} \pi}{1 + \theta^{-1}}, 0), \\
    n_2(z^{1 + \theta^{-1}} e^{-\frac{\pi i}{\theta}}), & \arg z \in (0, \frac{\theta^{-1} \pi}{1 + \theta^{-1}}), \\
    n_1(-(-z)^{1 + \theta^{-1}}), & \arg (-z) \in (-\frac{\pi}{1 + \theta^{-1}}, \frac{\pi}{1 + \theta^{-1}}).
  \end{cases}
\end{equation}
We find that $n(z)$ is analytically extended on the rays $\{ \arg z = 0, \pm \frac{\theta^{-1} \pi + \gamma}{1 + \theta^{-1}} \}$ by \eqref{eq:m_n}, so $n(z)$ is analytic on $D^*(0, b^{\theta/(\theta + 1)})$. Then from the limit behaviours of  $P^{(\infty)}_1(z)$, $P^{(\infty), \pre}_1(z)$, $P^{(\infty)}_2(z)$, $P^{(\infty), \pre}_2(z)$, $g(z)$, $g^{\pre}(z)$, $\g(z)$ and $\g^{\pre}(z)$ as $z \to 0$ in \eqref{eq:P_and_P^pre}, \eqref{eq:asy_P^infty_1}, \eqref{eq:asy_P^infty_2}, \eqref{eq:g_error}, \eqref{eq:g_tilde_error}, \eqref{eq:asy_formula_for_g}, \eqref{eq:asy_formula_for_g_tilde}, we have that $n(z)$ has no singularity at $0$, and $n(z)$ is analytic on $D^*(0, b^{\theta/(\theta + 1)})$. Moreover, by the limit behaviours listed above, we have that 
\begin{equation} \label{eq:est_n(z)}
  n(z) - 1 =
  \begin{cases}
    \bigO(n^{1 - m_{\theta}}), & \lvert z \rvert \in [1, C], \\
    \bigO(n^{\frac{1 - m_{\theta}}{1 + m_{\theta}}}), & \lvert z \rvert \in [10^{-1} r_n, 10 r_n],
  \end{cases}
\end{equation}
uniformly for all large enough $n$, where $C > 1$ is a constant.

Let $f(z)$ be a function whose domain is in $D(0, b^{\theta/(\theta + 1)})$. We define the transforms $\Dressing$ and $\Dressinginv$ on $f(z)$ by
\begin{align} \label{eq:G_divided_dressed}
  \Dressing(f)(z) = {}& n(z) e^{\varrho nz} f(\varrho nz), & \Dressinginv(f)(z) = {}& n^{-1}(z) e^{-\varrho nz} f(\varrho nz),
\end{align}
where $\varrho$ is defined in \eqref{eq:value_r_n}. Suppose $R \in (0, 10 \rho n r_n)$ and $f(z)$ is a function defined on $D^*(R) \setminus \{ \arg z = \pm \frac{\theta^{-1} \pi + \gamma}{1 + \theta^{-1}} \}$. so that $\Dressing(f)(z)$ is a function on $D^*(0, (\varrho n)^{-1} R) \setminus \{ \arg z = \pm \frac{\theta^{-1} \pi + \gamma}{1 + \theta^{-1}} \}$. Hence for any $r \in (0, 10 r_n)$, the inverse transform $\Dressing^{-1}$ is well defined on functions on $D^*(0, r) \setminus \{ \arg z = \pm \frac{\theta^{-1} \pi + \gamma}{1 + \theta^{-1}} \}$. let
\begin{equation}
  V^{(n), \dressing}_{\alpha}(r) = \{ \Dressing(f)(z) : f(z) \in V_{\alpha}(\varrho n r) \}.
\end{equation}
From Definition \ref{defn:V_alpha} of function space $V_{\alpha}(R)$, we derive the definition of function space $V^{(n), \dressing}_{\alpha}(r)$ as follows.
\begin{defn} \label{defn:V_dressing}
  $V^{(n), \dressing}_{\alpha}(r)$ consists of functions $f(z)$ on $z \in D^*(0, r) \setminus \{ \arg z = \pm \frac{\theta^{-1} \pi + \gamma}{1 + \theta^{-1}} \}$, such that $r \in (0, 10 r_n)$ and
  \begin{itemize}
  \item
    $f(z)$ is analytic in the sector $\arg (-z) \in (\frac{-\pi + \gamma}{1 + \theta^{-1}}, \frac{\pi - \gamma}{1 + \theta^{-1}})$ and the sector $\arg z \in (\frac{-\theta^{-1} \pi - \gamma}{1 + \theta^{-1}}, \frac{\theta^{-1} \pi + \gamma}{1 + \theta^{-1}})$ separately, and $f(z)$ is continuous up to the boundary on the two rays $\{ \arg z = \pm \frac{\theta^{-1} \pi + \gamma}{1 + \theta^{-1}} \}$.
  \item
    Let the two rays be oriented from $0$ to $\infty$. The boundary values of $f$ on the sides of the two rays satisfy
    \begin{align}
      &
        \begin{aligned}[b]
          f_+(z) - f_-(z) = {}& -e^{-\frac{2\alpha + 3}{\theta + 1} \pi i} e^{\varrho nz(1 - e^{{\frac{-2\pi}{\theta + 1} i}})} \frac{n(z)}{n(z e^{-\frac{2\pi}{\theta + 1} i})} f(z e^{-\frac{2\pi}{\theta + 1} i}) \\
          = {}& J_U(z) f(z e^{-\frac{2\pi}{\theta + 1} i}), 
        \end{aligned} 
                              & \arg z = {}& \frac{\theta^{-1} \pi + \gamma}{1 + \theta^{-1}}, \label{eq:G^lk_jump_1} \\
      &
        \begin{aligned}[b]
          f_+(z) - f_-(z) = {}& e^{\frac{2\alpha + 3}{\theta + 1} \pi i} e^{\varrho nz(1 - e^{{\frac{2\pi}{\theta + 1} i}})} \frac{n(z)}{n(z e^{\frac{2\pi}{\theta + 1} i})} f(z e^{\frac{2\pi}{\theta + 1} i}) \\
          = {}& J_U(z) f(z e^{\frac{2\pi}{\theta + 1} i}), 
        \end{aligned} 
                              & \arg z = {}& \frac{-\theta^{-1} \pi - \gamma}{1 + \theta^{-1}}, \label{eq:G^lk_jump_2}
    \end{align}
    where $J_U$ defined in \eqref{def:JU}.
  \item
    As $z \to 0$, $f(z)$ has the limit behaviour depending on $\alpha$ and $\theta$ characterized by \eqref{eq:V_alpha(R)_at_0:1} -- \eqref{eq:V_alpha(R)_at_0:3}.
  \end{itemize}
\end{defn}

As $z \to 0$, the limit behaviour of $U(z)$ that can be found from the limit behaviour of $(Q_1, Q_2) = (R^{\pre}_1, R^{\pre}_2)$ in \eqref{eq:Q_to_0:1} and \eqref{eq:Q_to_0:2} via the transform formula \eqref{eq:defn_U_from_Q}. By comparing the discontinuity condition \eqref{eq:U_jump} and the limit behaviour at $0$ of $U$ and Definition \ref{defn:V_dressing}, we find that $U(z)$ belongs to $V^{(n), \dressing}(r)$ for all $r \in (0, 10 r_n)$.

\subsubsection{Functions $G^{(\ell)}$ and $H^{(\ell)}$, and operators $P^{(0)}$ and $Q^{(0)}$}

We apply transform $\Dressing$ to $G^{(\ell), \model}(z)$ and transform $\Dressinginv$ to $H^{(\ell), \model}(z)$, and get
\begin{align} \label{G^ell_from_model}
  G^{(\ell)}(z) = {}& (\varrho n)^{-\ell} \Dressing(G^{(\ell), \model})(z), & H^{(\ell)}(z) = {}& (\varrho n)^{-\ell} \Dressinginv(H^{(\ell), \model})(z).
\end{align}
$G^{(\ell)}$ is analytic on $D(0, b^{\theta/(\theta + 1)}) \setminus \{ \arg z = \pm \frac{\theta^{-1} \pi + \gamma}{1 + \theta^{-1}} \}$, continuous up to the boundary on the two rays, and if $\ell \in \natN$, then $G^{(\ell)}(z) \in V^{(n), \dressing}_{\alpha}(r)$ for all $r \in (0, 10 r_n)$. Similarly, $H^{(\ell)}(z)$ is analytic on $D(0, b^{\theta/(\theta + 1)}) \setminus \{ \arg z = \pm \frac{\theta^{-1} \pi - \gamma}{1 + \theta^{-1}} \}$, and is continuous up to the boundary on the two rays.

From the definitions \eqref{eq:G_divided_model}, \eqref{eq:defn_G^ell_model}, \eqref{eq:H_divided_model}, \eqref{eq:defn_H^ell_model}, \eqref{G^ell_from_model}, and the estimates \eqref{eq:asy_IJK_at_0}, \eqref{eq:G_divided} and \eqref{eq:est_n(z)}, we have the estimate that for any constant $C > 1$, if $\zeta \in D(0, C) \setminus (D(0, 1) \cup \{ \arg z = \pm \frac{\theta^{-1} \pi + \gamma}{1 + \theta^{-1}} \})$, then
\begin{equation} \label{eq:est_G^ell_central}
  \zeta^{-\ell} \left( (\varrho n)^{\ell} G^{(\ell)}((\varrho n)^{-1} \zeta) - G^{(\ell), \model}(\zeta) \right) = \bigO(n^{1 - m_{\theta}}),
\end{equation}
and if $\zeta \in D(0, 10 r_n) \setminus (D(0, 10^{-1} r_n) \cup \{ \arg z = \pm \frac{\theta^{-1} \pi + \gamma}{1 + \theta^{-1}} \})$, then 
  
\begin{align} \label{eq:G^ell_at_0:5}
  \lvert z^{-\ell} (G^{(\ell)}(z) - 1) \rvert \leq {}& M n^{\frac{1 - m_{\theta}}{1 + m_{\theta}}}, & \lvert z^{-\ell} (H^{(\ell)}(z) - 1) \rvert \leq {}& M n^{\frac{1 - m_{\theta}}{1 + m_{\theta}}}.
\end{align}

From the operators $P^{\model}: H(R) \to V_{\alpha}(R)$ and $Q^{\model}: V_{\alpha}(R) \to H(R)$, we define the operators $P^{(0)}: H(r) \to V^{(n), \dressing}_{\alpha}(r)$ and $Q^{(0)}: V^{(n), \dressing}_{\alpha}(r) \to H(r)$, with $r \in (0, 10 r_n)$, as follows. For any $h(z) \in H(r)$, we denote the function $h^{\sharp} \in H(\varrho n r)$ by $h^{\sharp}(z) = h((\varrho n)^{-1} z)$. For any $h(z) \in H(r)$ with $h(z) = \sum^{\infty}_{\ell = 0} a_{\ell} z^{\ell}$, we define, for $z \in D(0, r)$
\begin{equation} \label{eq:defn_P^0}
  P^{(0)}(h)(z) = \Dressing(P^{\model}(h^{\sharp}))(z) = \sum^{\infty}_{\ell = 0} a_{\ell} G^{(\ell)}(z) = \frac{1}{2\pi i} \oint_{\lvert w \rvert = r'} h(w) \sum^{\infty}_{\ell = 0} w^{-\ell} G^{(\ell)}(z) \frac{dw}{w},
\end{equation}
where $r' \in (\lvert z \rvert, r)$. On the other hand, for any $h(z) \in H(\varrho n r)$, we denote the function $h^{\flat} \in H(r)$ by $h^{\flat}(z) = h(\varrho n z)$. For any $f(z) \in V^{(n), \dressing}_{\alpha}(r)$, we have that $(\Dressing^{-1} f)(z) \in V_{\alpha}(\varrho n r)$, and by Lemma \ref{lem:V_alpha(R)_series_rep}, it has a unique series representation as $(\Dressing^{-1} f)(z) = \sum^{\infty}_{\ell = 0} c_{\ell} (\varrho n)^{-\ell} G^{(\ell), \model}(z)$ for some coefficients $c_{\ell}$. Then $f(z)$ has a unique series representation $f(z) = \sum^{\infty}_{\ell = 0} c_{\ell} G^{(\ell)}(z)$. Then for such $f(z) \in V^{(n), \dressing}_{\alpha}(r)$, we define, for $z \in D^*(0, r) \setminus \{ \arg z = \pm \frac{\theta^{-1} \pi + \gamma}{1 + \theta^{-1}} \}$
\begin{equation} \label{eq:defn_Q^0}
  Q^{(0)}(f)(z) = (Q^{\model}(\Dressing^{-1}(f)))^{\flat}(z) = \sum^{\infty}_{\ell = 0} c_{\ell} z^{\ell} = \frac{1}{2\pi i} \oint_{\lvert w \rvert = r'} f(w) \sum^{\infty}_{\ell = 0} H^{(-\ell)}(w) z^{\ell} \frac{dw}{w},
\end{equation}
where $r' \in (\lvert z \rvert, r)$.

From Lemma \ref{lem:PQ_id}, we derive that that $Q^{(0)} P^{(0)} = I$ as an operator on $H(r)$, and $P^{(0)}Q^{(0)} = I$ as an operator on $V^{(n), \dressing}_{\alpha}(r)$. Moreover, the latter has a reproducing kernel representation that for all $f(z) \in V^{(n), \dressing}_{\alpha}(r)$ and $z \in D^*(0, r) \setminus \{ \arg z = \pm \frac{\theta^{-1} \pi + \gamma}{1 + \theta^{-1}} \}$
\begin{equation}
  f(z) = P^{(0)}(Q^{(0)}(f))(z) = \frac{1}{2\pi i} \oint_{\lvert w \rvert = r'} f(w) \sum^{\infty}_{k = 0} G^{(\ell)}(z) H^{(-\ell)}(w) \frac{dw}{w},
\end{equation}
where $r' \in (\lvert z \rvert, r)$.

\begin{rmk}
  The operators $P^{(0)}$ and $Q^{(0)}$ defined above are generalizations of the operator $P^{(0)}$ defined in \cite[Section 3.6]{Wang-Zhang21} and its inverse, respectively. We note that when $\theta$ is a rational number, then $P^{(0)}$ and $Q^{(0)}$ degenerate into finite rank operators by \eqref{eq:P^(0)_pre_rational} and \eqref{eq:Q^(0)_pre_rational}, and can be expressed in a matrix form like \cite[Equation (3.127)]{Wang-Zhang21}.
\end{rmk}

\subsection{Final transform to $R$ and $\tR$}

Let
\begin{equation}\label{def:sigmaR}
 \Sigma^{R}:=[0,b]\cup [b+\rb, \infty) \cup \partial D(0,r^{1 + \theta^{-1}}_n) \cup \partial D(b,\rb) \cup \Sigma^R_1 \cup \Sigma^R_2,
\end{equation}
where $\Sigma^R_i = \Sigma_i^R(r^{1 + \theta^{-1}}_n)$ ($i = 1, 2$) are defined in \eqref{def:SigmaiR}. see Figure \ref{fig:Sigma_R} for an illustration and the orientation of the arcs.

Let $U \in V^{(n), \dressing}_{\alpha}(10 r_n)$ be the function defined in \eqref{eq:defn_U_from_Q}. We define $V(z) \in H(10 r_n)$ by
\begin{equation}
  V(z) = Q^{(0)}(U)(z), \quad z \in D(0, 10 r_n).
\end{equation}
Then let
 \begin{equation} \label{eq:R_1_R_2_around_0}
   R := (R_1, R_2) =
   \begin{cases}
     (R^{\pre}_1, R^{\pre}_2), & \text{$R_1(z)$ on $\compC \setminus (D(0,r^{1 + \theta^{-1}}_n) \cup \Sigma^R$}) \\
                               & \text{and $R_2(z)$ on $\halfH \setminus (D(0,r^{1 + \theta^{-1}}_n) \cup \Sigma^R)$}, \\
     \transform^{-1}(V), & \text{$R_1(z)$ on $D^*(0, r^{1 + \theta^{-1}}_n) \setminus \realR_+$}  \\
                               & \text{and $R_2(z)$ on $\halfH \cap D^*(0, r^{1 + \theta^{-1}}_n) \setminus \realR_+$}.
   \end{cases}
 \end{equation}

Next, we denote, for $i = 1, 2$,
\begin{align} \label{eq:defn_SigmaR^i}
  \SigmaR^{(1)}_i := {}& I_i(\Sigma^R_1 \cup \Sigma^R_2), & \SigmaR^{(2)}_i := {}& I_i((b + \rb, +\infty)), & \SigmaR^{(3)}_i := {}& I_i(\partial D(b, \rb)),
\end{align}
and then define
\begin{equation}
  \SigmaR' = \SigmaR^{(1)}_1 \cup \SigmaR^{(1)}_2 \cup \SigmaR^{(2)}_1 \cup \SigmaR^{(2)}_2 \cup \SigmaR^{(3)}_1 \cup \SigmaR^{(3)}_2.
\end{equation}
We also denote, for a small $r > 0$, the contour that encircles $-1$
\begin{equation}
  C^R(r) = C^R_1(r) \cup C^R_2(r) \quad \text{where} \quad C^R_1(r) = I_1(\partial D(0, r^{1 + \theta^{-1}})), \quad C^R_2(r) = I_2(\halfH \cap \partial D(0, r^{1 + \theta^{-1}})),
\end{equation}
and denote the region encircled by $C^R(r)$ as $D^R(r)$. At last, we define
\begin{equation} \label{eq:defn_SigmaR}
  \SigmaR(r) = \SigmaR' \cup  C^R(r), \quad \text{and} \quad \SigmaR = \SigmaR(r_n).
\end{equation}
See Figure \ref{fig:jump_R_scalar} for illustration.

\begin{figure}[htb]
  \begin{minipage}[t]{0.45\linewidth}
    \centering
    \includegraphics{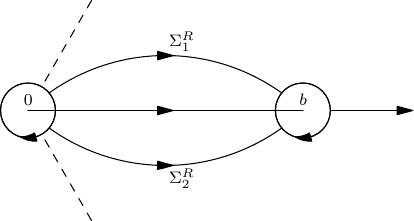}
    \caption[Contour $\Sigma^{R}$]{Schematic contour $\Sigma^{R}$.}
    \label{fig:Sigma_R}
  \end{minipage}
  \hspace{\stretch{1}}
  \begin{minipage}[t]{0.45\linewidth}
    \centering
    \includegraphics{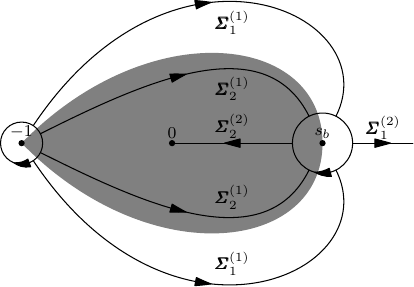}
    \caption[Contour $\SigmaR$]{Schematic contour $\SigmaR$. The circle around $s_b$ is $\SigmaR^{(3)}_1 \cup \SigmaR^{(3)}_2$, and the circle around $-1$ is $C^R(r_n)$. (The shaded region is $D$.)}
    \label{fig:jump_R_scalar}
  \end{minipage}
\end{figure}

Next, we define a transform $\transJ$ from a pair of functions $X = (X_1, X_2)$ defined on $(\compC, \halfH)$ to a function on $\compC \setminus (\gamma \cup [-1, 0])$, such that
\begin{equation} \label{eq:defn_transJ}
  \transJ(X)(s) =
  \begin{cases}
    X_1(J_c(s)), & s \in \compC \setminus \overline{D}, \\
    X_2(J_c(s)), & s \in D \setminus [-1, 0].
  \end{cases}
\end{equation}
For example, we consider two applications of $\transJ$: From $(X^{\id}_1, X^{\id}_2)$ defined in \eqref{eq:defn_X^id} and from $R = (R_1, R_2)$ defined in \eqref{eq:R_1_R_2_around_0}, and denote
\begin{align} \label{eq:tR_in_R_1_R_2}
  f_0(s) = {}& \transJ(\transform^{-1}(\id))(s) = \transJ(X^{\id}_1, X^{\id}_2)(s) = c^{\frac{\theta}{\theta + 1}} (s + 1) (-s)^{-\frac{1}{\theta + 1}}, && \text{and} & \tR = {}& \transJ(R_1, R_2),
\end{align}
where the $(-s)^{-1/(\theta + 1)}$ term takes the principal branch and $f_0(s)$ is analytic on $\compC \subseteq [0, +\infty)$, and $\tR$ is defined on on $\compC \setminus \SigmaR$. Furthermore, suppose $r$ is small enough, we have that the mapping $s \to f_0(s)$ is conformal from $D^R(r)$ to $D(0, r)$, and if $X = (X_1, X_2)$ satisfies $\transform(X)(s) = \varphi(s)$ for $s \in D(0, r)$, then $\transJ(X)(s)$ is analytic in $D^R(r)$ and is given by
\begin{equation}
  \transJ(X)(s) = \varphi(f_0(s)).
\end{equation}
Hence, $\tR(s) = V(f_0(s))$ for $s \in D^R(10 r_n)$.

We have that $\tR(s)$ is continuous and bounded on up to boundary $\SigmaR'$, since $R^{\pre}_1(s)$ and $R^{\pre}_2(s)$ are continuous and bounded on $\Sigma^R_1 \cup \Sigma^R_2 \cup [b + r^{(b)}, \infty) \cup \partial D(0, r^{1 + \theta^{-1}}_n)$. for $s \in C^R(r_n)$, $\tR_+(s)$ and $\tR_-(s)$ are continuous and bounded, since they are equal to $\transJ(R^{\pre})(s)$ and $V(f_0(s))$ respectively. 

We denote the function $f^R(s)$ for $s \in \SigmaR$ as
\begin{equation} \label{eq:defn_f^R}
  f^R(s) =
  \begin{cases}
    \tR_+(s) - \tR_-(s), & s \in \SigmaR \text{ and $s$ is not an intersection point}, \\
    0, & \text{otherwise}.
  \end{cases}
\end{equation}
By the property of $\tR$ discussed above, $f^R(s)$ is bounded and continuous on $\SigmaR$ except for the intersection points. Hence
\begin{equation} \label{eq:tR_minus_Cauchy_trans}
  \tR(s) - \frac{1}{2\pi i} \int_{w \in \SigmaR} f^R(w) \frac{dw}{w - s}
\end{equation}
has trivial discontinuity on $\SigmaR$ and can be extended to a holomorphic function on $\compC$. By the limit behaviour of $R_1(s) = R^{\pre}_1(s)$ as $s \to \infty$ given in Item \ref{enu:RHP:Svar:3} of RH problem \ref{RHP:Svar}, we find the holomorphic function in \eqref{eq:tR_minus_Cauchy_trans} converges to $1$ as $s \to \infty$, and it is the constant function $1$ by Liouville's theorem. Thus, we have
\begin{equation} \label{eq:R_in_f^R}
  \tR(s) = 1 + \frac{1}{2\pi i} \int_{w \in \SigmaR} f^R(w) \frac{dw}{w - s}.
\end{equation}

We also denote the function $g^R(s)$ for $s \in \SigmaR(2 r_n)$ as
\begin{equation} \label{defn_g^R}
  g^R(s) =
  \begin{cases}
    \tR_-(s), & s \in \SigmaR' \text{ and $s$ is not an intersection point}, \\
    \tR(s), & s \in C^R(2 r_n) \text{ and $s$ is not an intersection point}, \\
    0, & \text{otherwise}.
  \end{cases}
\end{equation}
Similar to $f^R(s)$, we also have that $g^R(s)$ is bounded on $\SigmaR(2 r_n)$, and is continuous except for the intersection points.

We define a transform $\Delta'$ that acts on functions defined on $\SigmaR'$. Let $f(s)$ be a function in $s \in \SigmaR'$, then
\begin{equation} \label{eq:defn_Delta_SigmaR}
  (\Delta' f)(s) =
  \begin{cases}
    J_{\SigmaR^{(1)}_1}(s) f(\s), & s \in \SigmaR^{(1)}_1 \text{ and } \s = I_2(J_c(s)) \in \SigmaR^{(1)}_2, \\
    J_{\SigmaR^{(2)}_2}(s) f(\s), & s \in \SigmaR^{(2)}_2 \text{ and } \s = I_1(J_c(s)) \in \SigmaR^{(2)}_1, \\
    J^1_{\SigmaR^{(3)}_1}(s) f(s) + J^2_{\SigmaR^{(3)}_2}(s) f(\s), & s \in \SigmaR^{(3)}_1 \text{ and } \s = I_2(J_c(s)) \in \SigmaR^{(3)}_2, \\
    J^1_{\SigmaR^{(3)}_2}(s) f(s) + J^2_{\SigmaR^{(3)}_1}(s) f(\s), & s \in \SigmaR^{(3)}_2 \text{ and } \s = I_1(J_c(s)) \in \SigmaR^{(3)}_1, \\
    0, & \text{$s \in \SigmaR^{(1)}_2 \cup \SigmaR^{(2)}_1$ or $s$ is an intersection point},
  \end{cases}
\end{equation}
where all contours do not contain intersection points, and
\begin{align}
  J_{\SigmaR^{(1)}_1}(s) = {}& (J_Q)_{21}(z), & s \in {}& \SigmaR^{(1)}_1 \text{ and $z = J_c(s) \in \Sigma^R_1 \cup \Sigma^R_2$}, \label{def:Jsigma1} \\
  J_{\SigmaR^{(2)}_2}(s) = {}& (J_Q)_{12}(z), & s \in {}& \SigmaR^{(2)}_2 \text{ and $z = J_c(s) \in (b + \rb, +\infty)$}, \label{def:Jsigma2'} \\
  J^{1}_{\SigmaR^{(3)}_1}(s) = {}& P^{(b)}_{11}(z) - 1, \quad J^{2}_{\SigmaR^{(3)}_1}(s) = P^{(b)}_{21}(z), & s \in {}& \SigmaR^{(3)}_1 \text{ and $z = J_c(s) \in \partial D(b,\rb)$}, \label{def:Jsigma3} \\
  J^{1}_{\SigmaR^{(3)}_2}(s) = {}& P^{(b)}_{22}(z) - 1, \quad J^{2}_{\SigmaR^{(3)}_2}(s) = P^{(b)}_{12}(z), & s \in {}& \SigmaR^{(3)}_2 \text{ and $z = J_c(s) \in \partial D(b,\rb)$}.
  \label{def:Jsigma32}
\end{align}
We also define a transform $\Delta_{2 r_n, r_n}$ that maps a function $f$ on $C^R(2 r_n)$ to a function $\Delta_{2 r_n, r_n} f$ on $C^R(r_n)$, such that, with $f_0$ defined in \eqref{eq:tR_in_R_1_R_2}, for $s \in C^R(r_n)$,
\begin{equation} \label{eq:defn_Delta_2rn_rn}
  (\Delta_{2 r_n, r_n}f)(s) = \frac{1}{2\pi i} \oint_{\lvert w \rvert = 2 r_n} f(f^{-1}_0(w)) \sum^{\infty}_{\ell = 0} H^{(-\ell)}(w) \left( G^{(\ell)}(f_0(s)) - (f_0(s))^{\ell} \right) \frac{dw}{w}.
\end{equation}
Then for a function $f$ defined on $\SigmaR(2r_n)$, we define the transform $\Delta$ maps $f$ to a function $\Delta(f)$ on $\SigmaR = \SigmaR(r_n)$, as 
\begin{equation} \label{eq:defn_Delta} 
  \Delta(f)(s) =
  \begin{cases}
    (\Delta' f)(s), & s \in \SigmaR' \text{ and $s$ is not an intersection point}, \\
    (\Delta_{2r_n, r_n} f)(s), & s \in C^R(r_n) \text{ and $s$ is not an intersection point}, \\
    0, & \text{otherwise}.
  \end{cases}
\end{equation}
For a function $g$ defined on $\SigmaR$, the transform $\Cauchy$ maps $g$ to a function defined on $\SigmaR(2r_n)$ as
\begin{multline} \label{eq:defn_Cauchy}
  \Cauchy(g)(s) = \\
  \begin{cases}
    \frac{1}{2\pi i} \int_{\SigmaR} g(w) \frac{dw}{w - s}, & s \in C^R(2r_n) \text{ and $s$ is not an intersection point}, \\
    \lim_{s' \to s \text{ from $-$ side}} \frac{1}{2\pi i} \int_{\SigmaR} g(w) \frac{dw}{w - s'}, & s \in \SigmaR' \text{ and $s$ is not an intersection point}, \\
    0, & \text{otherwise}.
  \end{cases}
\end{multline}
Although we have not specify the domain of the transforms $\Delta$ and $\Cauchy$, it is clear that they are well defined on $f^R$ and $g^R$ respectively, and $\Delta(1)$ is also well defined, where $1$ stands for the constant function on $\SigmaR(2r_n)$. Moreover, we have the identities
\begin{align}
  \Delta(g^R) = {}& f^R, & \Cauchy(f^R) + 1 = {}& g^R, & \Delta \Cauchy (f^R) + \Delta(1) = {}& f^R.
\end{align}

Now we put $\Delta$ and $\Cauchy$ into appropriate function spaces. We define the $L^2$ spaces $L^2(\SigmaR')$,  $L^2(C^R(r))$, $L^2(\SigmaR)$ and $L^2(\SigmaR(r))$, which consist of functions on $\SigmaR', C^R(r), \SigmaR = \SigmaR(r_n)$ and $\SigmaR(r)$ respectively, with the inner products defined by $\lvert ds \rvert$ means the arc length integral)
\begin{align} \label{eq:defn_inner_prod}
  \langle f, g \rangle_{L^2(\star)} = {}& \oint_{\star} f(s) \overline{g(s)} \lvert ds \rvert, & \star = {}& \SigmaR', C^R(r), \SigmaR, \SigmaR(r).
\end{align}
As an example, we see that $f^R \in L^2(\SigmaR)$, since $f^R(s)$ is bounded on $\SigmaR$ and continuous there, except for the intersection points, and $f^R(s)$ vanishes fast as $s \to \infty$.

Below we give the estimates of the norms of operators and functions in the $L^2$ spaces. Unless stated otherwise, all estimates are uniform for all large enough $n$. By the estimates \eqref{eq:est_JQ_12}, \eqref{eq:est_JQ_21}, \eqref{eq:est_Pb_row1} and \eqref{eq:est_Pb_row2}, we have that the transform $\Delta'$ defined in \eqref{eq:defn_Delta_SigmaR} is a bounded operator on $L^2(\SigmaR')$, and
\begin{equation} \label{eq:est_Delta'_norm}
  \lVert \Delta' \rVert_{L^2(\SigmaR')} = \bigO(n^{-1}).
\end{equation}
Also by the estimates \eqref{eq:G^ell_at_0:5}, we have that the transform $\Delta_{2 r_n, r_n}$ is a bounded operator from $L^2(C^R(2 r_n))$ to $L^2(C^R(r))$, and for any $g \in L^2(C^R(2r_n))$ with $\lVert g \rVert_{L^2(C^R(2r_n))} \neq 0$,
\begin{equation}
  \frac{\lVert \Delta_{2r_n, r_n} g \rVert_{L^2(C^R(r_n))}}{\lVert g \rVert_{L^2(C^R(2r_n))}} = \bigO(n^{\frac{1 - m_{\theta}}{1 + m_{\theta}}}).
\end{equation}
From the estimates of the operator norms of $\Delta'$ and $\Delta_{2r_n, r_n}$ above, we have that the transform $\Delta$ defined in \eqref{eq:defn_Delta} is a bounded operator from $L^2(\SigmaR(2 r_n))$ to $L^2(\SigmaR)$, and for any $f \in L^2(\SigmaR(2r_n))$ with $\lVert f \rVert_{L^2(\SigmaR(2r_n))} \neq 0$
\begin{equation}
  \frac{\lVert \Delta(f) \rVert_{L^2(\SigmaR)}}{\lVert f \rVert_{L^2(\SigmaR(2r_n))}} = \bigO(n^{\frac{1 - m_{\theta}}{m_{1 - \theta}}}).
\end{equation}
We note that although the constant function $1$ on $\SigmaR(2r_n)$ is not in $L^2(\SigmaR')$ or $L^2(\SigmaR(2r_n))$, $\Delta(1) \in L^2(\SigmaR)$ and
\begin{equation} \label{eq:est_Delta(1)}
  \lVert \Delta(1) \rVert_{L^2(\SigmaR)} = \lVert \Delta(1) \rVert_{L^2(\SigmaR')} + \lVert \Delta(1) \rVert_{L^2(C^R(r_n))} = \bigO(n^{-1}) + \bigO(n^{\frac{1 - m_{\theta}}{1 + m_{\theta}}} \cdot \sqrt{2 r_n}) = \bigO(n^{\frac{-m_{\theta}}{m_{\theta} + 1}}).
\end{equation}
By standard properties of Cauchy transform, we also have that the transform $\Cauchy$ defined in \eqref{eq:defn_Cauchy} is a bounded operator from $L^2(\SigmaR)$ to $L^2(\SigmaR(2r_n))$, and for any $g \in L^2(\SigmaR)$ with $\lVert g \rVert_{L^2(\SigmaR)} \neq 0$
\begin{equation} \label{eq:est_Cauchy_norm}
  \frac{\lVert \Cauchy(g) \rVert_{L^2(\SigmaR(2r_n))}}{\lVert g \rVert_{L^2(\SigmaR)}} = \bigO(1).
\end{equation}

We conclude this subsection by the following lemma:
\begin{lemma} \label{lem:tR_est}
  There exists $C > 0$, such that for all large enough $n$,
  \begin{enumerate*}[label=(\roman*)]
  \item \label{enu:lem:tR_est:1}
    $\lvert \tR(s) - 1 \rvert < C n^{\frac{1 - m_{\theta}}{1 + m_{\theta}}}$ for $s \in D^R(r_n/2)$, and
  \item \label{enu:lem:tR_est:2}
    $\lim_{s \to 0} \lvert \tR(s) - 1 \rvert < C n^{\frac{-m_{\theta}}{m_{\theta} + 1}}$.
  \end{enumerate*}
\end{lemma}
\begin{proof}
  By the estimates of the operator norms of $\Delta$ and $\Cauchy$ above, we find that the equation
  \begin{equation}
    \Delta \Cauchy f + \Delta(1) = f
  \end{equation}
  has a unique solution $(1 - \Delta\Cauchy)^{-1}(\Delta(1))$ in $L^2(\SigmaR)$, and its norm
  \begin{equation}
    \lVert f \rVert_{L^2(\SigmaR)} = \lVert (1 - \Delta\Cauchy)^{-1}(\Delta(1)) \rVert_{L^2(\SigmaR)} \leq \lVert (1 - \Delta\Cauchy)^{-1} \rVert_{L^2(\SigmaR)} \lVert \Delta(1) \rVert_{L^2(\SigmaR)} = \bigO(n^{\frac{-m_{\theta} }{m_{\theta} + 1}}).
  \end{equation}
  On the other hand, since $f^R \in L^2(\SigmaR)$ and it satisfies this equation, Hence the solution $f^R$, and we have that $\lVert f^R \rVert_{L^2(\SigmaR)} = \bigO(n^{\frac{-m_{\theta} }{m_{\theta} + 1}})$.

  For $s \in D^R(r_n/2)$, we have, by \eqref{eq:R_in_f^R} and the Cauchy-Schwarz inequality, 
  \begin{equation}
    \lvert \tR(s) - 1 \rvert \leq \lVert f^R \rVert_{L^2(\SigmaR)} \lVert \frac{1}{2\pi(w - s)} \rVert_{L^2(\SigmaR)} = \bigO(n^{\frac{-m_{\theta}}{m_{\theta} + 1}}) \bigO(n^{\frac{1}{m_{\theta} + 1}}) = \bigO(n^{\frac{1 - m_{\theta}}{1 + m_{\theta}}}),
  \end{equation}
  where $1/(2\pi(w - s))$ is regarded as a function in $w \in \SigmaR$.

  As $s \to 0$, we consider the integral in \eqref{eq:R_in_f^R} on $\SigmaR^{(2)}_2$ and $\SigmaR \setminus \SigmaR^{(2)}_2$ separately, and have
  \begin{equation}
    \lvert \tR(s) - 1 \rvert \leq \frac{1}{2\pi} \int_{\SigmaR^{(2)}_2} \frac{J_{\SigmaR^{(2)}_2}(s')}{s' - s} g^R(\s') ds' + \lVert f^R \rVert_{L^2(\SigmaR)} \lVert \frac{1_{w \notin \Sigma^{(2)}_2}}{2\pi(w - s)} \rVert_{L^2(\SigmaR)},
  \end{equation}
  where $\s' = I_1(J_c(s')) \in \SigmaR^{(2)}_1$. We find that the first integral term is $\bigO(e^{-\epsilon r_n})$ for some $\epsilon > 0$ due to the estimates \eqref{eq:est_JQ_21} and \eqref{eq:est_JQ_21_vanishing} of $J_{\SigmaR^{(2)}_2}(s') = (J_Q)_{12}(z')$ where $z' = J_c(s')$ and the boundedness of $\lVert g^R(z') \rVert_{L^2(\SigmaR)}$, and the second product term is $\bigO(n^{\frac{-m_{\theta}}{m_{\theta} + 1}})$. Hence we prove the result.
\end{proof}

\section{Asymptotic analysis of the RH problem for $\Y$} \label{sec:asy_RH_Y_tilde}

We take a parallel approach as for the RH problem for $Y$ in Section \ref{sec:asy_RH_Y}, and omit analogous technical details.

\subsection{Transforms of RH problem for $\Y$ from \cite{Wang-Zhang21} that are valid for all $\theta > 0$} \label{sec:asy_RH_Y_away_0_tilde}

\paragraph{Transform from $\Y$ to $\Q$}

Let $\Sigma_1 \subseteq \compC_+ \cap \halfH$, $\Sigma_2 = \overline{\Sigma_1}$ be a contour from $0$ to $b$, as defined in Section \ref{sec:asy_RH_Y_away_0}, and $\Sigma$ be defined in \eqref{def:Sigma}. 

Analogous to \eqref{def:thirdtransform}, we have, as shown in \cite[Sections 4.1--4.4]{Wang-Zhang21}, the explicit and invertible transformations $\Y \to \T \to \S \to \Q$, we derive (see \cite[Equations (4.1), (4.5) and (4.20)]{Wang-Zhang21})
\begin{multline} \label{eq:Qtilde_from_Ytilde}
  \Q(z) = (\Q_1(z), \Q_2(z)) = (\Y_1(z) e^{-n \g(z)}, \Y_2(z) e^{n (g(z) - \ell)}) \\
  \times
  \begin{cases}
    \begin{pmatrix}
      \P^{(\infty)}_1(z) & 0 \\
      0 & \P^{(\infty)}_2(z)
    \end{pmatrix}^{-1}
    & \text{$z$ outside the lens}, \\
    \begin{pmatrix}
      1 & 0 \\
      z^{-\alpha} e^{-n \phi(z)} & 1
    \end{pmatrix}
    \begin{pmatrix}
      \P^{(\infty)}_1(z) & 0 \\
      0 & \P^{(\infty)}_2(z)
    \end{pmatrix}^{-1},
    & \text{$z$ in the lower part of the lens}, \\
    \begin{pmatrix}
      1 & 0 \\
      -z^{-\alpha} e^{-n \phi(z)} & 1
    \end{pmatrix}
    \begin{pmatrix}
      \P^{(\infty)}_1(z) & 0 \\
      0 & \P^{(\infty)}_2(z)
    \end{pmatrix}^{-1},
    & \text{$z$ in the upper part of the lens}.
  \end{cases}
\end{multline}
Here $g(z), \g(z)$ and $\phi(z)$ are defined in \eqref{def:g}, \eqref{def:tildeg} and \eqref{def:phi}, and $\P^{(\infty)}_1(z)$ and $\P^{(\infty)}_2(z)$ are defined in \eqref{eq:P2} and \eqref{eq:P1}.

\paragraph{Local parametrix around $b$}

Let $r^{(b)}$ be the small positive constant the same as in Section \ref{sec:asy_RH_Y}, local conformal mapping $f_b(z)$ at $b$ be defined in \eqref{def:fb}, and $\P^{(\infty)}_1(z)$ and $\P^{(\infty)}_2(z)$ be defined in \eqref{eq:P2} and \eqref{eq:P1}. Analogous to \eqref{def:gib}, we define (\cite[Equation (4.26)]{Wang-Zhang21})
\begin{equation}\label{def:tildegib}
  \widetilde g^{(b)}_1(z) = \frac{z^{-\alpha /2}}{\widetilde P_1^{(\infty)}(z)}, \qquad \widetilde g^{(b)}_2(z) = \frac{z^{\alpha /2}}{\widetilde P_2^{(\infty)}(z)},
\end{equation}
and let, for $z \in D(0, \rb) \setminus \Sigma$, (\cite[Equation (4.27)]{Wang-Zhang21})
\begin{equation}\label{def:tildeEb}
\widetilde E^{(b)}(z)=\frac{1}{\sqrt{2}}
  \begin{pmatrix}
    \widetilde g^{(b)}_1(z) & 0  \\
    0 & \widetilde g^{(b)}_2(z)
  \end{pmatrix}^{-1}
  e^{\frac{\pi i}{4} \sigma_3}
  \begin{pmatrix}
    1 & -1 \\
    1 & 1
  \end{pmatrix}
  \begin{pmatrix}
    n^{\frac16} f_b(z)^{\frac14} & 0
    \\
    0 & n^{-\frac16} f_b(z)^{-\frac14}
  \end{pmatrix}.
\end{equation}
Analogous to $P^{(b)}(z)$ defined in \eqref{def:Pb}, local parametrix $\P^{(b)}(z)$ is the $2 \times 2$ matrix-valued function (\cite[Equations (3.38) and (4.27)]{Wang-Zhang21})
\begin{equation}\label{def:tildePb}
 \Pb (z) :=\widetilde E^{(b)}(z)
  \Psi^{(\Ai)}(n^{\frac23}f_b(z))
  \begin{pmatrix}
    e^{-\frac{n}{2} \phi(z)} \widetilde g^{(b)}_1(z) & 0
    \\
    0 & e^{\frac{n}{2} \phi(z)} \widetilde g^{(b)}_2(z)
  \end{pmatrix},  \quad
  z \in D(b, \rb) \setminus \Sigma,
\end{equation}
where $\Psi^{(\Ai)}$ is the Airy parametrix defined in \eqref{eq:defn_Psi_Airy}. Analogous to Proposition \ref{prop:Pb_asy}, we have the following properties of $\P^{(b)}(z)$ (\cite[RH problem 4.8(3), (4)]{Wang-Zhang21}):
\begin{prop}
  As $z \to b$, $(\P^{(b)}(z))_{ij} = \bigO((z - b)^{-1/4})$ and $(\P^{(b)}(z)^{-1})_{ij} = \bigO((z - b)^{-1/4})$, for $i, j = 1, 2$. For $z$ on the boundary $\partial D(b,\rb)$ (except for the intersecting points with $\Sigma$), we have, as $n \to \infty$, $\P^{(b)}(z) = I + \bigO(n^{-1})$ uniformly.
\end{prop}
The proof is identical to that of Proposition \ref{prop:Pb_asy}, and is omitted here. Next, analogous to \eqref{eq:defn_V^(b)}, from $\Q(z)$ defined in \eqref{eq:Qtilde_from_Ytilde}, we define a vector-valued function $\V^{(b)}(z) = (\V^{(b)}_1(z), \V^{(b)}_1(z))$ by (\cite[Equation (4.29)]{Wang-Zhang21})
\begin{equation} \label{eq:defn_V^(b)_tilde}
    \Vb(z) = \widetilde Q(z) \Pb (z)^{-1}, \qquad z \in D(b, \rb) \setminus \Sigma.
\end{equation}

\paragraph{Definition and properties of $\R^{\pre}$}

Let $\Sigma^R_1, \Sigma^R_2, \Sigma^R_1(r), \Sigma^R_2(r)$ be defined as in \eqref{def:SigmaiR}, and $\Sigma^{\pre}$ be defined as in \eqref{def:sigmapre}. With $\Q = (\Q_1, \Q_2)$ defined in \eqref{eq:Qtilde_from_Ytilde} and $\V^{(b)} = (\V^{(b)}_1, \V^{(b)}_2)$ in \eqref{eq:defn_V^(b)_tilde}, we define the $1 \times 2$ array of functions $\R^{\pre} = (\R^{\pre}_1, \R^{\pre}_2)$ on $(\halfH \setminus \Sigma^{\pre}, \compC \setminus \Sigma^{\pre})$ by
\begin{equation} \label{eq:defn_R^pre_tilde}
  \begin{aligned}
    \R^{\pre}_1(z) = {}& \Q_1(z), && z \in \halfH \setminus (\Sigma^{\pre} \cup D(b, \rb)), \\
    \R^{\pre}_2(z) = {}& \Q_2(z), && z \in \compC \setminus (\Sigma^{\pre} \cup D(b, \rb)), \\
    (\R^{\pre}_1, \R^{\pre}_2) = {}& (\V^{(b)}_1, \V^{(b)}_2), && \text{$\R_1(z)$ and $\R_2(z)$ on $D^*(b, \rb) \setminus (b - \rb, b)$}.
  \end{aligned}
\end{equation}
Since $\R^{\pre}$ is transformed from $\Y$, it satisfies the following RH problem that is derived from RH problem \ref{RHP:original_q}, that is analogous to RH problem \ref{RHP:Svar}:

\begin{RHP} \label{RHP:Svartilde} \hfill
  \begin{enumerate}[label=\emph{(\arabic*)}, ref=(\arabic*)]
  \item \label{enu:RHP:Svartilde:1}
    $\R^{\pre} = (\R^{\pre}_1, \R^{\pre}_2)$ is analytic in $(\halfH \setminus \Sigma^{\pre}, \compC \setminus \Sigma^{\pre})$, and is continuous up to the boundary, except for $0$.
  \item \label{enu:RHP:Svartilde:2}
    For $z \in (0, b) \cup \Sigma^R_1(0) \cup \Sigma^R_2(0) \cup (b + \rb, +\infty)$, we have that $(R^{\pre}_1)_{\pm}(z)$ and $(R^{\pre}_2)_{\pm}(z)$ are bounded for $z$ away from $0$, and
    \begin{equation}
      \R^{\pre}_+(z) = \R^{\pre}_-(z) J_{\Q}(z),
    \end{equation}
    where (\cite[Equation (4.22)]{Wang-Zhang21})
    \begin{equation} \label{def:Jtildecals}
      J_{\Q}(z) =
      \begin{cases}
        \begin{pmatrix}
          1 & 0 \\
          \frac{\Pinfty_2(z)}{\Pinfty_1(z)}z^{-\alpha}e^{-n \phi(z)} & 1
        \end{pmatrix},
        & \text{$z \in \Sigma^R_1(0) \cup \Sigma^R_2(0)$}, \\
        \begin{pmatrix}
          0 & 1 \\
          1 & 0
        \end{pmatrix},
        & \text{$z \in (0, b)$}, \\
        \begin{pmatrix}
          1 & \frac{\Pinfty_1(z)}{\Pinfty_2(z)}z^{\alpha}e^{n \phi(z)}  \\
          0 & 1
        \end{pmatrix},
        & \text{ $z \in (b + \rb, +\infty)$}.
      \end{cases}
    \end{equation}
    and for $z \in \partial D(b, \rb)$,
    \begin{equation}
      \R^{\pre}_+(z) = \R^{\pre}_-(z) \P^{(b)}(z).
    \end{equation}
  \item \label{enu:RHP:Svartilde:3}
    As $z \to \infty$ in $\halfH$, $\R^{\pre}_1$ behaves as $\R^{\pre}_1(z)=1+\bigO(z^{-\theta})$.
  \item \label{enu:RHP:Svartilde:4}
    As $z \to \infty$ in $\compC$, $\R^{\pre}_2$ behaves as $\R^{\pre}_2(z)=\bigO(1)$.
  \item
    As $z \to 0$ in $\halfH \setminus \Sigma$, we have
    \begin{equation} \label{eq:tildeasy_Q_1_in_lens}
      \R^{\pre}_1(z) =
      \begin{cases}
        \bigO(z^{ \frac{\theta-2\alpha}{2(1+\theta)} }), & \text{$\alpha  > 0$ and $z$ inside the lens,} \\
        \bigO(z^{\frac{\theta/2}{1+\theta}}\log z), & \text{$\alpha  = 0$ and $z$ inside the lens,} \\
        \bigO(z^{\frac{(\alpha+1/2)\theta}{1+\theta}}), & \text{$z$ outside the lens or $-1 < \alpha < 0$.}
      \end{cases}
    \end{equation}
  \item
    As $z \to 0$ in $\compC \setminus \Sigma$,  we have
    \begin{equation} \label{eq:tildeasy_Q_2_in_lens}
     \R^{\pre}_2(z)=
      \begin{cases}
        \bigO(z^{\frac{\theta-2\alpha}{2(1+\theta)}}), & \alpha > 0,
       \\
        \bigO (  z^{\frac{\theta}{2(1+\theta)}} \log z ), & \alpha =0,
        \\
        \bigO (z^{\frac{\theta(1+2\alpha)}{2(1+\theta)}} ), & \alpha < 0.
      \end{cases}
    \end{equation}
  \item
    As $z \to b$, we have $\R^{\pre}_1(z) = \bigO(1)$ and $\R^{\pre}_2(z) = \bigO(1)$.
  \item
    For $x>0$, we have the boundary condition $\R^{\pre}_1(e^{\pi i/\theta}x) = \R^{\pre}_1(e^{-\pi i/\theta}x)$.
  \end{enumerate}
\end{RHP}
By the regularity assumption in Section \ref{subsec:regularity} and analogous to \eqref{eq:est_JQ_12}, \eqref{eq:est_JQ_21}, \eqref{eq:est_Pb_row1} and \eqref{eq:est_Pb_row2}, we have that that there exists $\epsilon > 0$ such that
\begin{align}
  \lvert (J_{\Q})_{12}(z) \rvert < {}& e^{-\epsilon n}, & z \in {}& (b + \rb, +\infty), \label{eq:est_JQtilde_12} \\
  \lvert (J_{\Q})_{21}(z) \rvert < {}& e^{-\epsilon n r_n}, & z \in {}& (\Sigma_1 \cup \Sigma_2) \setminus (D(0, r^{1 + \theta^{-1}}_n) \cup D(0, \rb)). \label{eq:est_JQtilde_21}
\end{align}
and for $z \in \partial D(b, \rb)$
\begin{align}
  \lvert (\P^{(b)})_{11}(z) - 1 \rvert = {}& \bigO(n^{-1}), & \lvert (\P^{(b)})_{12}(z) \rvert = {}& \bigO(n^{-1}), \label{eq:est_Ptildeb_row1} \\
  \lvert (\P^{(b)})_{21}(z) \rvert = {}& \bigO(n^{-1}), & \lvert (\P^{(b)})_{22}(z) - 1 \rvert = {}& \bigO(n^{-1}). \label{eq:est_Ptildeb_row2}
\end{align}
These estimates are uniform for all large enough $n$, and we omit the derivations.

Analogous to Proposition \ref{prop_uniqueness}, we have
\begin{prop} \label{prop_uniqueness_tilde}
  RH problem \ref{RHP:Svartilde} has a unique solution.
\end{prop}
The proof is omitted since it is analogous to that of Proposition \ref{prop_uniqueness}.

\subsection{Local parametrix around $0$}

\subsubsection{Transformation of $Q$ into function space $V^{(n), \dressing}_{\alpha}(r)$}
Let the transform $\transform$ be defined by \eqref{eq:defn_trans_T}. We define a function $\U(s)$ from $\Q = (\Q_1, \Q_2)$ in \eqref{eq:Qtilde_from_Ytilde}, analogous to $U(s)$ in \eqref{eq:defn_U_from_Q}, by
\begin{equation} \label{eq:defn_U_from_Q_tilde}
  \U = \transform(\Q_2, \Q_1).
\end{equation}
Throughout this paper, we only consider $U(s)$ with $\lvert s \rvert < 10 r_n$. Like $U(s)$, $\U(s)$ is well defined by analytic continuation on the rays $\{ \arg z = 0 \}$ and $\{ \arg s = \pm \pi/(\theta + 1) \}$, while it has jumps along the rays $\{ \arg s = \pm \frac{\theta^{-1} \pi - \gamma}{1 + \theta^{-1}} \}$, see Figure \ref{fig:U_tilde_in_Q_tilde}. Moreover, $\U(s)$ is continuous up to the boundary at the two rays, and satisfies

\begin{equation} \label{eq:U_jump_tilde}
  \U_+(z) - \U_-(z) =
  \begin{cases}
    J_{\U}(z) \U(z e^{\frac{2\pi}{1 + \theta^{-1}} i}), & \arg z = \frac{\theta^{-1} \pi - \gamma}{1 + \theta^{-1}}, \\
    J_{\U}(z) \U(z e^{-\frac{2\pi}{1 + \theta^{-1}} i}), & \arg z = \frac{-\theta^{-1}\pi + \gamma}{1 + \theta^{-1}}, 
  \end{cases}
\end{equation}
where, with $J_{\Q}$ defined in \eqref{def:Jtildecals},
\begin{equation} \label{def:JU_tilde}
  J_{\U}(z) =
  \begin{cases}
    (J_{\Q}(z^{1 + \theta^{-1}} e^{-\frac{\pi}{\theta} i}))_{21}, & \arg z = \frac{\theta^{-1} \pi - \gamma}{1 + \theta^{-1}}, \\
    (J_{\Q}(z^{1 + \theta^{-1}} e^{\frac{\pi}{\theta} i}))_{21}, & \arg z = \frac{-\theta^{-1}\pi + \gamma}{1 + \theta^{-1}}.
  \end{cases}
\end{equation}
The limit behaviour of $\U(z)$ as $z \to 0$ can be derived from that of $\Q = \R^{\pre}$ in \eqref{eq:tildeasy_Q_1_in_lens} and \eqref{eq:tildeasy_Q_2_in_lens}.

Analogous to \eqref{eq:defn_n(z)}, we define
\begin{align} \label{eq:defn_n(z)_tilde}
  \n_1(z) = {}& \frac{\P^{(\infty), \pre}_1(z)}{\P^{(\infty)}_1(z)} \frac{e^{n(-\g(z) + \g_+(0))}}{e^{-n \g^{\pre}(z)}}, & \n_2(z) = \frac{\P^{(\infty), \pre}_2(z)}{\P^{(\infty)}_2(z)} \frac{e^{n(g(z) - V(z) - \ell + \g_-(0))}}{e^{n g^{\pre}(z)}}.
\end{align}
We have $\n_1(z)$ is well defined on $(D^*(0, b) \cap \halfH) \setminus \realR_+$ and $\n_2(z)$ is well defined on $D^*(0, b) \setminus \realR_+$. By \eqref{eq:P^infty_jump} and \eqref{eq:gequal}, we have, like \eqref{eq:m_n},
\begin{align} \label{eq:m_n_tilde}
  (\n_1)_+(x) = {}& (\n_2)_-(x), & (\n_2)_-(x) = {}& (\n_1)_+(x), & \n_1(x e^{\pi i/\theta}) = {}& \n_1(x e^{-\pi i/\theta}), & \text{for $x \in (0, b)$}. 
\end{align}
Then we define the function $\n(z)$ on $D^*(0, b^{\theta/(\theta + 1)}) \setminus \{ \arg z = 0, \pm \frac{\theta^{-1} \pi}{1 + \theta^{-1}} \}$, analogous to $n(z)$ defined in \eqref{eq:func_n}, by
\begin{equation}
  \n(z) =
  \begin{cases}
    \n_1(z^{1 + \theta^{-1}} e^{\frac{\pi i}{\theta}}), & \arg z \in (\frac{-\theta^{-1} \pi}{1 + \theta^{-1}}, 0), \\
    \n_1(z^{1 + \theta^{-1}} e^{-\frac{\pi i}{\theta}}), & \arg z \in (0, \frac{\theta^{-1} \pi}{1 + \theta^{-1}}), \\
    \n_2(-(-z)^{1 + \theta^{-1}}), & \arg (-z) \in (-\frac{\pi}{1 + \theta^{-1}}, \frac{\pi}{1 + \theta^{-1}}).
  \end{cases}
\end{equation}
Like $n(z)$, $\n(z)$ is naturally extended to be analytic on $D^*(0, b^{\theta/(\theta + 1)})$, and like \eqref{eq:est_n(z)},
\begin{equation}
  \n(z) - 1 =
  \begin{cases}
    \bigO(n^{1 - m_{\theta}}), & \lvert z \rvert \in [1, C], \\
    \bigO(n^{\frac{1 - m_{\theta}}{1 + m_{\theta}}}), & \lvert z \rvert \in [10^{-1} r_n, 10 r_n],
  \end{cases}
\end{equation}
where $C > 1$ is an arbitrary constant, and the estimates are uniform for all large enough $n$.

Let $f(z)$ be a function whose domain is in $D(0, b^{\theta/(\theta + 1)})$. Analogous to \eqref{eq:G_divided_dressed}, we define the transforms $\Dressingtilde$ and $\Dressinginvtilde$ on $f(z)$ by
\begin{align}
  \Dressingtilde(f)(z) = {}& \n(z) e^{\varrho nz} f(\varrho nz), & \Dressinginvtilde(f)(z) = {}& \n^{-1}(z) e^{-\varrho nz} f(\varrho nz),
\end{align}
Suppose $R \in (0, 10 \rho n r_n)$ and $f(z)$ is a function defined on $D^*(R) \setminus \{ \arg z = \pm \frac{\theta^{-1} \pi - \gamma}{1 + \theta^{-1}} \}$. so that $\Dressingtilde(f)(z)$ is a function on $D^*(0, (\varrho n)^{-1} R) \setminus \{ \arg z = \pm \frac{\theta^{-1} \pi - \gamma}{1 + \theta^{-1}} \}$. Hence for any $r \in \realR_+ \cup \{ \infty \}$, the inverse transform $\Dressingtilde^{-1}$ is well defined on functions on $D^*(0, r) \setminus \{ \arg z = \pm \frac{\theta^{-1} \pi + \gamma}{1 + \theta^{-1}} \}$. let
\begin{equation}
  \V^{(n), \dressing}_{\alpha}(r) = \{ \Dressingtilde(f)(z) : f(z) \in \V_{\alpha}(\varrho n r) \}.
\end{equation}
From Definition \ref{defn:V_alpha_tilde} of function space $\V_{\alpha}(R)$, we derive the definition of function space $\V^{(n), \dressing}_{\alpha}(r)$ as follows.
\begin{defn} \label{defn:V_dressing_tilde}
  $\V^{(n), \dressing}_{\alpha}(r)$ consists of functions $f(z)$ on $z \in D^*(0, r) \setminus \{ \arg z = \pm \frac{\theta^{-1} \pi - \gamma}{1 + \theta^{-1}} \}$, such that $r \in (0, 10 r_n)$ and
  \begin{itemize}
  \item
    $f(z)$ is analytic in the sector $\arg (-z) \in (\frac{-\pi - \gamma}{1 + \theta^{-1}}, \frac{\pi + \gamma}{1 + \theta^{-1}})$ and the sector $\arg z \in (\frac{-\theta^{-1} \pi + \gamma}{1 + \theta^{-1}}, \frac{\theta^{-1} \pi - \gamma}{1 + \theta^{-1}})$ separately, and $f(z)$ is continuous up to the boundary on the two rays $\{ \arg z = \pm \frac{\theta^{-1} \pi - \gamma}{1 + \theta^{-1}} \}$.
  \item
    Let the two rays be oriented from $0$ to $\infty$. The boundary values of $f$ on the sides of the two rays satisfy
    \begin{align}
      &
        \begin{aligned}[b]
          f_+(z) - f_-(z) = {}& -e^{\frac{2\alpha + 1}{\theta + 1} \pi i} e^{\varrho n z(1 - e^{\frac{2\pi}{\theta + 1} i})} \frac{\n(z)}{\n(z e^{-\frac{2\pi}{\theta + 1} i)})} f(z e^{-\frac{2\pi}{\theta + 1} i}) \\
          = {}& J_{\U}(z) f(z e^{-\frac{2\pi}{\theta + 1} i}), 
        \end{aligned} 
                              & \arg z = {}& \frac{\theta^{-1} \pi - \gamma}{1 + \theta^{-1}}, \label{eq:G^lk_jump_1_tilde} \\
      &
        \begin{aligned}[b]
          f_+(z) - f_-(z) = {}& e^{-\frac{2\alpha + 1}{\theta + 1} \pi i} e^{\varrho n z(1 - e^{-\frac{2\pi}{\theta + 1} i})} \frac{\n(z)}{\n(z e^{\frac{2\pi}{\theta + 1} i})} f(z e^{\frac{2\pi}{\theta + 1} i}) \\
          = {}& J_{\U}(z) f(z e^{\frac{2\pi}{\theta + 1} i}), 
        \end{aligned} 
                              & \arg z = {}& \frac{-\theta^{-1} \pi + \gamma}{1 + \theta^{-1}}. \label{eq:G^lk_jump_2_tilde}
    \end{align}
    where $J_{\U}$ defined in \eqref{def:JU_tilde}.
  \item
    As $z \to 0$, $f(z)$ has the limit behaviour depending on $\alpha$ and $\theta$ characterized by \eqref{eq:V_alpha(R)_at_0:1_tilde} -- \eqref{eq:V_alpha(R)_at_0:3_tilde} (for $\f(z)$ there).
\end{itemize}
\end{defn}

As $z \to 0$, the limit behaviour of $U(z)$ that can be found from the limit behaviour of $(\Q_1, \Q_2) = (\R^{\pre}_1, \R^{\pre}_2)$ in \eqref{eq:tildeasy_Q_1_in_lens} and \eqref{eq:tildeasy_Q_2_in_lens} via the transform formula \eqref{eq:defn_U_from_Q_tilde}. By comparing the discontinuity condition \eqref{eq:U_jump_tilde} and the limit behaviour at $0$ of $\U$ and Definition \ref{defn:V_dressing_tilde}, we find that $\U(z)$ belongs to $\V^{(n), \dressing}(r)$ for all $r \in (0, 10 r_n)$.

\subsubsection{Functions $\G^{(\ell)}$ and $\H^{(\ell)}$, and operators $\P^{(0)}$ and $\Q^{(0)}$}

Analogous to \eqref{G^ell_from_model}, we apply transform $\Dressingtilde$ to $\G^{(\ell), \model}(z)$ and transform $\Dressinginvtilde$ to $\H^{(\ell), \model}(z)$, and get
\begin{align} \label{G^ell_from_model_tilde}
  \G^{(\ell)}(z) = {}& (\varrho n)^{-\ell} \Dressingtilde(\G^{(\ell), \model})(z), & \H^{(\ell)}(z) = {}& (\varrho n)^{-\ell} \Dressinginvtilde(\H^{(\ell), \model})(z).
\end{align}
$\G^{(\ell)}$ is analytic on $D(0, b^{\theta/(\theta + 1)}) \setminus \{ \arg z = \pm \frac{\theta^{-1} \pi - \gamma}{1 + \theta^{-1}} \}$, continuous up to the boundary on the two rays, and if $\ell \in \natN$, then $\G^{(\ell)}(z) \in \V^{(n), \dressing}_{\alpha}(r)$ for all $r \in (0, 10 r_n)$. Similarly, $H^{(\ell)}(z)$ is analytic on $D(0, b^{\theta/(\theta + 1)}) \setminus \{ \arg z = \pm \frac{\theta^{-1} \pi + \gamma}{1 + \theta^{-1}} \}$, and is continuous up to the boundary on the two rays.

Similar to the estimates \eqref{eq:est_G^ell_central} and \eqref{eq:G^ell_at_0:5}, from the definitions \eqref{eq:G_divided_model}, \eqref{eq:H_divided_model}, \eqref{eq:G_H_ell_tilde}, and the estimates \eqref{eq:asy_IJK_at_0} and \eqref{eq:G_divided}, we have the estimate that for any constant $C > 1$, If $\zeta \in D(0, C) \setminus (D(0, 1) \cup \{ \arg z = \pm \frac{\theta^{-1} \pi - \gamma}{1 + \theta^{-1}} \})$, then
  \begin{equation} \label{eq:est_G^ell_central_tilde}
    \zeta^{-\ell} \left( (\varrho n)^{\ell} \G^{(\ell)}((\varrho n)^{-1} \zeta) - \G^{(\ell), \model}(\zeta) \right) = \bigO(n^{1 - m_{\theta}}),
  \end{equation}
and if $\zeta \in D(0, 10 r_n) \setminus (D(0, 10^{-1} r_n) \cup \{ \arg z = \pm \frac{\theta^{-1} \pi - \gamma}{1 + \theta^{-1}} \})$, then 
  \begin{align}
    \lvert z^{-\ell} (\G^{(\ell)}(z) - 1) \rvert \leq {}& M n^{\frac{m_{\theta} - 1}{m_{\theta} + 1}}, & \lvert z^{-\ell} (\H^{(\ell)}(z) - 1) \rvert \leq {}& M n^{\frac{m_{\theta} - 1}{m_{\theta} + 1}}.
  \end{align}

From the operators $\P^{\model}: H(R) \to \V_{\alpha}(R)$ and $\Q^{\model}(f)(z): \V_{\alpha}(R) \to H(R)$, we define the operators $\P^{(0)}: H(r) \to \V^{(n), \dressing}_{\alpha}(r)$ and $\Q^{(0)}(f)(z): \V^{(n), \dressing}_{\alpha}(r) \to H(r)$ as follows. For any $h(z) \in H(r)$, we denote the function $h^{\sharp}(z) \in H(\varrho n R)$ by $h^{\sharp}(z) = h((\varrho n)^{-1} z)$ as in \eqref{eq:defn_P^0}. For any $h(z) \in H(r)$ with $h(z) = \sum^{\infty}_{\ell = 0} a_{\ell} z^{\ell}$, we define, for $z \in D(0, r)$, analogous to \eqref{eq:defn_P^0},
\begin{equation}
  \P^{(0)}(h)(z) = \Dressingtilde(\P^{\model}(h^{\sharp}))(z) = \sum^{\infty}_{\ell = 0} a_{\ell} \G^{(\ell)}(z) = \frac{1}{2\pi i} \oint_{\lvert w \rvert = r'} h(w) \sum^{\infty}_{\ell = 0} w^{-\ell} \G^{(\ell)}(z) \frac{dw}{w},
\end{equation}
where $r' \in (\lvert z \rvert, r)$. On the other hand, for any $h(z) \in H(\varrho n r)$, we denote the function $h^{\flat} \in H(r)$ by $h^{\flat}(z) = h(\varrho n z)$, like in \eqref{eq:defn_Q^0}. For any $f(z) \in \V^{(n), \dressing}_{\alpha}(r)$, we have that $\Dressingtilde^{-1} f)(z) \in \V_{\alpha}(\varrho nr)$, and by Lemma \ref{lem:V_alpha(R)_series_rep}, it has a unique series representation $(\Dressingtilde^{-1} f)(z) = \sum^{\infty}_{\ell = 0} c_{\ell} (\varrho n)^{-\ell} \G^{(\ell, \model)}(z)$ for some coefficients $c_{\ell}$. Then $f(z)$ has a unique series representation $f(z) = \sum^{\infty}_{\ell = 0} c_{\ell} G^{(\ell)}(z)$. Then for such $f(z) \in \V^{(n), \dressing}_{\alpha}(r)$, we define, for $z \in D^*(0, r) \setminus \{ \arg z = \pm \frac{\theta^{-1} \pi - \gamma}{1 + \theta^{-1}} \}$
\begin{equation}
  \Q^{(0)}(f)(z) = \Q^{\model}(\Dressingtilde^{-1}(f))^{\flat}(z) = \sum^{\infty}_{\ell = 0} a_{\ell} z^{\ell} = \frac{1}{2\pi i} \oint_{\lvert w \rvert = r'} f(w) \sum^{\infty}_{\ell = 0} \H^{(-\ell)}(w) z^{\ell} \frac{dw}{w},
\end{equation}
where  $r' \in (\lvert z \rvert, r)$.

From Lemma \ref{lem:PQ_id}, we derive that that $\Q^{(0)} \P^{(0)} = I$ as an operator on $H(r)$, and $\P^{(0)}\Q^{(0)} = I$ as an operator on $\V^{(n), \dressing}_{\alpha}(r)$. Moreover, the latter has a reproducing kernel representation that for all $f(z) \in \V^{(n), \dressing}_{\alpha}(r)$ and $z \in D^*(0, r) \setminus \{ \arg z = \pm \frac{\theta^{-1} \pi - \gamma}{1 + \theta^{-1}} \}$
\begin{equation}
  f(z) = \P^{(0)}(\Q^{(0)}(f))(z) = \frac{1}{2\pi i} \oint_{\lvert w \rvert = R'} f(w) \sum^{\infty}_{k = 0} \G^{(\ell)}(z) \H^{(-\ell)}(w) \frac{dw}{w},
\end{equation}
where $r' \in (\lvert z \rvert, r)$.

\subsection{Final transform to $\R$ and $\tRtilde$}

Recall the contours $\Sigma^R, \Sigma^R_1, \Sigma^R_2$ defined in \eqref{def:sigmaR} and \eqref{def:SigmaiR}. Let $\U$ be the function analytic on $D^*(0, 10 r_n) \setminus \{ \arg z = \pm \frac{\theta^{-1} \pi - \gamma}{1 + \theta^{-1}} \}$, as defined in \eqref{eq:defn_U_from_Q_tilde}. We define $\V(z) \in H(10 r_n)$ by
\begin{equation}
  \V(z) = \Q^{(0)}(\U)(z), \quad z \in D(0, 10 r_n).
\end{equation}
Then let the array of functions $(\R_2, \R_1)$ on $(\compC \setminus \SigmaR, \halfH \setminus \SigmaR)$ be defined as
 \begin{equation} \label{eq:R_1_R_2_around_0_tilde}
   (\R_2, \R_1) =
   \begin{cases}
     (\R^{\pre}_2, \R^{\pre}_1), & \text{$\R_2(z)$ on $\compC \setminus (D(0,r^{1 + \theta^{-1}}_n) \cup \Sigma^R$}) \\
                               & \text{and $\R_1(z)$ on $\halfH \setminus (D(0,r^{1 + \theta^{-1}}_n) \cup \Sigma^R)$}, \\
     \transform^{-1}(\V), & \text{$\R_2(z)$ on $D^*(0, r^{1 + \theta^{-1}}_n) \setminus \realR_+$}  \\
                               & \text{and $\R_1(z)$ on $\halfH \cap D^*(0, r^{1 + \theta^{-1}}_n) \setminus \realR_+$}.
   \end{cases}
 \end{equation}
Recall the contours $\SigmaR^{(j)}_i$ ($i = 1, 2$ and $j = 1, 2, 3$) defined in \eqref{eq:defn_SigmaR^i}. We denote the inversion mapping $\omega: \compC \setminus \{ 0 \} \to \compC \setminus \{ 0 \}$ as $\omega(z) = z^{-1}$, and then define
\begin{equation}
  \SigmaRtilde^{(j)}_i := \omega(\SigmaR^{(j)}_i), \quad \text{$i = 1, 2$ and $j = 1, 2, 3$},
\end{equation}
and
\begin{equation}
  \SigmaRtilde = \SigmaRtilde' \cup  C^{\R}(r_n), \quad \text{and} \quad \SigmaRtilde = \SigmaRtilde(r_n), \quad \text{where} \quad \SigmaRtilde' = \omega(\SigmaR'), \quad C^{\R}(r_n) = \omega(C^R(r_n)).
\end{equation}

Recall the transform $\transJ$ defined in \eqref{eq:defn_transJ}. We denote the inversion transform $\Inversion$ as $\Inversion(f)(s) = f(s^{-1})$. We note that $\Inversion^{-1} = \Inversion$. Let
\begin{equation} \label{eq:tR_in_R_1_R_2_tilde}
  \tRtilde = \Inversion(\transJ(\R_2, \R_1)).
\end{equation}
Like $\tR$ defined in \eqref{eq:tR_in_R_1_R_2}, we have that $\tRtilde(s)$ is continuous and bounded up to boundary $\SigmaRtilde$. We denote, analogous to $f^R(s)$ in \eqref{eq:defn_f^R}, the function $f^{\R}(s)$ for $s \in \SigmaRtilde$ as
\begin{equation}
  f^{\R}(s) =
  \begin{cases}
    \tRtilde_+(s) - \tRtilde_-(s), & s \in \SigmaRtilde \text{ and $s$ is not an intersection point}, \\
    0, & \text{otherwise}.
  \end{cases}
\end{equation}
Then analogous to \eqref{eq:R_in_f^R}, we have
\begin{equation} \label{eq:R_in_f^R_tilde}
  \tRtilde(s) = 1 + \frac{1}{2\pi i} \int_{w \in \SigmaRtilde} f^{\R}(w) \frac{dw}{w - s}.
\end{equation}
We also denote, analogous to $g^R(s)$ in \eqref{defn_g^R}, the function $g^{\R}(s)$ for $s \in \SigmaRtilde(2 r_n)$ as
\begin{equation}
  g^{\R}(s) =
  \begin{cases}
    \tRtilde_-(s), & s \in \SigmaRtilde' \text{ and $s$ is not an intersection point}, \\
    \tRtilde(s), & s \in C^{\R}(2 r_n) \text{ and $s$ is not an intersection point}, \\
    0, & \text{otherwise}.
  \end{cases}
\end{equation}
We note that both $f^{\R}(s)$ and $g^{\R}(s)$ are bounded and continuous except for the intersection points. Then analogous to $\Delta'$ in \eqref{eq:defn_Delta_SigmaR}, we define the transform $\Deltatilde'$ that acts on functions defined on $\SigmaRtilde'$. Let $f(s)$ be a function in $s \in \SigmaRtilde'$, then
\begin{equation}\label{def:DeltatildeR}
  \Deltatilde' f(s) =
  \begin{cases}
   % 0, & s \in \SigmaR^{(1')} \cup \SigmaR^{(2)}, \\
    J_{\SigmaRtilde^{(1)}_2}(s) f(\s), & \text{$s \in \SigmaRtilde^{(1)}_2$ and $\s = \omega(I_1(J_c(\omega(s)))) \in \SigmaRtilde^{(1)}_1$}, \\
    J_{\SigmaRtilde^{(2)}_1}(s) f(\s), & \text{$s \in \SigmaRtilde^{(2)}_1$ and $\s = \omega(I_2(J_c(\omega(s)))) \in \SigmaRtilde^{(2)}_2$}, \\
    J^{1}_{\SigmaRtilde^{(3)}_1}(s) f(s) + J^{2}_{\SigmaRtilde^{(3)}_2}(s) f(\s), & \text{$s \in \SigmaRtilde^{(3)}_1$ and $\s = \omega(I_2(J_c(\omega(s)))) \in \SigmaRtilde^{(3)}_2$}, \\
    J^{1}_{\SigmaRtilde^{(3)}_2}(s) f(s) + J^{2}_{\SigmaRtilde^{(3)}_1}(s) f(\s), & \text{$s \in \SigmaRtilde^{(3)}_2$ and $\s = \omega(I_1(J_c(\omega(s)))) \in \SigmaRtilde^{(3)}_1$}, \\
    0, & \text{$s \in \SigmaRtilde^{(1)}_1 \cup \SigmaRtilde^{(2)}_2$ or $s$ is an intersection point},
  \end{cases}
\end{equation}
where all contours do not contain intersection points, and
\begin{align}
  J_{\SigmaRtilde^{(1)}_2}(s) = {}& (J_{\Q})_{21}(z), & s \in {}& \SigmaRtilde^{(1)}_2 \text{ and $z = J_c(\omega(s)) \in \Sigma^R_1 \cup \Sigma^R_2$}, \\
  J_{\SigmaRtilde^{(2)}_1}(s) = {}& (J_{\Q})_{12}(z), & s \in {}& \SigmaRtilde^{(2)}_1 \text{ and $z = \J(s) \in (b + \rb, +\infty)$.} \\
  J^{1}_{\SigmaRtilde^{(3)}_1}(s) = {}& \P^{(b)}_{22}(z) - 1, \quad J^{2}_{\SigmaRtilde^{(3)}_2}(s) = \P^{(b)}_{12}(z), & s \in {}& \SigmaRtilde^{(3)}_1 \text{ and $z = \J(s) \in \partial D(b,\rb)$}, \\
  J^{1}_{\SigmaRtilde^{(3)}_2}(s) = {}& \P^{(b)}_{11}(z) - 1, \quad J^{2}_{\SigmaRtilde^{(3)}_1}(s) = \P^{(b)}_{21}(z), & s \in {}& \SigmaRtilde^{(3)}_2 \text{ and $z = \J(s) \in \partial D(b,\rb)$}.
\end{align}
Analogous to $\Delta_{2 r_n, r_n}$ in \eqref{eq:defn_Delta_2rn_rn}, we define a transform $\Deltatilde_{2 r_n, r_n}$ that maps from a function $f$ on $C^{\R}(2 r_n)$ to a function $\Deltatilde_{2 r_n, r_n} f$ on $C^{\R}(r_n)$, such that, with $f_0$ defined in \eqref{eq:tR_in_R_1_R_2}, for $s \in C^{\R}(r_n)$,
\begin{equation}
  (\Deltatilde_{2 r_n, r_n}f)(s) = \frac{1}{2\pi i} \oint_{\lvert w \rvert = 2 r_n} f(\omega(f^{-1}_0(w))) \sum^{\infty}_{\ell = 0} H^{(-\ell)}(w) \left( G^{(\ell)}(f_0(\omega(s))) - (f_0(\omega(s)))^{\ell} \right) \frac{dw}{w}.
\end{equation}
Then analogous to \eqref{eq:defn_Delta}, for a function $f$ defined on $\SigmaRtilde(2r_n)$, we define the transform $\Deltatilde$ maps $f$ to a function $\Deltatilde(f)$ on $\SigmaRtilde = \SigmaRtilde(r_n)$, as 
\begin{equation}
  \Deltatilde(f)(s) =
  \begin{cases}
    (\Deltatilde' f)(s), & s \in \SigmaRtilde' \text{ and $s$ is not an intersection point}, \\
    (\Deltatilde_{2r_n, r_n} f)(s), & s \in C^{\R}(r_n) \text{ and $s$ is not an intersection point}, \\
    0, & \text{otherwise}.
  \end{cases}
\end{equation}
Analogous to $\Cauchy$ defined in \eqref{eq:defn_Cauchy}, for a function $g$ defined on $\SigmaRtilde$, the transform $\Cauchytilde$ maps $g$ to a function defined on $\SigmaRtilde(2r_n)$ as
\begin{multline}
  \Cauchytilde(g)(s) = \\
  \begin{cases}
    \frac{1}{2\pi i} \int_{\SigmaRtilde} g(w) \frac{dw}{w - s}, & s \in C^{\R}(2r_n) \text{ and $s$ is not an intersection point}, \\
    \lim_{s' \to s \text{ from $-$ side}} \frac{1}{2\pi i} \int_{\SigmaRtilde} g(w) \frac{dw}{w - s'}, & s \in \SigmaR' \text{ and $s$ is not an intersection point}, \\
    0, & \text{otherwise}.
  \end{cases}
\end{multline}

Like the $L^2$ spaces $L^2(\SigmaR')$,  $L^2(C^R(r))$, $L^2(\SigmaR)$ and $L^2(\SigmaR(r))$ defined in \eqref{eq:defn_inner_prod}, we introduce the $L^2$ spaces $L^2(\SigmaRtilde'), L^2(C^{\R}(r)), L^2(\SigmaRtilde), L^2(\SigmaRtilde(r))$ also by \eqref{eq:defn_inner_prod}, with $\star = \SigmaRtilde', C^{\R}(r), \SigmaRtilde, \SigmaRtilde(r)$. We have estimates for the operator norms of $\Deltatilde$ and $\Cauchytilde$. Since the derivation of the estimates is parallel to \eqref{eq:est_Delta'_norm} -- \eqref{eq:est_Cauchy_norm}, we only summarize the results below:
\begin{align}
  \frac{\lVert \Deltatilde(f) \rVert_{L^2(\SigmaRtilde)}}{\lVert f \rVert_{L^2(\SigmaRtilde(2r_n))}} = {}& \bigO(n^{-\frac{1 - m_{\theta}}{1 + m_{\theta}}}), & \frac{\lVert \Cauchytilde(g) \rVert_{L^2(\SigmaRtilde(2r_n))}}{\lVert g \rVert_{L^2(\SigmaRtilde)}} = \bigO(1).
\end{align}
for any $f \in L^2(\SigmaR(2r_n))$ with $\lVert f \rVert_{L^2(\SigmaR(2r_n))} \neq 0$ and any $g \in L^2(\SigmaR)$ with $\lVert f \rVert_{L^2(\SigmaR(2r_n))} \neq 0$. Also as examples, we have that $f^{\R}, \Deltatilde(1) \in L^2(\SigmaRtilde)$ and analogous to \eqref{eq:est_Delta(1)} we have $\lVert \Delta(1) \rVert_{L^2(\SigmaR)} = \bigO(n^{\frac{-m_{\theta}}{m_{\theta} + 1}})$.

Analogous to Lemma \ref{lem:tR_est}, we have the following result:
\begin{lemma} \label{lem:tRtilde_est}
  There exists $C > 0$, such that for all large enough $n$, $\lvert \tRtilde(s) - 1 \rvert < C n^{\frac{1 - m_{\theta}}{1 + m_{\theta}}}$ for $s \in D^{\R}(r_n/2)$, where $D^{\R}(r)$ is the region encircled by $C^{\R}$.
\end{lemma}

\begin{proof}[Sketch of proof]
  The proof of the lemma is analogous to that of Part \ref{enu:lem:tR_est:1} of Lemma \ref{lem:tR_est}. We show by Proposition \ref{prop_uniqueness_tilde} that $f^{\R}$ is the unique solution of the equation $\Deltatilde \Cauchytilde f + \Deltatilde(1) = f$ in $L^2(\SigmaRtilde)$, and then $\lVert f^{\R} \rVert_{L^2(\SigmaRtilde)} = \lVert (1 - \Deltatilde\Cauchytilde)^{-1}(\Deltatilde(1)) \rVert_{L^2(\SigmaRtilde)} = \bigO(n^{\frac{-m_{\theta} }{m_{\theta} + 1}})$. Then the desired estimate of $\tRtilde - 1$ follows from this estimate directly through \eqref{eq:R_in_f^R_tilde}. We omit the detail.
\end{proof}

\section{Proof of main results} \label{sec:proof_main}

\paragraph{Proof of Theorem \ref{thm:pqkappa}}

 From the transform $Y \to Q$ given in \eqref{def:thirdtransform}, it is readily seen from \eqref{def:Y} that
\begin{equation} \label{eq:Y_1_in_Z_i}
  p_n(z)=Y_1(z) =
  \begin{cases}
    Z_1(z), & \text{$z \in \compC$ outside the lens}, \\
    Z_1(z) + Z_2(z), & \text{$z$ in the upper part of the lens}, \\
    Z_1(z) - Z_2(z), & \text{$z$ in the lower part of the lens},
  \end{cases}
  \end{equation}
where
\begin{equation} \label{def:Zi}
  Z_1(z) = Q_1(z) P^{(\infty)}_1(z) e^{ng(z)}, \quad Z_2(z) = Q_2(z) P^{(\infty)}_2(z) \theta z^{-\alpha - 1 + \theta} e^{n(V(z) - \g(z) + \ell)},
\end{equation}
with $g(z)$ and $\g(z)$ defined in \eqref{def:g} and \eqref{def:tildeg}, respectively. By tracing back the transformations $Q \to R^{\pre} \to R \to \tR$ given in \eqref{eq:defn_R^pre}, \eqref{eq:R_1_R_2_around_0} and \eqref{eq:tR_in_R_1_R_2}, we have that if we consider $Q_1(z)$ as a function on the domain $D^*(0, (r_n/2)^{1 + \theta^{-1}}) \setminus \Sigma$ and $Q_2(z)$ as a function on the domain $(D^*(0, (r_n/2)^{1 + \theta^{-1}}) \cap \halfH) \setminus \Sigma$, then they are expressed by $(Q_1, Q_2) = \transform^{-1}(U)$, and
\begin{equation} \label{eq:asy_U(s)_at_0}
  \begin{split}
    U(s) = {}& (P^{(0)}(\transform(\transJ^{-1}(\tR))))(s) \\
    = {}& \frac{1}{2\pi i} \oint_{\lvert w \rvert = r_n/2} \transform \transJ^{-1}(\tR)(w) \sum^{\infty}_{\ell = 0} w^{-\ell} G^{(\ell)}(s) \frac{dw}{w} \\
    = {}& G^{(0)}(s) + \frac{1}{2\pi i} \oint_{\lvert w \rvert = r_n/2} (\transform \transJ^{-1}(\tR)(w) - 1) \sum^{\infty}_{\ell = 0} w^{-\ell} G^{(\ell)}(s) \frac{dw}{w}.
  \end{split}
\end{equation}
Suppose $C > 1$ is a constant, and we consider the value of $U(s)$ for $s = (\varrho n)^{-1} \zeta$ and $\zeta \in D(0, C) \setminus (D(0, 1) \cup \{ \arg z = \pm \frac{\theta^{-1} \pi + \gamma}{1 + \theta^{-1}} \})$. Then by \eqref{eq:est_G^ell_central}, $G^{(0)}(s) = e^{\zeta} G^{(0), \model}(\zeta) + \bigO(n^{1 - m_{\theta}})$, and by estimate of $G^{(\ell)}(s)$ in \eqref{eq:est_G^ell_central} and the estimate of $\tR(w)$ by Lemma \ref{lem:tR_est}, the contour integral in \eqref{eq:asy_U(s)_at_0} is $\bigO(n^{\frac{1 - m_{\theta}}{1 + m_{\theta}}})$.

The estimate above of $U(s)$ implies the estimate of $(Q_1(z), Q_2(z))$, and using the formulas \eqref{eq:defn_G^ell_model} and \eqref{eq:G_divided_model}, the estimate of $(Q_1(z), Q_2(z))$ is expressed by $I^{(1)}_{\theta, a}(z)$ and $I^{(2)}_{\theta, a}(z)$. Then with the estimates of $P^{(\infty)}_1(z), P^{(\infty)}_2(z)$ in \eqref{eq:P_and_P^pre} and the estimates of $g(z), \g(z)$ in \eqref{eq:g_error} and \eqref{eq:g_tilde_error}, we have the estimate of $Z_1(t)$ and $Z_2(t)$ in \eqref{eq:Y_1_in_Z_i}, and conclude that ($C_n$ is defined in \eqref{eq:result_p_n})
\begin{equation} \label{eq:y_n_in_annulus}
  p_n(z) = (-1)^n C_n \left( J_{\frac{\alpha + 1}{\theta}, \frac{1}{\theta}}(\theta (\rho n)^{1 + \theta^{-1}} z) + \bigO(n^{\frac{1 - m_{\theta}}{1 + m_{\theta}}}) \right), \quad \frac{1}{(\varrho n)^{1 + \theta^{-1}}} < \lvert z \rvert < \frac{C^{1 + \theta^{-1}}}{(\varrho n)^{1 + \theta^{-1}}}.
\end{equation}
Since $p_n(z)$ is a polynomial and $J_{\frac{\alpha + 1}{\theta}, \frac{1}{\theta}}(z)$ is an entire function, we have that the approximation above holds in the disk $\lvert z \rvert < \frac{C^{1 + \theta^{-1}}}{(\varrho n)^{1 + \theta^{-1}}}$. Hence we prove \eqref{eq:result_p_n}.

We prove \eqref{eq:result_q_n} in the same way. From the transform $\Y \to \Q$ given in \eqref{eq:Qtilde_from_Ytilde}, we have, analogous to \eqref{eq:Y_1_in_Z_i},
\begin{equation}
  q_n(z^{\theta}) = \Y_1(z) =
  \begin{cases}
    \Z_1(z), & \text{$z \in \halfH$ outside the lens}, \\
    \Z_1(z) + \Z_2(z), & \text{$z$ in the upper part of the lens}, \\
    \Z_1(z) - \Z_2(z), & \text{$z$ in the lower part of the lens},
  \end{cases}
  \end{equation}
where
\begin{equation}
  \Z_1(z) = \Q_1(z) \P^{(\infty)}_1(z) e^{n\g(z)}, \quad Z_2(z) = \Q_2(z) \P^{(\infty)}_2(z) z^{\alpha} e^{n(V(z) - g(z) + \ell)},
\end{equation}
with $g(z)$ and $\g(z)$ defined in \eqref{def:g} and \eqref{def:tildeg}, respectively. By tracing back the transformations $\Q \to \R^{\pre} \to \R \to \tRtilde$ given in \eqref{eq:defn_R^pre_tilde}, \eqref{eq:R_1_R_2_around_0_tilde} and \eqref{eq:tR_in_R_1_R_2_tilde}, we have that if we consider $\Q_1(z)$ as a function on the domain $(D^*(0, (r_n/2)^{1 + \theta^{-1}}) \cap \halfH) \setminus \Sigma$ and $\Q_2(z)$ as a function on the domain $D^*(0, (r_n/2)^{1 + \theta^{-1}}) \setminus \Sigma$, then they are expressed by $(\Q_2, \Q_1) = \transform^{-1}(\U)$, and, analogous to \eqref{eq:asy_U(s)_at_0}
\begin{equation} \label{eq:asy_U(s)_at_0_tilde}
  \begin{split}
    \U(s) = {}& (\P^{(0)}(\transform(\transJ^{-1}(\Inversion(\tRtilde)))))(s) \\
    = {}& \frac{1}{2\pi i} \oint_{\lvert w \rvert = r_n/2} \transform \transJ^{-1} \Inversion(\tR)(w) \sum^{\infty}_{\ell = 0} w^{-\ell} \G^{(\ell)}(s) \frac{dw}{w} \\
    = {}& \G^{(0)}(s) + \frac{1}{2\pi i} \oint_{\lvert w \rvert = r_n/2} (\transform \transJ^{-1} \Inversion(\tR)(w) - 1) \sum^{\infty}_{\ell = 0} w^{-\ell} G^{(\ell)}(s) \frac{dw}{w}.
  \end{split}
\end{equation}
Suppose $C > 1$ is a constant, and we consider the value of $\U(s)$ for $s = (\varrho n)^{-1} \zeta$ and $\zeta \in D(0, C) \setminus (D(0, 1) \cup \{ \arg z = \pm \frac{\theta^{-1} \pi - \gamma}{1 + \theta^{-1}} \})$. Then by \eqref{eq:est_G^ell_central_tilde}, $\G^{(0)}(s) = e^{\zeta} \G^{(0), \model}(\zeta) + \bigO(n^{1 - m_{\theta}})$, and by estimate of $\G^{(\ell)}(s)$ in \eqref{eq:est_G^ell_central_tilde} and the estimate of $\tRtilde(w)$ by Lemma \ref{lem:tRtilde_est}, the contour integral in \eqref{eq:asy_U(s)_at_0_tilde} is $\bigO(n^{\frac{1 - m_{\theta}}{1 + m_{\theta}}})$. Hence analogous to \eqref{eq:y_n_in_annulus}, we have the estimate of $q_n(z^{\theta})$, and express the result in terms of $q_n(z)$ as
\begin{equation}
  q_n(z) = (-1)^n \C_n \left( J_{\alpha + 1, \theta}(\theta^{\theta} (\rho n)^{\theta + 1} z) + \bigO(n^{\frac{1 - m_{\theta}}{1 + m_{\theta}}}) \right), \quad \frac{1}{(\varrho n)^{\theta + 1}} < \lvert z \rvert < \frac{C^{\theta + 1}}{(\varrho n)^{\theta + 1}}.
\end{equation}
This estimate extends to the disk $\lvert z \rvert < \frac{C^{\theta + 1}}{(\varrho n)^{\theta + 1}}$ by the analyticity of $q_n$ and $J_{\alpha + 1, \theta}$. Hence we prove \eqref{eq:result_q_n}.

\paragraph{Proof of Lemma \ref{lem:kappa}}

The proof of Lemma \ref{lem:kappa} is the same as the proof of \cite[Equation (5.31)]{Wang-Zhang21}. For completeness we give it here. Since $Y_2(z) = Cp_n(z)$ in \eqref{def:Y} has the limit \eqref{eq:Y_2_asy_infty} at $\infty$, we only need to find the limit of $Y_2(z) z^{-(n + 1)\theta}$ as $z \to \infty$ in $\halfH$. Like \eqref{eq:Y_1_in_Z_i}, we have that for $z \in \halfH$ and $\lvert z \rvert$ large enough, by \eqref{def:thirdtransform}, \eqref{eq:defn_R^pre} and \eqref{eq:R_1_R_2_around_0}, $Y_2(z) = P^{(\infty)}_2(z) R_2(z) e^{-n(\g(z) - \ell)}$ (see \cite[Equation (5.33)]{Wang-Zhang21}). From the definition formulas \eqref{def:g} of $g(z)$ and \eqref{eq:P1} of $P^{(\infty)}_2(z)$, we have the limit of $P^{(\infty)}_2(z) e^{-n\g(z)} z^{-n\theta}$ as $z \to \infty$ (see \cite[Equation (5.34)]{Wang-Zhang21}). Using relation \eqref{eq:tR_in_R_1_R_2} to express $R_2$ by $\tR$, and using the limit of $\tR(s)$ as $s \to 0$ in Part \ref{enu:lem:tR_est:2} of Lemma \ref{lem:tR_est}, we have that $R_2(z) = 1 + \bigO(n^{\frac{-m_{\theta}}{m_{\theta} + 1}})$ as $z \to \infty$. (Here our result is slightly stronger than \cite[Equation (5.35)]{Wang-Zhang21}, because the corresponding estimate of $\tR(s)$ as $s \to 0$ in \cite[Lemma 3.22]{Wang-Zhang21} is offhand.) Hence we derive \eqref{eq:limit_kappa_n}.

\paragraph{Sketch of proof of Theorem \ref{thm:kernel}}

The proof is identical the that of \cite[Theorem 1.3]{Wang-Zhang21} in \cite[Section 5.2]{Wang-Zhang21}. We outline the strategy here and refer to \cite[Section 5.2]{Wang-Zhang21} for details.

We assume, without loss of generality, $V(0) = 0$, and then due to the analyticity of $V$ and the assumption \eqref{eq:one-cut_regular},
\begin{equation} \label{eq:vanishing_V}
  V(x) = x^r + \bigO(x^{r + 1}), \qquad x\to 0,
\end{equation}
for some positive integer $r$. Then we define a family of functions $V_{\tau}$ indexed by a continuous parameter $\tau \in [0, 1]$ as follows:
\begin{equation}\label{def:vtau}
  V_{\tau}(x): =
  \begin{cases}
    \tau^{-1} V(\tau^{1/r} x), & \tau \in (0, 1], \\
    x^r, & \tau = 0.
  \end{cases}
\end{equation}
Clearly, $V_{\tau}(x)$ is continuous in both $x$ and $\tau$, and our assumption on the external field $V$ implies that Theorem \ref{thm:pqkappa} still holds with $V$ replaced by $V_{\tau}$. 

By Theorem \ref{thm:pqkappa} and Lemma \ref{lem:kappa}, we have, as $n\to\infty$ and uniformly for $x,y$ in compact subsets of $[0,\infty)$, (see \cite[Lemma 5.1]{Wang-Zhang21})

\begin{equation}\label{eq:k_n_asy}
  \frac{e^{-nV(x/(\rho n)^{1 + 1/\theta})}}{(\rho n)^{\alpha(1 + 1/\theta) + 1/\theta}} k^{(V)}_{n, j}\left( \frac{x}{(\rho n)^{1 + 1/\theta}}, \frac{y}{(\rho n)^{1 + 1/\theta}} \right) = \theta c^{-\frac{\theta}{\theta + 1}} k^{(\alpha, \theta)}(x, y) + \bigO (n^{\frac{1 - m_{\theta}}{1 + m_{\theta}}}),
\end{equation}
where $K^{(V)}_{n, j}$ is defined in \eqref{eq:sum_form_K} and
\begin{equation}\label{def:kalphatheta}
  k^{(\alpha, \theta)}(x, y) = \theta^{\alpha} J_{\frac{\alpha + 1}{\theta}, \frac{1}{\theta}}(\theta x) J_{\alpha + 1, \theta}((\theta y)^{\theta}).
\end{equation}
When $\theta$ is an integer, $k^{(\alpha, \theta)}(x, y)$ in \eqref{def:kalphatheta} agrees with \cite[Equation (5.29)]{Wang-Zhang21}. Identity \eqref{eq:k_n_asy} holds for all $V$ satisfying the condition in Theorem \ref{thm:pqkappa} and it is straightforward to check that the error term in \eqref{eq:k_n_asy} is uniform for all $V_{\tau}$ ($\tau \in (0, 1]$) in place of $V$. Hence we have the following estimate (see \cite[Lemma 5.2]{Wang-Zhang21}):
\begin{lemma} \label{prop:generalized}
  Suppose $V$ satisfies \eqref{eq:loggrowth}, \eqref{eq:one-cut_regular} and \eqref{eq:vanishing_V}. With $V_{\tau}$ defined in \eqref{def:vtau}, we have, for any $M, \epsilon > 0$, there exists a positive integer $N_{M, \epsilon}$ such that if $n > N_{M, \epsilon}$, then
\begin{multline}
  \left\lvert n^{-\alpha(1 + \frac{1}{\theta}) - \frac{1}{\theta}} k^{(V_{\tau})}_{n, n} \left( \frac{x}{n^{1 + 1/\theta}}, \frac{y}{n^{1 + 1/\theta}} \right) \right. \\
  - \left. \theta (c^{(V_{\tau})})^{-\frac{\theta}{\theta + 1}} (\rho^{(V_{\tau})})^{\alpha(1 + \frac{1}{\theta}) + \frac{1}{\theta}} k^{(\alpha, \theta)}((\rho^{(V_{\tau})})^{1 + \frac{1}{\theta}} x, (\rho^{(V_{\tau})})^{1 + \frac{1}{\theta}} y) \right\rvert < \epsilon,
\end{multline}
uniformly for all $\tau \in [0, 1]$ and $x, y \in [0, M]$, where $k^{(V_{\tau})}_{n, n}(x,y)$ and $k^{(\alpha, \theta)}(x,y)$ are defined in \eqref{eq:sum_form_K} and \eqref{def:kalphatheta}, respectively.
\end{lemma}

The strategy of proof now is to split the summation in the correlation kernel $K_n$ into two parts, such that one part sums from $n = 0$ to $n = N$ with $N$ a large constant, and the other part is the remains, and estimate each part separately. For the first part, we have that uniformly for $x, y \in [0, M]$,
\begin{multline} \label{eq:est_k_sum_remainder}
  \lim_{n \to \infty} \frac{1}{n^{(\alpha + 1)(1 + \frac{1}{\theta})}} \sum^N_{k = 0} k^{(V)}_{n, k}\left( \frac{x}{n^{1 + 1/\theta}}, \frac{y}{n^{1 + 1/\theta}} \right) \\
  = \lim_{n \to \infty} \frac{1}{n^{(\alpha + 1)(1 + \frac{1}{\theta})}} \sum^N_{k = 0} k^{(V_0)}_{n, k}\left( \frac{x}{n^{1 + 1/\theta}}, \frac{y}{n^{1 + 1/\theta}} \right) = 0.
\end{multline}
For the second part, we have
\begin{multline}
  n^{-(\alpha + 1)(1 + \frac{1}{\theta}) +1 } \sum^n_{j = N + 1} k^{(V)}_{n, j} \left( \frac{x}{n^{1 + 1/\theta}}, \frac{y}{n^{1 + 1/\theta}} \right) =
 \\
  \sum^n_{j = N + 1} \left( \frac{j}{n} \right)^{(\alpha + 1)(\frac{1}{\theta} - \frac{1}{r}) + \alpha} j^{-\alpha(1 + \frac{1}{\theta}) - \frac{1}{\theta}} k^{(V_{j/n})}_{j, j} \left( \left( \frac{j}{n} \right)^{1 + \frac{1}{\theta} - \frac{1}{r}} \frac{x}{j^{1 + 1/\theta}}, \left( \frac{j}{n} \right)^{1 + \frac{1}{\theta} - \frac{1}{r}} \frac{y}{j^{1 + 1/\theta}} \right).
\end{multline}
Given $M, \epsilon > 0$, if we further take $N > N_{M, \epsilon}$ in the above formula with $N_{M, \epsilon}$ given in Lemma \ref{prop:generalized}, it then follows that
\begin{multline} \label{eq:est_K_first}
 \left\lvert \frac{1}{n^{(\alpha + 1)(1 + 1/\theta)}} \sum^n_{j = N + 1} k^{(V)}_{n, j} \left( \frac{x}{n^{1 + 1/\theta}}, \frac{y}{n^{1 + 1/\theta}} \right) - \frac{1}{n} \sum^n_{j = N + 1} f \left( \frac{j}{n} \right) \right\rvert
\\
\leq \frac{\epsilon }{n} \sum^n_{j = N + 1} \left( \frac{j}{n} \right)^{(\alpha + 1)(\frac{1}{\theta} - \frac{1}{r}) + \alpha}
< \epsilon',
\end{multline}
for $x, y \in [0, M]$, where
\begin{multline}
f(t) = t^{(\alpha + 1)(\frac{1}{\theta} - \frac{1}{r}) + \alpha} \theta (c^{(V_{t})})^{-\frac{\theta}{\theta + 1}} (\rho^{(V_{t})})^{\alpha(1 + \frac{1}{\theta}) + \frac{1}{\theta}}
\\
\times k^{(\alpha, \theta)} \left( t^{1 + \frac{1}{\theta} - \frac{1}{r}} (\rho^{(V_{t})})^{1 + \frac{1}{\theta}} x, t^{1 + \frac{1}{\theta} - \frac{1}{r}} (\rho^{(V_{t})})^{1 + \frac{1}{\theta}} y \right)
\end{multline}
and $\epsilon' = \epsilon/[(\alpha + 1)(1 + \frac{1}{\theta} - \frac{1}{r})]$. Since $k^{(\alpha, \theta)}(x, y)$ is a continuous function in $x, y \in [0, \infty)$, and $c^{(V_t)}$ and $\rho^{(V_t)}$ are continuous functions in $t$ with values in a compact subset of $(0, \infty)$ as $t \in [0, 1]$, we have that $f(t)$ is continuous for $t \in (0, 1]$ with $f(t) = \bigO(t^{(\alpha + 1)(\frac{1}{\theta} - \frac{1}{r}) + \alpha})$ as $t \to 0_+$. We note that for all $\alpha > -1$ and $r \geq 1$, $(\alpha + 1)(\theta^{-1} - r^{-1}) + \alpha > -1$. Thus, we observe that the summation involving $f(j/n)$ in \eqref{eq:est_K_first} is a Riemann sum of a definite integral as $n\to \infty$, that is,
\begin{equation} \label{eq:est_k_sum_main}
  \lim_{n \to \infty} \frac{1}{n} \sum^n_{j = N + 1} f \left( \frac{j}{n} \right) = \int^1_0 f(t)dt.
\end{equation}
Hence, we only need to show the equality
\begin{equation} \label{eq:integral_for_kernel}
  \int^1_0 f(t) dt = \theta^2 \int^{(\rho^{(V)})^{1 + 1/\theta}}_0 u^{\alpha} k^{(\alpha, \theta)}(ux, uy) du,
\end{equation}
to conclude Theorem \ref{thm:kernel}. The proof of \eqref{eq:integral_for_kernel} is technical and involves analysis of the equilibrium measure. Since the proof \eqref{eq:integral_for_kernel} is exactly the same as the proof of \cite[Equation (5.46)]{Wang-Zhang21} that is in \cite[Section 5.2]{Wang-Zhang21}, we omit it here. (The proof there does not require $\theta$ to be an integer or a rational, and is valid for all real $\theta > 0$.)

% \bibliographystyle{plain}
% \bibliography{../bibliography/bibliography.bib}

\end{document}